\documentclass[11pt,reqno]{amsart}

\usepackage[utf8]{inputenc}
\usepackage[T1]{fontenc}
\usepackage{lmodern}
\usepackage{microtype}         % Better typography

\usepackage[letterpaper,margin=1in]{geometry}

\usepackage{setspace}

\usepackage{xcolor}

\usepackage{amsmath, amssymb, amsthm}  % Core math tools
\usepackage{mathtools}        % Enhancements to amsmath
\usepackage{mathrsfs}         % For \mathscr
\usepackage{stmaryrd}         % Special brackets and symbols
\usepackage{bm}               % Bold math symbols
\usepackage{enumitem}         % Customizable lists

\usepackage{graphicx}         % Include graphics
\usepackage{tikz}             % Drawing diagrams
\usetikzlibrary{arrows.meta, calc, intersections, patterns, positioning}

\usepackage{algorithm}
\usepackage{algpseudocode}
 
\usepackage{aliascnt}
\usepackage[colorlinks=true,linkcolor=blue,citecolor=blue,urlcolor=blue]{hyperref}
\usepackage[capitalize, nameinlink]{cleveref} 

\usepackage{appendix}
\crefname{appsec}{Appendix}{Appendices}

\theoremstyle{plain}

\newtheorem{theorem}{Theorem}[section]
\crefname{theorem}{theorem}{theorems}
\Crefname{theorem}{Theorem}{Theorems}

\newaliascnt{lemma}{theorem}
\newtheorem{lemma}[lemma]{Lemma}
\aliascntresetthe{lemma}
\crefname{lemma}{lemma}{lemmas}
\Crefname{lemma}{Lemma}{Lemmas}

\newaliascnt{corollary}{theorem}
\newtheorem{corollary}[corollary]{Corollary}
\aliascntresetthe{corollary}
\crefname{corollary}{corollary}{corollaries}
\Crefname{corollary}{Corollary}{Corollaries}

\newaliascnt{proposition}{theorem}
\newtheorem{proposition}[proposition]{Proposition}
\aliascntresetthe{proposition}
\crefname{proposition}{proposition}{propositions}
\Crefname{proposition}{Proposition}{Propositions}

\theoremstyle{definition}

\newaliascnt{definition}{theorem}
\newtheorem{definition}[definition]{Definition}
\aliascntresetthe{definition}
\crefname{definition}{definition}{definitions}
\Crefname{definition}{Definition}{Definitions}

\theoremstyle{remark}

\newaliascnt{remark}{theorem}
\newtheorem{remark}[remark]{Remark}
\aliascntresetthe{remark}
\crefname{remark}{remark}{remarks}
\Crefname{remark}{Remark}{Remarks}

\newaliascnt{example}{theorem}

\aliascntresetthe{example}
\crefname{example}{example}{examples}
\Crefname{example}{Example}{Examples}

\newaliascnt{claim}{theorem}

\aliascntresetthe{claim}
\crefname{claim}{claim}{claims}
\Crefname{claim}{Claim}{Claims}

\theoremstyle{definition}
\newaliascnt{assumption}{theorem}
\newtheorem{assumption}[assumption]{Assumption}
\aliascntresetthe{assumption}
\crefname{assumption}{assumption}{assumptions}
\Crefname{assumption}{Assumption}{Assumptions}

\newcommand{\set}[1]{\left\{#1\right\}}

\newcommand{\floor}[1]{\left\lfloor#1\right\rfloor}

\DeclareMathOperator{\poly}{poly}

\newcommand{\transpose}{\mathsf T}

\newcommand{\eps}{\varepsilon}

\newcommand{\vect}[1]{\boldsymbol{#1}}

\newcommand{\neighbor}[1]{N(#1)}
\newcommand{\neighborCl}[1]{N^{+}(#1)}
\newcommand{\size}[1]{|#1|}
\DeclareMathOperator{\E}{\mathbb{E}}
\DeclareMathOperator{\Var}{Var}

\DeclareMathOperator*{\argmin}{arg\,min}

\newcommand{\Cov}{\operatorname{Cov}}
\DeclareMathOperator*{\argmax}{arg\,max}

\newcommand{\ignore}[1]{}

\newcommand{\Pcol}{P_{\mathrm{col}}}

\newcommand{\OmegaG}{\Omega_G}
\newcommand{\Pdist}{\mathcal P}
\newcommand{\Ppos}{\mathcal P_+}

\newcommand{\lambdac}{\lambda_c(\Delta)}
\newcommand{\alphac}{\alpha_c(\Delta)}

\newcommand{\dHam}{d_{\mathrm H}}
\newcommand{\TV}{\mathrm{TV}}

\newcommand{\AG}{A_G}

\newcommand{\transK}{\mathcal K}

\newcommand{\Kmf}{\mathcal K_{\mathrm{mf}}}
\newcommand{\Kss}{\mathcal K_{\mathrm{ss}}}

\newcommand{\kernelQ}{\mathcal Q}
\newcommand{\Qmf}{\mathcal Q^{\mathrm{mf}}}
\newcommand{\Qss}{\mathcal Q^{\mathrm{ss}}}

\newcommand{\Tmf}{T_{\mathrm{mf}}}
\newcommand{\Tss}{T_{\mathrm{ss}}}

\newcommand{\given}{\,|\,}

\newcommand{\OmegaNa}{\Omega^{\mathrm{mf}}_{N,\a}}
\newcommand{\OmegaNm}{\OmegaNa} 
\newcommand{\OmegaNr}{\Omega^{\mathrm{ss}}_{N,\vec r}} 

\newcommand{\muNm}{\mu^{\mathrm{mf}}_{N,\alpha}} 
\newcommand{\muNr}{\mu^{\mathrm{ss}}_{N,\vec r}}

\newcommand{\PmfN}{P^{\mathrm{mf}}_{N}}
\newcommand{\PmfNM}{\PmfN}

\newcommand{\PssNr}{P^{\mathrm{ss}}_{N}}

\newcommand{\cA}{\ensuremath{\mathcal A}}

\newcommand{\cD}{\ensuremath{\mathcal D}} 
\newcommand{\cE}{\ensuremath{\mathcal E}}

\newcommand{\cI}{\ensuremath{\mathcal I}} 
 
\newcommand{\cK}{\ensuremath{\mathcal K}}

\newcommand{\cP}{\ensuremath{\mathcal P}} 
\newcommand{\cQ}{\ensuremath{\mathcal Q}} 
 
\newcommand{\cS}{\ensuremath{\mathcal S}}

\newcommand{\bbE}{{\ensuremath{\mathbb E}} }

\newcommand{\bbN}{{\ensuremath{\mathbb N}} } 
 
\newcommand{\bbP}{{\ensuremath{\mathbb P}} } 
 
\newcommand{\bbR}{{\ensuremath{\mathbb R}} }

\DeclareMathOperator{\Ent}{Ent}

\let\a=\alpha \let\b=\beta   \let\d=\delta  \let\e=\varepsilon
 \let\g=\gamma       \let\l=\lambda
\let\m=\mu   \let\n=\nu         
  \let\s=\sigma \let\t=\tau   
  
\let\D=\Delta   \let\G=\Gamma   
\let\O=\Omega

\title{Nonlinear Exchange Dynamics for Independent Sets}

\author[M.~Abbaszadeh Minab]{Mehrad Abbaszadeh Minab}
\address{School of Computer Science, Georgia Institute of Technology, Atlanta, GA, USA}
\email{mminab3@gatech.edu}

\author[P.~Caputo]{Pietro Caputo}
\address{
Department of Mathematics and Physics, Roma Tre University, Largo San Murialdo 1, 00146 Roma, Italy.}
\email{pietro.caputo@uniroma3.it}
\thanks{PC was supported in part by the Italian Ministry of Foreign Affairs and International Cooperation, grant number BR26GR05}
\author[Z.~Chen]{Zongchen Chen}
\address{School of Computer Science, Georgia Institute of Technology, Atlanta, GA, USA}
\email{chenzongchen@gatech.edu}

\author[M.~Morellini]{Mario Morellini}
\address{
Department of Mathematics and Physics, Roma Tre University, Largo San Murialdo 1, 00146 Roma, Italy}
\email{mario.morellini@uniroma3.it}

\author[A.~Sinclair]{Alistair Sinclair}
\address{Computer Science Division, University of California Berkeley, Berkeley, CA, USA}
\email{sinclair@cs.berkeley.edu}
\thanks{AS was supported in part by NSF grant CCF-223109.}

\date{}

\begin{document}

\begin{abstract}
In recent years, nonlinear dynamics derived from kinetic theory have gained attention in the context of sampling configurations of spin systems such as the Ising model. In this paper, we focus on nonlinear dynamics for the \emph{hard-core model}, a canonical spin system with hard constraints that specifies a distribution over independent sets in a graph, weighted by their sizes. We explore two distinct types of nonlinear dynamics: the \emph{mean-field} dynamics, which preserves the \emph{density} (or average size) of independent sets, and the \emph{single-site} dynamics, which preserves the \emph{marginal vector} (i.e., the occupancy probabilities of the vertices). These dynamics are natural stochastic processes for sampling from the hard-core model with a specified density or marginal vector, respectively, both of which are canonical instances of maximum entropy distributions that have been studied in various contexts.

In contrast to linear Markov chains, there is a significant lack of a fundamental theoretical framework for nonlinear dynamics. In this paper, we develop foundational theoretical tools for analyzing nonlinear dynamics within the context of the hard-core model. We establish almost linear convergence of both the mean-field and single-site dynamics at sufficiently low density through novel coupling arguments. We also establish exponential decay of relative entropy for the mean-field dynamics all the way up to the critical density. Additionally, we design new algorithms for sampling from the hard-core distribution with either a specified density or a specified marginal vector. These algorithms are based on a related \emph{linear} Markov chain, called the particle-system dynamics and inspired by the so-called Kac's program, that approximates the associated nonlinear dynamics.  As we demonstrate in the paper, they are comparable in time complexity, but simpler to implement, than traditional approaches based on learning parameter values.

\par\bigskip\noindent
{\bf AI Statement:} No AI tools were used in the derivation or presentation of the results reported in this paper.

\end{abstract}

\maketitle

\thispagestyle{empty}
\newpage

\tableofcontents
\thispagestyle{empty}

\newpage
\setcounter{page}{1}

\section{Introduction}
Sampling from discrete high-dimensional distributions is a fundamental problem in computer science. Many important applications, such as statistical estimation, causal inference, and experimental simulation, rely on massive samples drawn independently and approximately from a complex target distribution. 
Developing fast samplers for general distributions is thus a crucial task.

%We examine distributions on the discrete Boolean hypercube $\{0,1\}^n$. 
In this paper we focus on a fundamental discrete distribution that has been widely studied, namely the \emph{hard-core model}. Let $G = (V,E)$ be a graph and $\lambda > 0$ be a parameter called the \emph{fugacity}. The Gibbs distribution for the hard-core model is supported on the family of independent sets of $G$, denoted by $\OmegaG$, where each independent set $\sigma \in \{0,1\}^V$ (represented as a Boolean indicator vector) is assigned a probability mass
\begin{align*}
    \mu_{G,\lambda}(\sigma) = \frac{\lambda^{|\sigma|}}{Z_G(\lambda)},
\end{align*}
where $|\sigma| = \sum_{v \in V} \sigma_v$ is the size of the independent set and $Z_G(\lambda) = \sum_{\sigma \in \OmegaG} \lambda^{|\sigma|}$ is the normalizing constant for the distribution, known as the \emph{partition function}.

One of the most popular methods for generating samples approximately from a target distribution is the Markov Chain Monte Carlo (MCMC) method. This involves simulating a suitable Markov chain for a sufficiently long period and outputting the final state. Among the various MCMC methods, \emph{Glauber dynamics}, also known as the \emph{Gibbs sampler}, is a standard approach widely used in many applications. 
At each step, this dynamics picks a random coordinate and updates it according to the conditional marginal distribution with all other coordinates fixed. 

The Glauber dynamics for the hard-core model has a simple description. If the current state is $\sigma \in \OmegaG$, we pick a vertex $v \in V$ uniformly at random and let $\sigma' := \sigma \cup \{v\}$. If $\sigma' \notin \OmegaG$, then we set the new state $\tau = \sigma$; otherwise, we set $\tau = \sigma'$ with probability $\frac{\lambda}{1+\lambda}$, and $\tau = \sigma' \setminus \{v\}$ with the remaining probability $\frac{1}{1+\lambda}$. It is easy to verify from elementary Markov chain theory that this process is ergodic (i.e., irreducible and aperiodic) and converges to the Gibbs distribution $\mu_{G,\lambda}$ as the number of updates tends to infinity. Its utility as an algorithm for sampling from~$\mu_{G,\lambda}$ thus depends on its rate of convergence.

Let $\Delta$ be the maximum degree of the graph. Recent research has determined the exact threshold $\lambdac = \frac{(\Delta-1)^{\Delta-1}}{(\Delta-2)^\Delta}$ for the phase transition between fast and slow convergence of the Glauber dynamics. Specifically, the Glauber dynamics converges rapidly whenever $\lambda \le \lambdac$ \cite{ALO21,CFYZ22,CE25,CCYZ25}, and can be exponentially slow when $\lambda > \lambdac$ \cite{MWW09}.\footnote{This computational phase transition actually coincides exactly with a physical phase transition on the $\D$-regular tree, defined by the onset of long-range correlations between vertices in the independent set; see, e.g., \cite{Weitz06}.} Here, the speed of convergence is characterized by the \emph{mixing time}, which is the time until the distribution of the chain is close to $\mu_{G,\lambda}$ in total variation distance. Thus, Glauber dynamics provides an efficient sampling algorithm for the hard-core model at all fugacities $\lambda \le \lambdac$.

While fugacity is the classical parameterization of the hard-core model in statistical physics, many modern applications actually call for alternative parameterizations in terms of the density or occupancy probabilities of vertices in the independent set.
The \emph{density} of the hard-core model, defined as
\begin{align*}
    \alpha = \alpha(\mu_{G,\lambda}) := \frac{1}{n} \E_{\mu_{G,\lambda}}[|\sigma|],
\end{align*}
describes the average (normalized) size of an independent set drawn from the hard-core model.\footnote{Throughout, we write $n=|V|$ for the number of vertices in~$G$.} The map
$\lambda \mapsto \alpha(\mu_{G, \lambda})$ is strictly increasing and gives a one-to-one correspondence between the fugacity $\lambda \in (0,\infty)$ and the density $\alpha \in (0,i_{\mathrm{max}}/n)$, where $i_{\mathrm{max}}$ is the size of a largest independent set; see \Cref{app:bisection-prescribed-density}. 
Thus, we can alternatively parameterize the hard-core model using $\alpha$ and represent its Gibbs distribution as $\mu_G^{\alpha}$. 
This approach is applicable when we want to sample at a specific density but do not have access to the corresponding fugacity.

We can extend the density to a finer parameterization in terms of the vector of marginal occupation probabilities of the vertices.
To introduce this, we first consider a generalized version of the hard-core model in which each vertex has a distinct fugacity. 
Given a vector of positive fugacities
$\vect\lambda=(\lambda_v)_{v\in V}\in(0,\infty)^V$, we define the Gibbs distribution as
\[
\mu_{G,\vect\lambda}(\sigma)
=
\frac{1}{Z_G(\vect\lambda)}
\prod_{v:\sigma_v=1}\lambda_v,
\qquad \text{where}\quad
Z_G(\vect\lambda)
:=
\sum_{\sigma\in\OmegaG}
\prod_{v:\sigma_v=1}\lambda_v .
\]
If $\lambda_v=\lambda$ for all $v\in V$, we recover the hard-core model $\mu_{G,\lambda}$ with uniform fugacity~$\lambda$.
We now define the marginal vector $\vect m =  (m_v)_{v\in V} \in (0,1)^V$ as
\begin{align*}
    m_v = \Pr_{\sigma\sim\mu_{G,\lambda}}(\sigma_v = 1).
\end{align*}
Analogously to the bijection between fugacity $\lambda$ and density $\alpha$ mentioned above, there is also a one-to-one correspondence between the fugacity vector $\vect{\lambda} \in (0,\infty)^V$ and the marginal vector $\vect{m} \in \mathcal{P}_{\mathrm{IS}}^{\circ}$; see \Cref{app:gradient-descent-prescribed-marginals}. Here, $\mathcal{P}_{\mathrm{IS}}$ denotes the independent set polytope (i.e., the convex hull of vectors in $\Omega_G$), and $\mathcal{P}_{\mathrm{IS}}^\circ$ its interior. Therefore, the hard-core model can alternatively be parameterized using the marginal vector $\vect{m}$, with associated Gibbs distribution denoted by $\mu_G^{\vect m}$. 

In this paper, we study the sampling problem when we are given either the density $\alpha$ or the marginal vector $\vect m$, while the corresponding fugacity~$\lambda$ or, respectively, fugacity vector $\vect\lambda$ are unknown to us.  
The Gibbs distribution $\mu_G^{\alpha}$ or $\mu_G^{\vect m}$ is known as the \emph{maximum entropy distribution}: it is the unique distribution supported on $\OmegaG$ that achieves maximum entropy, when either the density or the marginal vector is fixed~\cite{Jaynes1, Jaynes2}.
Maximum entropy distributions have been extensively studied for their applications in exponential families, machine learning, and combinatorial optimization: see, e.g., \cite{Kapur,WainwrightJordan2008,O-GSS,SV14,AGMOS17,SV19,CKV20,MSP22} for representative examples. 

As a concrete toy example, consider a social network in which edges between two people represent mutual friendships. We want to select a committee from a distribution that is ``as random as possible'' (to prevent predictability or manipulation), subject to the constraints that no two selected individuals are friends and that each individual $v$ is chosen with a given probability~$m_v$ (e.g, one important scenario is to set $m_v = c$ uniformly across individuals, thereby ensuring fairness of the selection). To achieve this, we should sample from the maximum entropy distribution over independent sets with marginal vector~$\vect m$, which corresponds to the Gibbs distribution~$\mu_G^{\vect m}$.  

It is important to note that, without access to the fugacity parameter(s) $\lambda$ or $\vect\lambda$, we cannot implement the Glauber dynamics for the hard-core model specified by a density $\alpha$ or a marginal vector $\vect m$. Indeed, the classical approach to this problem is to first estimate the fugacity parameter(s) given the density or marginals, and then use Glauber dynamics. This approach is known as solving the ``inverse problem'', and is well studied; see~\cite{NZB17} for a survey and \cite{SV14,BGS14,Montanari15} for related results.
Solving the inverse problem is in general a complex learning task relying on methods such as bisection or gradient descent, and involves tuning multiple parameters and the use of sampling or inference as subroutines.
In the appendix, we sketch these approaches in the context of the hard-core model and compare them with our new algorithms.

This discussion motivates the central question we aim to address in this paper, namely:
\begin{itemize}  
    \item \emph{Is there an efficient and easy-to-implement stochastic process that samples directly from the hard-core model with a given density or marginal vector?}  
\end{itemize}

\subsection{Nonlinear exchange dynamics}
To address this question we consider {\it nonlinear exchange dynamics},
a natural nonlinear version of Glauber dynamics that possesses the additional feature of supporting conserved quantities (in our case, the density~$\a$ or marginal probabilities~$\vect m$).  These dynamics can be seen as a combinatorial version of the Boltzmann equation from kinetic theory~\cite{cercignani1988boltzmann,villani08} and a generalization of the Hardy-Weinberg model of population genetics~\cite{Hardy, Weinberg}.  They were first introduced in the theoretical CS community in~\cite{RSW92}, and have been explored in the context of population genetics~\cite{RRS98,CS18} and most recently the Ising
model~\cite{caputo2024nonlineardynamicsisingmodel,caputo2025kac}.\footnote{Those papers obtain an analogous reparameterization of the Ising model in terms of the magnetization vector, though the emphasis was less algorithmic than in the current paper.}
Computationally, they are inspired by genetic algorithms (see, e.g., \cite{Goldberg,Mitchell}). 

In a Markov chain, we use a transition matrix, or kernel,~$P$ to describe the transition from one state to another. Specifically, $P(\tau \given \sigma)$ represents the probability of moving from state $\sigma$ to state $\tau$ in a single step.  This gives rise to a linear evolution on probability distributions over the states: if the current distribution over states is~$p$, then the distribution after one step of the Markov chain is given by the map $p\mapsto T(p)$, where $[T(p)](\tau)
=\sum_\sigma p(\sigma)P(\tau \given \sigma)$.

In our nonlinear dynamics, the transition involves an interaction, or ``collision'', between a {\it pair\/} of states, $(\sigma, \sigma')$, leading to a new state~$\tau$. We specify the probability of transitioning from the pair $(\sigma, \sigma')$ to~$\tau$ via the transition kernel $Q(\tau \given \sigma, \sigma')$.  The overall evolution is governed by the so-called {\it mass-action principle}, under which the probability that the pair $(\sigma, \sigma')$ interact in the distribution~$p$ is the product $p(\sigma)p(\sigma')$; in other words, we can think of $\sigma$ and $\sigma'$ drawn independently from the current distribution $p$.  Thus, the one-step evolution is specified by the map
\begin{equation}\label{eqn:massaction}
   [T(p)](\tau) = \sum_{\sigma,\sigma'} p(\sigma)p(\sigma')Q(\tau \given \sigma, \sigma').
\end{equation}   
This evolution is clearly nonlinear (actually quadratic).

To apply nonlinear dynamics to the hard-core model, we need to specify the associated transition kernel~$Q$.  We give two similar but different kernels. The first, which we call the \emph{mean-field} exchange kernel, is suited to the density parameterization, while the second, called the \emph{single-site} exchange kernel,  is suited to the marginal vector parameterization.  Both of these kernels take a pair of independent sets \( (\sigma, \sigma') \in \OmegaG \times \OmegaG \) and produce a pair \( (\tau, \tau')\) of independent sets by attempting to exchange the bits \(\sigma_v\) and \(\sigma'_u\) for two randomly selected vertices \(u\) and \(v\). The exchange is accepted if it results in two independent sets, otherwise it is rejected.  Thus, the collision $(\sigma,\sigma')\mapsto (\tau,\tau')$ changes the initial pair $(\sigma,\sigma')$ only when the chosen bits satisfy $\sigma_v \neq \sigma'_u$,  and the attempted exchange still produces two independent sets.  The output of the kernel is then the first element~$\tau$ of the resulting pair; notice that $\tau$ differs at at most one vertex from $\sigma$, but can be very different from $\sigma'$. We refer to \Cref{sec:prelim} for precise definitions, and \Cref{fig:mean-field-exchange} for an illustration.

\begin{figure}[t]	
\centering
\includegraphics[scale=0.11]{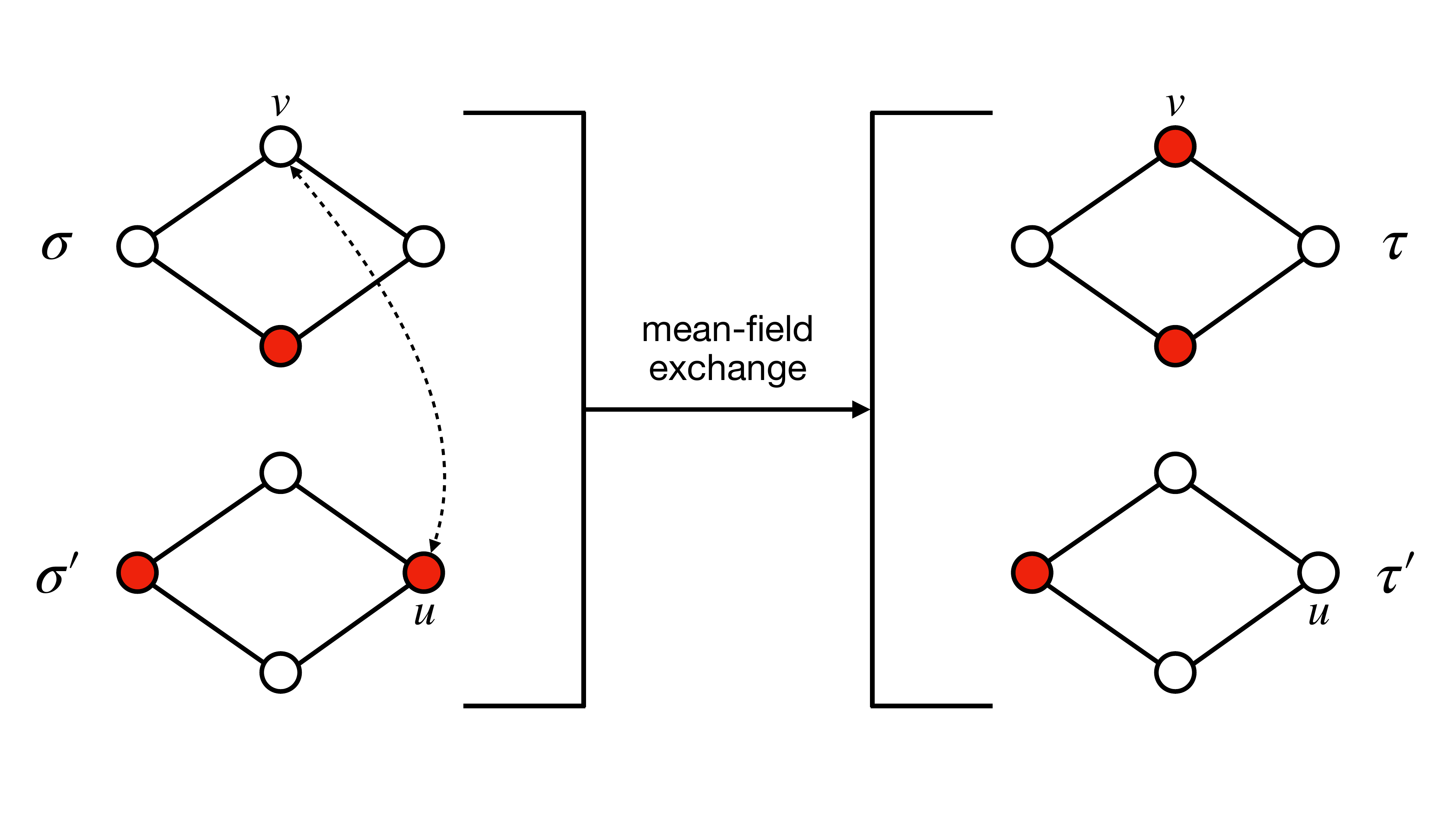}
\qquad
\includegraphics[scale=0.11]{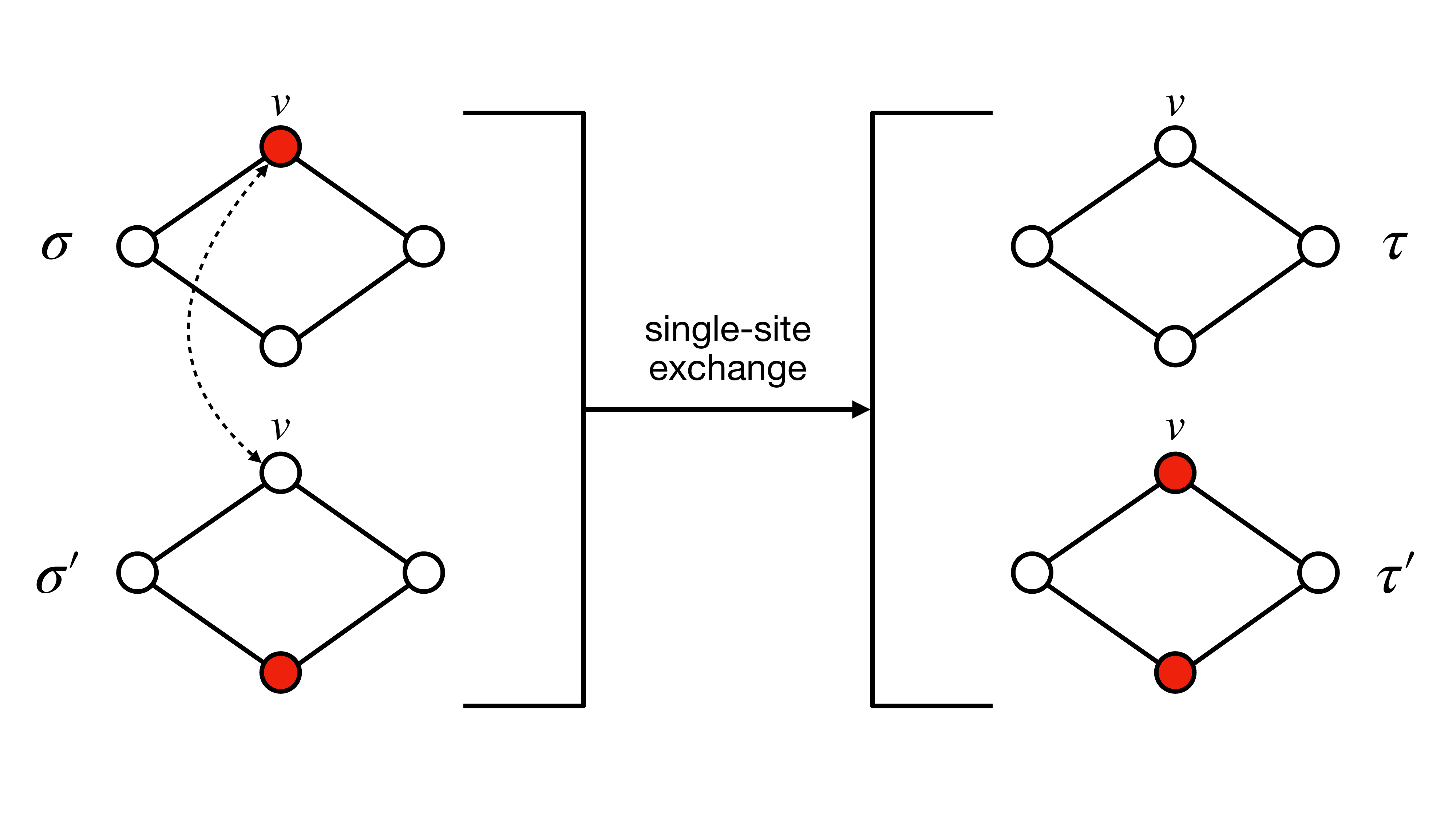}
\caption{An illustration of the mean-field (left) and single-site (right) exchange kernels.}
\label{fig:mean-field-exchange}
\end{figure}

In the mean-field exchange kernel, we choose the two vertices $u,v$ independently and uniformly at random. In the single-site exchange kernel, we pick the {\it same\/} vertex $u=v$ uniformly at random. We refer to the associated nonlinear dynamics (defined as in~\eqref{eqn:massaction}) as the 
\emph{mean-field dynamics} 
and the \emph{single-site dynamics}, respectively, and denote
their respective one-step maps by $\Tmf$ and $\Tss$.

We can now explain why nonlinear dynamics are well suited to the sampling problem we want to solve.  First, updates in nonlinear dynamics do not require knowledge of fugacities.  Second, as can be easily verified from the definitions of the relevant collision kernels, the mean-field dynamics {\it preserves the density}, i.e., \(\alpha(\Tmf(p)) = \alpha(p)\) for all distributions~$p$ over~$\OmegaG$, while the single-site dynamics {\it preserves the marginal vector}, i.e., \(\vect m(\Tss(p)) = \vect m(p)\).  
Since it can be shown that both dynamics converge to a Gibbs distribution (which is a non-trivial result), when initialized with the desired density or marginal vector, they both converge to the desired maximum entropy distributions.

While nonlinear dynamics in principle provide a simple and elegant solution to our sampling problem, a number of substantial technical challenges remain if we are to turn this idea into an efficient algorithm.
First, in contrast to linear Markov chains, which have a well-established theoretical foundation and extensive literature, research on nonlinear dynamics is very limited. 
For instance, we have a comprehensive understanding of the convergence of Markov chains, under simple conditions, at an exponential rate. However, for nonlinear Markov chains, there are no general conditions that ensure convergence, at \emph{any} rate, to a unique stationary distribution.\footnote{Notable exceptions for general nonlinear Markov chains are~\cite{AJDM} and~\cite{Butkovsky}, but the conditions for convergence that these works require are far too restrictive to apply in our setting.  A famous open problem known as the ``Global Attractor Conjecture''~\cite{Feinberg} postulates the convergence of any nonlinear dynamics under fairly mild conditions.} 
Additionally, many techniques, such as coupling and functional inequalities, that have been successfully applied to the analysis 
of the speed of convergence, or mixing time, of Markov chains do not have immediate analogs in the nonlinear case.

Furthermore, while nonlinear dynamics may be straightforward to describe, their efficient implementation is not trivial. As is clear from~\eqref{eqn:massaction}, each iteration of the dynamics requires two independent samples from the current state. Consequently, to obtain a sample after \( t \) iterations, \( 2^t \) independent samples from the initial distribution are needed, which is computationally prohibitive. In contrast, a linear Markov chain can be simulated by following the trajectory of a single sample.
%\Mario{It could be worth noting somewhere that the block dynamics achieves logarithmic mixing time (see \cite{caputo2024nonlineardynamicsisingmodel}). This would address the difficulty of sampling directly via the nonlinear dynamics: if $T_{\mathrm{mix}} \approx \log n$, then $2^{T_{\mathrm{mix}}} \approx \mathrm{poly}(n)$. The point is that the obstruction arises from updating one bit (single-site exchange) at a time, a block update should circumvent it. To be checked.}

%In this paper, we will address these issues in the context of the hard-core model, and demonstrate the benefits of nonlinear dynamics for sampling for a specified density or marginal vector.  
In this paper, we will address these issues in the context of the hard-core model, and show how ideas from nonlinear dynamics can be leveraged for sampling with a specified density or marginal vector.
We first establish some general structural properties 
of the nonlinear dynamics, concerning invariant states and qualitative convergence to stationarity. We then address {\it quantitative\/} bounds on the speed of convergence of these evolutions under suitable assumptions. These results are a significant improvement and extension of recent analyses of nonlinear dynamics~\cite{caputo2024nonlineardynamicsisingmodel, caputo2025kac}, which apply only to systems without hard % with soft 
constraints (such as the Ising model),
%\footnote{The hard-core model has a hard constraint because in an independent set all neighbors of an occupied vertex must be unoccupied.}, 
 require more severe restrictions on the parameters, and use technically more involved techniques.  We believe that our new methods are of independent interest.  We then move on to the algorithmic applications. These will be based on the fruitful idea, dating back to Mark Kac~\cite{kac56} and often referred in the kinetic theory literature as \emph{propagation of chaos}, that the nonlinear evolution is well approximated by a suitable {\it linear\/} Markov chain, which we refer to as the \emph{particle system dynamics}. In addition to providing an algorithmization of the nonlinear dynamics, the particle system dynamics also allows us to obtain refined results on entropic convergence to stationarity of the nonlinear dynamics. 

\subsection{Convergence of nonlinear dynamics}
Our first results show that both the mean-field and single-site dynamics, when appropriately initialized, converge to the Gibbs distribution of the corresponding hard-core model, for a suitable regime of parameter values. Unlike traditional Markov chains, such a qualitative convergence is by no means immediate.

\begin{proposition}%[informal versions of \Cref{prop:qual:mean-field-convergence} and \Cref{prop:qual:single-site-convergence}]
\label{thm:intro-convergence}
    Consider the mean-field and single-site nonlinear dynamics on independent sets in a graph $G=(V,E)$. Suppose the maximum degree of $G$ is $\Delta$.
    
    \begin{enumerate}[label=(\arabic*)]
        \item If $p$ is a distribution on $\OmegaG$ with density $\alpha = \alpha(p)$ satisfying $0 < \alpha < \frac{1}{\Delta+1}$, then
        \begin{align*}
            \lim_{t \to \infty} \Tmf^t(p) = \mu_G^{\alpha}.
        \end{align*}
        
        \item If $p$ is a distribution on $\OmegaG$ with marginal vector $\vect m = \vect m (p)$ satisfying $0 < m_v < \frac{1}{\Delta+1}$ for all $v \in V$, then
        \begin{align*}
            \lim_{t \to \infty} \Tss^t(p) = \mu_G^{\vect m}.
        \end{align*}   
    \end{enumerate}
\end{proposition}
\noindent %Note that the regime $0 < \alpha < \frac{1}{\Delta+1}$ for the density is essentially optimal in the sense that, if $ \alpha > \frac{1}{\Delta+1}$, then there exist graphs $G$ with maximum degree $\D$ such that $\O_G=\emptyset$, as one can see by taking, e.g., the complete graph with $\D+1$ vertices.  
Note that the regime $0 < \alpha < \frac{1}{\Delta+1}$ for the density, or $0 < m_v < \frac{1}{\Delta+1}$ for marginal occupation probabilities, is necessary, so that $\mu_G^{\alpha}$ or $\mu_G^{\vect m}$ is well-defined on any graph $G$ of maximum degree $\Delta$. In contrast, the complete graph with $\D+1$ vertices cannot achieve density $\alpha \ge \frac{1}{\Delta+1}$ or uniform marginal probabilities $m_v = c \ge \frac{1}{\Delta+1}$ at finite fugacity.

We shall prove a slightly more general statement that holds for a family of collision kernels including both the single-site and mean-field cases; see \Cref{th:multi-mean-field-convergence}. The proof strategy is based on a general convergence theorem for nonlinear dynamics from \cite[Theorem~2.8]{caputo2024nonlineardynamicsisingmodel}. The main condition to check in order to apply that theorem is an \emph{irreducibility} property, namely that after some finite time the nonlinear dynamics $T^t(p)$ puts strictly positive mass on every independent set in~$\OmegaG$.
The 
facts that the mean-field dynamics preserves the density~$\alpha(p)$, and that the single-site dynamics preserves the marginal vector~$\vect m(p)$, then identify the limiting stationary distribution as the unique Gibbs measure with the prescribed density or marginal vector.

\Cref{thm:intro-convergence} does not provide any information on the rate of convergence. Indeed, in contrast with the classical Markov chain setting, it remains open whether convergence is exponentially fast, even with a possibly very small dimension-dependent rate. In particular, there is no general contractive coupling available in the nonlinear setting. Moreover, the total variation distance to equilibrium need not even be monotonically decreasing with 
time~\cite{caputo2025cutoff}.

Our next result shows that, assuming the density is sufficiently low, 
a contractive coupling for both the mean-field and single-site dynamics exists, and exponential convergence to equilibrium occurs at a rate $\d/n$ where $\d>0$ is a constant.

\begin{theorem}%%[informal version of \Cref{thm:quant:mean-field-rapid-convergence} and \Cref{thm:quant:single-site-rapid-convergence}]
\label{th:informal1}
%    Consider the mean-field dynamics on independent sets in a graph $G=(V,E)$. Suppose the maximum degree of $G$ is $\Delta$. If $p$ is a distribution on $\OmegaG$ with density $0 < \alpha(p) < \frac{1}{3(\Delta+1)}$, then for $t = O(n \log n)$ we have
 %       \begin{align*}
  %          \|\Tmf^t(p) - \mu_{G,\lambda} \|_{\TV} \le \frac{1}{4},
   %     \end{align*}
    %    where $\lambda > 0$ is the unique fugacity such that $\alpha(\mu_{G,\lambda}) = \alpha(p)$.
Consider the mean-field and single-site nonlinear dynamics on independent sets in a graph $G=(V,E)$ on $n$ vertices. Suppose the maximum degree of $G$ is $\Delta$. Let $\delta \in (0,1)$. 
    \begin{enumerate}[label=(\arabic*)]
        \item If $p$ is a distribution on $\OmegaG$ with density $\alpha = \alpha(p)$ satisfying $ \alpha \le \frac{1-\d}{3(\Delta+1)}$, then for all $t\in\bbN$,
        \begin{align*}
          \|\Tmf^t(p) - \mu_G^{\alpha} \|_{\TV} \le n\,e^{-\d\,t/n}.
        \end{align*}
        
        \item If $p$ is a distribution on $\OmegaG$ with marginal vector $\vect m = \vect m (p)$ satisfying $ m_v \le \frac{1-\d}{3(\Delta+1)}$ for all $v \in V$, then for all $t\in\bbN$,
        \begin{align*}
            \|\Tss^t(p) - \mu_G^{\vect m} \|_{\TV} \le n\,e^{-\d\,t/n}.
        \end{align*}   
    \end{enumerate}
    In particular, in both cases, after $t=O(n\log(n/\varepsilon))$ steps the system is within variation distance $\varepsilon$ from stationarity, for any $\varepsilon>0$. 
\end{theorem}
The proof is based on a one-step contraction in the Wasserstein metric associated with the Hamming distance. %Although the coupling constructions differ in the two settings, they are conceptually similar. 
Take the mean-field case as an example. Our goal is to show that, if $p$ and $q$ have the same density $\alpha$, then one step of the dynamics contracts their Wasserstein distance by a factor $1-\frac{\delta}{n}$. To this end, we introduce a procedure (\Cref{alg:quant:multi-mean-field-coupling}) that lifts any coupling of $p$ and $q$ to a coupling of $\Tmf(p)$ and $\Tmf(q)$, while decreasing the expected Hamming distance.  In one update of the mean-field dynamics, a vertex $v$ is selected as the site to be updated, while a second vertex $u$ is selected to provide the auxiliary spin used in the possible exchange. The key observation is that, whenever the current spin at $v$ is zero, the auxiliary configuration is needed only through its value at $u$. Since $p$ and $q$ have the same density $\a$, and $u$ is uniformly random, this single bit can be coupled directly between the two updates.
Using  this idea, we show that if $\a$ is sufficiently low, then disagreements at the updated site $v$ are resolved with reasonably high probability, whereas agreements at $v$ give rise to new disagreements only with fairly small probability. Quantifying these competing effects yields the desired contraction in expected Hamming distance.
%At a high level, this argument may be seen as a nonlinear analog of the well known path coupling argument for Markov chains \cite{bubley1997path}. The construction is however more involved due to the quadratic nature of the evolution.

We emphasize that our methods for constructing couplings are novel and differ from existing coupling techniques used for linear Markov chains. For instance, the well-known path coupling argument for Markov chains~\cite{bubley1997path} cannot be applied directly to nonlinear dynamics. This is because creating a path between two independent sets either requires the inclusion of non-independent sets along the path, or alters the density or marginal vector of the intermediate independent sets, which poses considerable challenges for analysis. 
Additionally, the traditional coupling approach requires showing that the expected one-step contraction in Hamming distance occurs for two variables \(X_0\) and \(Y_0\) generated from any pair of coupled initial distributions \(p\) and \(q\). Thus, it is sufficient to consider \emph{fixed} initializations, such that \(X_0 = x_0\) and \(Y_0 = y_0\) (in other words, \(p\) and \(q\) are both delta measures). However, for our analysis, both \(p\) and \(q\) must have the specified density or marginal vector, and our coupling analysis critically depends on this fact.

While the threshold of $\frac{1}{3(\D+1)}$ in \Cref{th:informal1} is presumably suboptimal, it is considerably tighter than the analogous rate of convergence results for the Ising model (and other models with soft constraints) proved using more sophisticated recursive constructions in~\cite{caputo2024nonlineardynamicsisingmodel}; those results hold for values of the analogous ``temperature'' parameter up to $\frac{c}{\D}$, where $c$ is a small constant.  Indeed, 
the ideas in the proof of \Cref{th:informal1} can be used to provide a simpler proof of one of the main results in \cite{caputo2024nonlineardynamicsisingmodel} (cf.\ Theorem 1.2 there). We point out also that techniques used in~\cite{RRS98} to prove lower bounds imply that the convergence rate of order $1/n$ in \Cref{th:informal1} is optimal up to a constant. 

We also provide precise results on the {\it decay of relative entropy\/} for the nonlinear dynamics.  These are more conveniently formulated for a continuous-time version of the nonlinear mean-field dynamics, which we define naturally as follows, where $p_t$ denotes the distribution of the dynamics at time $t\ge 0$:
\begin{equation}\label{eq:pre:continuous}
\frac{d p_t}{dt}
=
\Tmf(p_t)-p_t,
\qquad p_0=p.
\end{equation}
%The dynamical system \eqref{eq:pre:continuous} is well posed, and for any initial distribution $p$, the solution $p_t$ admits an explicit probabilistic representation; see, e.g., \cite{caputo2025kac}. 
As in classical kinetic theory, the relative entropy of $p_t$ with respect to $\mu_G^\alpha$, which we denote by $H(p_t \mid \mu_G^\alpha)$, provides a natural Lyapunov functional, decreasing monotonically in time and quantifying convergence to equilibrium.

We establish decay of relative entropy for the continuous-time mean-field dynamics for all densities less than the \emph{critical density} $\a_c(\D)$, which is defined by
\[
\alphac
=
\frac{\lambdac}{1+(\Delta+1)\lambdac},
\]
where $\lambdac = \frac{(\Delta-1)^{\Delta-1}}{(\Delta-2)^\Delta}$ denotes the critical fugacity for the $\D$-regular tree.
This critical density $\alpha_c(\Delta)$ was first introduced by Davies and Perkins~\cite{davies2023approximatelycountingindependentsets}, marking a computational phase transition for approximately sampling and counting fixed-size independent sets. Specifically, for independent sets of a given size $\alpha n$ in an $n$-vertex graph of maximum degree $\Delta$, there exist efficient sampling and counting algorithms if $\alpha < \alpha_c(\Delta)$, while approximate counting (and thus sampling) is intractable if $\alpha > \alpha_c(\Delta)$ under widely believed complexity conjectures (namely, $\text{NP} \neq \text{RP}$).

\begin{theorem}%[Informal version of \Cref{thm-entropydecaymf}]
\label{thm-informal-ct-ent}
Let $G$ be a graph on $n$ vertices with maximum degree $\D \ge 3$. For any
$\d\in(0,1)$, and any initial distribution $p$ over $\O_G$ with density $\alpha(p)=\alpha$ satisfying
$\d\le \a \le (1-\d)\a_c(\D)$, the continuous-time version of the mean-field dynamics~$p_t$ defined in~\eqref{eq:pre:continuous} satisfies
\[
  H(p_t \mid \mu_G^\alpha) \,\le\, e^{-r\,t/n}\,H(p \mid \mu_G^\alpha),\qquad t\ge0,
\]
for some constant $r = r(\d,\D) > 0$ independent of $n$.
\end{theorem}

\Cref{thm-informal-ct-ent} offers improvements over \Cref{th:informal1} for mean-field dynamics in two respects. First, while \Cref{th:informal1} is applicable only for densities below $\frac{1}{3(\Delta+1)}$, \Cref{thm-informal-ct-ent} extends this to densities up to $\alpha_c(\Delta) \approx \frac{e}{(1+e)\Delta}$. 
We suspect that rapid convergence of the mean-field dynamics fails when $\a > \a_c(\D)$ since, as we will see in \Cref{thm:prescribed-density-particle-system-sampler,thm:hardness-prescribed-density}, the critical density $\alpha_c(\Delta)$ characterizes the computational phase transition for sampling from~$\mu_G^\a$.
Second, the exponential decay of relative entropy in \Cref{thm-informal-ct-ent} is stronger than the bounds on total variation distance, and can be interpreted as a \emph{modified log-Sobolev inequality} for nonlinear dynamics. Specifically, by applying Pinsker's inequality, we can translate \Cref{thm-informal-ct-ent} to bounds on variation distance between $p_t$ and~$\mu_G^\a$.

Our proof of \Cref{thm-informal-ct-ent} proceeds via Kac's program based on particle systems discussed in the next subsection, as developed in \cite{carlen2010entropy,mischler2013kac}. 
Interestingly, in sharp contrast with the classical setting of kinetic theory, the program can be carried out here in a much stronger form, yielding decay rates that are independent of the initial
distribution~$p$. A key ingredient in our proof is tight control 
of the entropy production in the particle system, which in our setting can be inferred from the recent works \cite{jain2023optimalmixingdownupwalk,JainPhamVuong_EntropicSparseLocalization}. 

\subsection{Algorithms via particle systems}

As noted earlier, the direct implementation of nonlinear dynamics can be prohibitive because generating a single sample from \( T(p) \) requires two independent samples from \( p \). Consequently, producing an output from \( T^t(p) \) after \( t \) iterations necessitates \( 2^t \) independent samples from the initial distribution \( p \). Given that our mixing time in \Cref{th:informal1}  is \( t = O(n \log n) \), the overall runtime becomes exponential.

We can overcome this obstacle using the idea of {\it particle systems}, a concept from kinetic theory dating back to Mark Kac~\cite{kac56}.  The idea is to approximate the nonlinear evolution via a finite population of $N$ particles $\sigma^{(i)}$, $i=1,\dots,N$, each represented by an independent set $\sigma^{(i)}\in\O_G$, which undergo pairwise collisions 
%under the same kernel~$Q$ and the mass action principle (i.e., colliding pairs $(\sigma,\sigma')$ are chosen with probability $p(\sigma)p(\sigma')$).  
\[
(\sigma^{(i)},\sigma^{(j)})\mapsto(\tau^{(i)},\tau^{(j)})
\]
where the pair $(\tau^{(i)},\tau^{(j)})$ is obtained from the pair $(\sigma^{(i)},\sigma^{(j)})$ through the exchange of two bits $\sigma^{(i)}_v$ and $\sigma^{(j)}_u$, according to a choice of sites $u,v$ that matches the behavior of the nonlinear collision kernel. This evolution is actually a (linear) Markov chain on the set~$\OmegaG^N$; crucially, it is efficiently implementable (provided $N$ is not too large).  Moreover, one expects that if the number of particles $N$ is sufficiently large, then the particle system should approximate the evolution of the nonlinear dynamics.  In particular, since the nonlinear dynamics converges to the desired Gibbs distribution, then after sufficiently many steps the law of a typical particle should be close to this target measure.

This suggests the following general algorithmic template:
\begin{enumerate}
    [label = (\arabic*)]
    \item Initialize $N$ particles so that the relevant conserved quantity (density $\alpha$ or marginal vector $\vect m$) has the desired value.
    \item Run the corresponding particle-system chain for $t$ steps, where $t$ is large enough.
    \item Output one particle.
\end{enumerate}
The analysis of this algorithm involves two main steps. First, we show that, for $N$ sufficiently large, the one-particle marginal of the stationary distribution of the particle system is close to the desired hard-core measure. Second, we bound the mixing time of the particle-system chain, which determines how large $t$ should be. Combining these two steps shows that the output particle is approximately distributed according to the target measure in polynomial time.

For the mean-field dynamics, we carry out this program for densities up to the critical density $\a_c(\D)$ and obtain a new sampling algorithm for the hard-core Gibbs distribution parameterized by the density.

\begin{theorem}
    %[informal version of \Cref{thm:prescribed-density-particle-system-sampler}]
    \label{thm:intro-alg}
There is a sampling algorithm (\Cref{alg:prescribed-density-particle-system-sampler}) such that, for any graph $G$ on $n$ vertices with maximum degree $\Delta \ge 3$, any density $\d\le \alpha \le (1-\delta)\alpha_c(\Delta)$ where $\delta \in (0,1)$, and any approximation error $\eps \in (0,1)$, the algorithm runs in time\footnote{Hereafter, we use $\widetilde O(\cdot)$ to suppress dependencies on $\poly(\log n, \log(1/\eps))$.} $\widetilde O(\frac{n}{\eps})$ and outputs an independent set from a distribution $\mu_{\mathrm{alg}}$ on $\OmegaG$ satisfying
        \begin{align*}
\|\mu_{\mathrm{alg}} - \mu_G^{\alpha} \|_{\TV} \le \eps.
        \end{align*}
%        The running time is $\widetilde O(\frac{n}{\eps^2})$. 
\end{theorem}
We note that the critical threshold $\alpha_c(\Delta)$ in \Cref{thm:intro-alg} is the best possible, beyond which the sampling problem becomes computationally hard. Specifically, in \Cref{thm:hardness-prescribed-density}, we show that there is no polynomial-time algorithm to sample from $\mu_{G}^{\alpha}$ for any fixed $\alpha \in (\alphac, \frac{1}{\Delta + 1})$, assuming $\mathrm{NP} \neq \mathrm{RP}$. This result follows from an adaptation of the approach used in \cite{davies2023approximatelycountingindependentsets} for proving the hardness of sampling from $\mu_{G,\lambda}(\cdot \mid |\sigma| = \alpha n)$ under the same regime (that is, uniformly random independent sets of size $\alpha n$).

The proof of \Cref{thm:intro-alg} follows the two-step strategy mentioned above. For the first step, let $\muNm$ be the stationary distribution of the mean-field particle-system chain, where $M=\lfloor Nn\alpha\rfloor$ is the prescribed total occupation over all particles. This distribution is uniform over all $N$-particle configurations with total occupation $M$. Equivalently, it can be viewed as $N$ independent samples from $\mu_G^{\alpha}$ conditioned on the event that the total number of occupied vertices over all the particles is exactly $M$. We prove that this conditioning does not significantly change the law of a typical particle. The intuition is that $M$ is chosen close to the expected total occupation under the product measure~$\bigl(\mu_G^\alpha\bigr)^{\otimes N}$, and fixing one particle~$\sigma$ only changes the required total occupation of the remaining particles by $|\sigma|$, which is at most~$n$. When $N$ is large, this change is small compared to the fluctuations of the total occupation over all particles. We make this intuition precise using local central limit theorem estimates, which show that the one-particle marginal of $\muNm$ is within variation distance $\eps$ of $\mu_G^{\alpha}$ when $N = \widetilde O(\frac{1}{\eps})$. For the second step, we use entropy contraction for the mean-field particle system to show that the particle-system chain converges to $\muNm$ within time~$\widetilde O(Nn)$. This step relies crucially on a recent result from \cite{jain2023optimalmixingdownupwalk,JainPhamVuong_EntropicSparseLocalization}, which allows us to obtain a tight entropy production bound for the particle system at all subcritical densities. Therefore, outputting one particle from the resulting configuration gives a sample close to $\mu_G^{\alpha}$. 

We go on to provide an analogous analysis for the particle system simulation of the single-site dynamics.

\begin{theorem}%[informal version of \Cref{thm:particle-system-sampler}]
\label{thm:intro-ss}
There is a sampling algorithm (\Cref{alg:prescribed-marginal-particle-system-sampler}) such that, for any graph $G$ on $n$ vertices with maximum degree $\Delta \ge 0$, any marginal vector $\vect m \in (0,1)^V$ such that $\delta \le m_v \le \frac{1-\delta}{2(\Delta+1)}$ for all $v \in V$ where $\delta \in (0,1)$, and any approximation error $\eps \in (0,1)$, the algorithm runs in time $\widetilde O(\frac{n^3}{\eps^4})$ and outputs an independent set from a distribution $\mu_{\mathrm{alg}}$ on $\OmegaG$ satisfying
\begin{align*}
\|\mu_{\mathrm{alg}} - \mu_G^{\vect m} \|_{\TV} \le \eps.
\end{align*}
%The running time is $\widetilde O(\frac{n^3}{\eps^4})$. 
\end{theorem}

\begin{remark}\label{rem:coloptbound}
    In fact, our sampling algorithm works as long as all marginals are less than $\frac{1}{c_{\mathrm {col}}(\Delta + 1)}$, where $c_{\mathrm {col}}$ is the smallest constant such that the Glauber dynamics for sampling $q$-colorings on any graph of maximum degree $\Delta$ (not necessarily bounded) mixes in nearly linear time whenever $q \ge c_{\mathrm {col}} \Delta$.
    Under the widely believed conjecture that Glauber dynamics has optimal mixing when $q \ge \Delta+2$, we would achieve the marginal bound $\frac{1-\delta}{\Delta + 1}$ for any $\delta \in (0,1)$, matching the constraint in \Cref{thm:intro-convergence}.
    This would not contradict the hardness derived from~\cite{davies2023approximatelycountingindependentsets} as discussed immediately following \Cref{thm:intro-alg}, since that hardness corresponds to constraints on the global density rather than individual marginals.
\end{remark}

The proof of \Cref{thm:intro-ss} is again based on the two steps as before, but the ingredients are different from the density case in \Cref{thm:intro-alg}. Let $\vec r=(r_v)_{v\in V}$ be an integer vector with $r_v = \floor {Nm_v}$, and let $\muNr$ be the stationary distribution of the single-site particle-system chain. This distribution is uniform over all $N$-particle configurations in which exactly $r_v$ particles contain vertex $v$, for every $v\in V$. Equivalently, it can be viewed as $N$ independent samples from $\mu_G^{\vect m}$ conditioned on having these exact vertex counts. For the first step, we show that this conditioning does not significantly change the law of a typical particle. By a standard relative entropy argument, it is enough to show that the exact-count event has sufficiently large probability under $\bigl(\mu_G^{\vect m}\bigr)^{\otimes N}$. To lower bound this probability, we first use concentration inequalities and a union bound to show that the vertex counts are close to their expected values with high probability, and then compare the probabilities of nearby count classes. This shows that, for $N=\widetilde O(\frac{n^2}{\eps^4})$, the one-particle marginal of $\muNr$ is close enough to $\mu_G^{\vect m}$ with $\eps$ total variation error. For the second step, we prove rapid mixing of the single-site particle-system chain by comparing it with Glauber dynamics for proper colorings of an auxiliary graph constructed from $G$ and $\vec r$. This gives mixing after $\widetilde O(Nn)$ many steps. Combining these two steps shows that, after running the chain for its mixing time and outputting one particle, the output distribution is close to $\mu_G^{\vect m}$.

We remark that, while our algorithms based on particle systems are crucially inspired by the nonlinear dynamics, their analysis does not depend on the convergence rate of those nonlinear dynamics, nor do they imply such a rate.  See \Cref{sec:open} for a further discussion of this point.

\medskip

We conclude the introduction by comparing our algorithms based on nonlinear dynamics with more classical approaches that involve solving the inverse problem (i.e., learning the relevant fugacities given the density or marginal vector).  
The formal statements and proofs are included in \Cref{app:gradient-descent-prescribed-marginals,app:bisection-prescribed-density}, and we summarize them here.

To sample from \(\mu_G^{\alpha}\) with a specified density \(\alpha \leq (1 - \delta)\alphac\), we can begin by utilizing the bisection method (as outlined in \Cref{alg:app:bisection-sampler}) to estimate the corresponding fugacity \(\lambda\) of the hard-core model. Then, we can generate samples using efficient sampling algorithms, such as Glauber dynamics, based on this value of~$\lambda$. The running time of such an algorithm is $\widetilde O(\frac{n}{\eps^2})$, whereas our algorithm achieves a faster running time $\widetilde O(\frac{n}{\eps})$.

To sample from the distribution $\mu_G^{\vect m}$ with a specified marginal vector $\vect m$, we can utilize the gradient descent method (see \Cref{alg:app:gd-sampler}) to determine the fugacity vector $\vect \lambda$. Once we have (an estimate of) $\vect \lambda$, we can sample from $\mu_{G,\vect \lambda}$ using Glauber dynamics. 
The running time for this algorithm is $\widetilde O(\frac{n^2}{\eps^2})$, which is more efficient than what we claimed in \Cref{thm:intro-ss} for our algorithm. 
It is thus interesting to see whether our analysis can be improved to match this quadratic running time. One possibility is to apply strong local central limit theorems as we did in the mean-field case; however, this becomes technically much more challenging in the single-site case and the bounds we have been able to obtain so far are not sufficient to establish the desired result.

The gradient descent algorithm works when all marginal probabilities are less than $\alphac$. 
We note that our algorithm is potentially effective over the wider range of marginals up to $\frac{1}{\Delta + 1}$, assuming the widely accepted mixing conjecture on colorings (see \Cref{rem:coloptbound}). This raises the interesting question of whether the sampler based on gradient descent could also work efficiently in this regime.

In summary, we see that in the prescribed-density case, our nonlinear dynamics based algorithm has a better performance guarantee than the traditional learning method. In the prescribed-marginal case, however, the performance guarantee for our algorithm is slightly worse than that of the gradient descent approach. Nevertheless, in both cases, our algorithms have the advantage of being very simple to implement, involving only the direct simulation of a particle system of suitable size with no additional computation.  Moreover, they represent a potentially powerful new computational paradigm whose full power is far from fully understood.

\section{Preliminaries and basic properties}\label{sec:prelim}

\subsection{Graphs, independent sets, and hard-core measures}

Let $G=(V,E)$ be a finite graph with vertex set $V$, edge set $E$,
and maximum degree $\Delta$. We write $n=|V|$. For $v\in V$, let
\[
\neighbor{v}:=\{u\in V:\{u,v\}\in E\},
\qquad
\neighborCl{v}:=\neighbor{v}\cup\{v\}.
\]

An \emph{independent set} of $G$ is a set $I\subseteq V$ such that
no two vertices of $I$ are adjacent. We identify an independent set
with its indicator vector $\sigma\in\{0,1\}^V$, where $\sigma_v=1$
means that $v$ is occupied (in the set) and $\sigma_v=0$ means that $v$ is vacant (not in the set).
Thus
\[
\OmegaG
:=
\left\{
\sigma\in\{0,1\}^V:
\sigma_u\sigma_v=0 \text{ for every } \{u,v\}\in E
\right\}
\]
is the set of all independent sets of $G$. We use set notation and
vector notation interchangeably; for example, $v\in\sigma$ means
$\sigma_v=1$, and
\[
|\sigma|:=\sum_{v\in V}\sigma_v
\]
denotes the size of the independent set.

We write $\Pdist(\OmegaG)$ for the set of probability measures on
$\OmegaG$, and
\[
\Ppos(\OmegaG)
:=
\{p\in \Pdist(\OmegaG): p(\sigma)>0 \text{ for all } \sigma\in\OmegaG\}
\]
for the set of strictly positive probability measures. %For
%$p\in\Pdist(\OmegaG)$, we denote its support by
%\[
%\supp(p):=\%{\sigma\in\OmegaG:p(\sigma)>0\}.
%\]

For a fugacity $\lambda>0$, the hard-core Gibbs measure with uniform
fugacity $\lambda$ is the probability measure $\mu_{G,\lambda}$ on
$\OmegaG$ given by
\[
\mu_{G,\lambda}(\sigma)
=
\frac{\lambda^{|\sigma|}}{Z_G(\lambda)},
\qquad
Z_G(\lambda)
:=
\sum_{\sigma\in\OmegaG}\lambda^{|\sigma|}.
\]
More generally, for a vector of positive fugacities
$\vect\lambda=(\lambda_v)_{v\in V}\in(0,\infty)^V$, we define
\[
\mu_{G,\vect\lambda}(\sigma)
=
\frac{1}{Z_G(\vect\lambda)}
\prod_{v:\sigma_v=1}\lambda_v,
\qquad
Z_G(\vect\lambda)
:=
\sum_{\sigma\in\OmegaG}
\prod_{v:\sigma_v=1}\lambda_v .
\]
The uniform-fugacity measure $\mu_{G,\lambda}$ is the special case
$\lambda_v=\lambda$ for all $v\in V$.

\subsection{Density, marginal vectors, and fugacity parametrizations}\label{subsec:ab}

The \emph{density} of $p\in\Pdist(\OmegaG)$ is defined as
\[
\alpha(p)
:=
\frac{1}{n}\E_{\sigma\sim p}[|\sigma|].
\]
For each vertex $v\in V$, the $v$-th marginal of $p$ is
\[
m_v(p)
:=
\Pr_{\sigma\sim p}[\sigma_v=1].
\]
We write 
\(
\vect m(p) = (m_v(p))_{v \in V}
\)
for the \emph{marginal vector} of $p$.
Thus
\[
\alpha(p)=\frac1n\sum_{v\in V}m_v(p).
\]
%When a marginal vector is prescribed independently of a distribution, we
%denote it by
%\[
%\vect m=(m_v)_{v\in V}.
%\]
For the uniform hard-core measure, we have
\[
\alpha(\mu_{G, \lambda})
=
\frac{\lambda}{n}\frac{ d}{ d\lambda}
\log Z_G(\lambda).
\]

We will also use log-fugacity coordinates. For
$\vect\theta=(\theta_v)_{v\in V}\in\mathbb R^V$, define (with slight abuse of notation)
\[
Z_G(\vect\theta)
:=
\sum_{\sigma\in\OmegaG}
\exp\left(\langle \vect\theta,\sigma\rangle\right),
\qquad
\AG(\vect\theta):=\log Z_G(\vect\theta),
\]
and let
\[
\mu_{G,\vect\theta}(\sigma)
=
\exp\left(
\langle\vect\theta,\sigma\rangle-\AG(\vect\theta)
\right),
\qquad \sigma\in\OmegaG.
\]
This is the same as the fugacity parametrization with
$\lambda_v=e^{\theta_v}$. In these coordinates,
\[
\nabla \AG(\vect\theta)
=
\E_{\mu_{G,\vect\theta}}[\sigma]
= \vect m(\mu_{G,\vect\theta})
 \;\;\;\hbox{\rm and} \;\;\; \nabla^2\AG(\vect \theta) = \Cov_{\mu_{G, \vect\theta}}(\sigma).
\]

\subsection{The uniqueness threshold and the critical density}

For $\Delta\ge 3$, the hard-core uniqueness threshold on the
$\Delta$-regular tree is
\[
\lambdac
=
\frac{(\Delta-1)^{\Delta-1}}{(\Delta-2)^\Delta}.
\]
We define the corresponding critical density by
\begin{equation}\label{eq:crit_dens}
\alphac
:=
\frac{\lambdac}{1+(\Delta+1)\lambdac}.
\end{equation}
Equivalently, $\alphac$ is the occupancy ratio of the complete graph
$K_{\Delta+1}$ at fugacity $\lambdac$. The regime
$
\alpha < \alphac
$
will be referred to as the \emph{subcritical} density regime. We refer to \cite{davies2023approximatelycountingindependentsets} for a discussion of the critical density parameter.

\subsection{Distances and entropy}
\label{subsec:ac}
For $\sigma,\eta\in\OmegaG$, let
\[
\dHam(\sigma,\eta)
:=
|\{v\in V:\sigma_v\ne \eta_v\}|
\]
be the Hamming distance. For $p,q\in\Pdist(\OmegaG)$, the
$L^1$-Wasserstein distance induced by $\dHam$ is
\[
W_1(p,q)
:=
\inf_{\pi\in\Gamma(p,q)}
\E_{(\sigma,\eta)\sim\pi}
\left[
\dHam(\sigma,\eta)
\right],
\]
where $\Gamma(p,q)$ is the set of all couplings of $p$ and $q$.

The total variation distance is
\[
\|p-q\|_{\TV}
:=
\frac12\sum_{\sigma\in\OmegaG}|p(\sigma)-q(\sigma)|.
\]

For $p,q\in\Pdist(\OmegaG)$, the relative entropy of $p$ with respect
to $q$ is
\[
H(p\mid q)
:=
\begin{cases}
\displaystyle
\sum_{\sigma\in\OmegaG}
p(\sigma)\log\frac{p(\sigma)}{q(\sigma)},
& p\ll q,\\[1.2em]
+\infty, & \text{otherwise}.
\end{cases}
\]
We also write
\[
D_{\mathrm{KL}}(p\|q):=H(p\mid q)
\]
when convenient.
\subsection{The collision kernel}

A \emph{collision} between two independent sets $(\sigma,\sigma')\in\OmegaG\times\OmegaG$
consists of exchanging the value of $\sigma$ at one vertex with the value of
$\sigma'$ at another vertex. For $v\in V$, $\sigma\in\{0,1\}^V$, and
$a\in\{0,1\}$, let $S_v(\sigma,a)$ be the configuration defined by
\[
S_v(\sigma,a)_u
=
\begin{cases}
\sigma_u, & u\ne v,\\
a, & u=v.
\end{cases}
\]
Thus the exchange of the $v$-th coordinate of $\sigma$ with the $u$-th
coordinate of $\sigma'$ gives the pair
\begin{equation}\label{eq:pre:exchange-map}
(\sigma,\sigma')
\mapsto
(\sigma,\sigma')^{v,u}
:=
\bigl(S_v(\sigma,\sigma'_u),S_u(\sigma',\sigma_v)\bigr).
\end{equation}

This exchange is allowed only when it produces two independent sets. In this case we say the exchange is \emph{admissible}. We call
$(v,u)_{\sigma,\sigma'}$ the event that the exchange $(v, u)$ is admissible for the pair of independent sets $(\sigma, \sigma')$, and we write
$\widehat{(v, u)}_{\sigma,\sigma'}$ for its complement.

For fixed $v,u\in V$, define the local kernel $Q_{v,u}$ on
$\OmegaG\times\OmegaG$ by
\[
Q_{v,u}(\sigma,\sigma';\tau,\tau')
=
\mathbf 1\{(v,u)_{\sigma,\sigma'}\}
\mathbf 1\{(\tau,\tau')=(\sigma,\sigma')^{v,u}\}
+
\mathbf 1\{\widehat{(v,u)}_{\sigma,\sigma'}\}
\mathbf 1\{(\tau,\tau')=(\sigma,\sigma')\}.
\]
In words, the proposed exchange is performed if it is admissible, and
otherwise the pair is left unchanged.
Also, we define $\Phi_{v, u}(\sigma, \sigma')$ to be the unique pair $(\tau, \tau')$ such that $Q_{v,u}(\sigma, \sigma'; \tau, \tau') > 0$, i.e., 
\begin{equation}\label{eq:phi_def}
\Phi_{v,u}(\sigma,\sigma')
:=
\begin{cases}
(\sigma,\sigma')^{v,u}, & \text{if } (v,u)_{\sigma,\sigma'} \text{ holds},\\
(\sigma,\sigma'), & \text{otherwise}.
\end{cases}
\end{equation}

One can define a collision kernel associated with any
symmetric stochastic matrix $\transK$ on $V$:
\begin{equation}\label{eq:pre:transport-collision-kernel}
\cQ^{\transK}(\sigma,\sigma';\tau,\tau')
=
\frac1n
\sum_{v, u\in V}
\transK(v, u)Q_{v, u}(\sigma,\sigma';\tau,\tau').
\end{equation}

Equivalently, $v$ is chosen uniformly from $V$, then $u$ is chosen
according to $\transK(v,\cdot)$, and the exchange $(v, u)$ is attempted.

For fixed $(\sigma,\sigma')\in\OmegaG\times\OmegaG$,
the quantity $\cQ^{\transK}(\sigma,\sigma';\tau,\tau')$ is the probability
that the next pair is $(\tau,\tau')$. Thus $\cQ^{\transK}$ can be viewed as
a transition kernel for a Markov chain on the pair space
$\OmegaG\times\OmegaG$. When the choice of $\transK$ is not important, we
write $\kernelQ$ for a generic collision kernel on this pair space.

Whenever $\kernelQ=\cQ^{\transK}$ for a symmetric stochastic matrix
$\transK$, the kernel $\kernelQ$ satisfies the mild exchange symmetry
\begin{equation}\label{eq:pre:mild-exchange-symmetry}
\kernelQ(\sigma,\sigma';\tau,\tau')
=
\kernelQ(\sigma',\sigma;\tau',\tau),
\end{equation}
for all $\sigma,\sigma',\tau,\tau'\in\OmegaG$. Note that exchanging the two
input configurations and the two output configurations corresponds to
replacing the proposed exchange $(v,u)$ by $(u,v)$, and this has the same
weight since $\transK$ is symmetric. We will mostly consider the following two cases for the matrix $\transK$.

The \emph{mean-field} exchange kernel corresponds to
\begin{equation}\label{eq:pre:mean-field-collision-kernel}
\Kmf(v, u)=\frac1n
\qquad \forall v, u\in V,
\end{equation}
% so that
% \begin{equation}\label{eq:pre:mean-field-collision-kernel}
% \Qmf(\sigma,\sigma';\tau,\tau')
% =
% \frac1{n^2}
% \sum_{v, u \in V}
% Q_{v, u}(\sigma,\sigma';\tau,\tau').  \end{equation}
%The 
while the \emph{single-site} exchange kernel corresponds to
\begin{equation}\label{eq:pre:single-site-collision-kernel}
\Kss(v,u)=\mathbf 1\{v = u\}.
\end{equation}
% so that
% \begin{equation}\label{eq:pre:single-site-collision-kernel}
% \Qss(\sigma,\sigma';\tau,\tau')
% =
% \frac1n
% \sum_{v\in V}
% Q_{v,v}(\sigma,\sigma';\tau,\tau').
% \end{equation}
It will be convenient to treat both mean-field and single-site kernels in a unified way using the following more general kernel, which we call \emph{multi-component mean-field}. 

Let \(\cA\) be a partition of the vertex set \(V\), i.e., a collection of pairwise disjoint nonempty sets \(A\subseteq V\) whose union is \(V\). The $\cA$-multi-component mean-field kernel associated with \(\cA\) is
\begin{equation}\label{eq:KA}
\cK_\cA(v,u)=\frac{\mathbf{1}(u\in A(v))}{|A(v)|},
\end{equation}
where \(A(v)\in\cA\) denotes the unique set containing \(v\in V\). 
Note that, if the partition is \(\cA=\{\{u\}:u\in V\}\) then $\cK_\cA=\Kss$, while if  \(\cA=\{V\}\), then $\cK_\cA=\Kmf$. All other cases interpolate between these two extremes. 
We sometimes write $\cQ^\cA:=\cQ^{\cK_\cA}$ for the associate collision kernel, 
\begin{equation}\label{eq:QA}
\cQ^\cA(\sigma,\sigma';\tau,\tau')
=
\frac1n
\sum_{v\in V}\frac1{|A(v)|}
\sum_{u\in A(v)}Q_{v,u}(\sigma,\sigma';\tau,\tau').
\end{equation}

\subsection{The nonlinear dynamics}

We now describe how a collision kernel induces a nonlinear evolution on
probability measures. Let $\kernelQ$ be a collision kernel on $\OmegaG\times\OmegaG$. For fixed
$\sigma,\sigma'\in\OmegaG$, define the marginal of the first coordinate by
\begin{equation}\label{eq:pre:first-coordinate-marginal}
Q(\tau\given \sigma,\sigma')
:=
\sum_{\tau'\in\OmegaG}
\kernelQ(\sigma,\sigma';\tau,\tau').
\end{equation}
Thus $Q(\cdot\given \sigma,\sigma')$ is the law of the first output
configuration of the collision. %\footnote{In the introduction, we used the term ``collision kernel'' to denote this Markov kernel~$Q$.}.

We view the nonlinear dynamics as a dynamical system
$p\mapsto T^t(p)$, $t\in\bbN$, on $\Pdist(\OmegaG)$. Here $T^0(p)=p$ is the initial
distribution, $T^t(p)$ is the distribution after $t$ steps, and
\[
T^t(p)=T^{t-1}(p)\circ T^{t-1}(p).
\]
One step of the dynamics is defined by
\begin{equation}\label{eq:pre:one-step-map}
p\mapsto T(p)=p\circ p
:=
\sum_{\sigma,\sigma'\in\OmegaG}
p(\sigma)p(\sigma')Q(\cdot\given\sigma,\sigma').
\end{equation}
Equivalently, for every $\tau\in\OmegaG$,
\[
(p\circ p)(\tau)
=
\sum_{\sigma,\sigma'\in\OmegaG}
p(\sigma)p(\sigma')Q(\tau\given\sigma,\sigma')
=
\sum_{\sigma,\sigma',\tau'\in\OmegaG}
p(\sigma)p(\sigma')
\kernelQ(\sigma,\sigma';\tau,\tau').
\]

More generally, for $p,q\in\Pdist(\OmegaG)$, we define the commutative
collision product by
\begin{equation}\label{eq:pre:commutative-collision-product}
(p\circ q)(\tau)
=
\frac12
\sum_{\sigma,\sigma',\tau'\in\OmegaG}
\bigl(p(\sigma)q(\sigma')+p(\sigma')q(\sigma)\bigr)
\kernelQ(\sigma,\sigma';\tau,\tau').
\end{equation}
Thus $p\circ q=q\circ p$, and \eqref{eq:pre:one-step-map} is recovered by
taking $q=p$. Therefore, a collision kernel $\kernelQ$ induces a nonlinear dynamics $\{T^t(p)\}_{t\geq0}$ as we just described. We use the notation $T_{\kernelQ}$ to clarify the exchange kernel that induces the one step map $T$ when it is not clear from the context.

When $\kernelQ=\Qmf$, we call the corresponding nonlinear dynamics the mean-field dynamics, and denote its one step map by $\Tmf$.
Similarly, when $\kernelQ=\Qss$, we call the corresponding nonlinear dynamics the single-site dynamics, and denote its one step map by $\Tss$. More generally, we use $T_\cA$ for the multi-component mean-field collision $\kernelQ=\cQ^\cA$.

We will mainly work with the discrete-time dynamics defined above. For our entropy
decay result, we also use the associated continuous-time evolution given by
\begin{equation}\label{eq:pre:continuous-nonlinear-dynamics}
\frac{d p_t}{dt}
=
p_t\circ p_t-p_t,
\qquad p_0=p.
\end{equation}
The existence and uniqueness for the Cauchy problem \eqref{eq:pre:continuous-nonlinear-dynamics} can be established by standard methods, and the solution $(p_t)_{t\ge0}$ admits an explicit probabilistic construction in terms of random trees. We refer to \cite{caputo2025kac}
and references therein for more details.  
\subsection{Conservation laws}
In the dynamics defined above, the map $p \mapsto p \circ p$ has a conservation property. For any partition $\cA$, the multi-component 
mean-field dynamics preserves the density vector defined as follows. 
Given \(p\in\cP(\O_G)\), define its \emph{\(\cA\)-density vector} \(\vect\a_\cA(p)=(\a_A(p))_{A\in\cA}\), with components
\begin{equation}
    \a_A(p)=\frac{1}{|A|}\sum_{u\in A}\bbP_{\t\sim p}[\t_u=1].
\end{equation}
Recalling the notation from \Cref{subsec:ab}, in the mean-field case $\cA=\{V\}$ the density $\vect\a_\cA(p)$ coincides with the scalar $\a(p)$, while in the single site case $\cA=\{\{u\},\,u\in V\}$ the density $\vect\a_\cA(p)$ coincides with the vector $\vect m(p)$. 
% In particular, while the single-site dynamics
% preserves the marginal vector.

\begin{lemma}\label{lem:pre:conservation-laws}
Let $p,q\in\Pdist(\OmegaG)$, and fix a partition $\cA$ of $V$. The collision product associated to $\cQ=\cQ^\cA$ satisfies
\begin{equation}\label{eq:multi-mean-field-conservation}
    \vect\alpha_\cA(p\circ q)=\frac12\vect\alpha_\cA(p)+\frac12\vect\alpha_\cA(q).
    \end{equation}
In particular,
    \[
\vect\alpha_\cA(T_\cA(p))=\vect\alpha_\cA(p).
    \]
Specializing to  $\cQ=\Qmf$ and $\cQ=\Qss$ one has
\begin{equation}\label{eq:pre:mean-field-conservation}
    \alpha(\Tmf(p))=\alpha(p)\,,\qquad 
\vect m(\Tss(p))=\vect m(p).    \end{equation}
%     In particular,
%     \[
%     \alpha(\Tmf(p))=\alpha(p).
%     \]

%     \item For the collision product associated with the single-site kernel $\Qss$, for every $v\in V$,
%     \begin{equation}\label{eq:pre:single-site-conservation}
%     m_v(p\circ q)=\frac12 m_v(p)+\frac12 m_v(q).
%     \end{equation}
%     In particular,
%     \[
%     \vect m(\Tss(p))=\vect m(p).
%     \]
% \end{enumerate}
\end{lemma}

\begin{proof}
We use the following observation. Let $\kernelQ$ be a collision kernel satisfying
the mild exchange symmetry \eqref{eq:pre:mild-exchange-symmetry}, and let
$F:\OmegaG\to\mathbb R$ be such that, whenever
$\kernelQ(\sigma,\sigma';\tau,\tau')>0$,
\[
F(\tau)+F(\tau')=F(\sigma)+F(\sigma').
\]
Then
\begin{equation} \label{eqn:obs}
\E_{\tau\sim(p\circ q)}[F(\tau)]
=
\frac12\E_{\sigma\sim p}[F(\sigma)]
+
\frac12\E_{\sigma'\sim q}[F(\sigma')].
\end{equation}
Indeed, by the definition of the commutative collision product in
\eqref{eq:pre:commutative-collision-product},
\begin{align*}
\E_{\tau\sim(p\circ q)}[F(\tau)]
&=
\sum_{\tau\in\OmegaG}F(\tau)(p\circ q)(\tau) \\
&=
\frac12
\sum_{\sigma,\sigma',\tau,\tau'\in\OmegaG}
p(\sigma)q(\sigma')
\kernelQ(\sigma,\sigma';\tau,\tau')F(\tau) \\
&\qquad+
\frac12
\sum_{\sigma,\sigma',\tau,\tau'\in\OmegaG}
p(\sigma')q(\sigma)
\kernelQ(\sigma,\sigma';\tau,\tau')F(\tau).
\end{align*}
For the second term, exchanging the variables $\sigma$ and $\sigma'$, using
\eqref{eq:pre:mild-exchange-symmetry}, and then exchanging the variables $\tau$
and $\tau'$ gives
\[
\frac12
\sum_{\sigma,\sigma',\tau,\tau'\in\OmegaG}
p(\sigma)q(\sigma')
\kernelQ(\sigma,\sigma';\tau,\tau')F(\tau').
\]
Therefore,
\begin{align*}
\E_{\tau\sim(p\circ q)}[F(\tau)]
&=
\frac12
\sum_{\sigma,\sigma',\tau,\tau'\in\OmegaG}
p(\sigma)q(\sigma')
\kernelQ(\sigma,\sigma';\tau,\tau')
\bigl(F(\tau)+F(\tau')\bigr) \\
&=
\frac12
\sum_{\sigma,\sigma'\in\OmegaG}
p(\sigma)q(\sigma')
\bigl(F(\sigma)+F(\sigma')\bigr)
\sum_{\tau,\tau'\in\OmegaG}
\kernelQ(\sigma,\sigma';\tau,\tau') \\
&=
\frac12\E_{\sigma\sim p}[F(\sigma)]
+
\frac12\E_{\sigma'\sim q}[F(\sigma')].
\end{align*}

We now apply this observation to our kernel $\cQ=\cQ^\cA$. By construction, 
% write the two kernels. For the mean-field kernel
% $\Qmf$, every allowed exchange preserves the total number of occupied vertices
% in the pair, and if the exchange is not admissible, the pair is left unchanged.
%Hence, 
if $\cQ^\cA(\sigma,\sigma';\tau,\tau')>0$, then
\[
|\tau|_A+|\tau'|_A=|\sigma|_A+|\sigma'|_A\,,\qquad A\in \cA\,,
\]
where we use the notation $|\sigma|_A=\sum_{u\in A}\sigma_u$. Applying~\eqref{eqn:obs} with %$\kernelQ=\Qmf$ and 
$F(\sigma)=\frac1{|A|}|\sigma|_A$ for each $A\in\cA$, gives \eqref{eq:multi-mean-field-conservation}.
% \[
% \E_{\tau\sim(p\circ q)}[|\tau|]
% =
% \frac12\E_{\sigma\sim p}[|\sigma|]
% +
% \frac12\E_{\sigma'\sim q}[|\sigma'|].
% \]
% Dividing by $n$ proves \eqref{eq:pre:mean-field-conservation}. Taking
% $p=q$ implies $\alpha(\Tmf(p))=\alpha(p)$.
% For the single-site kernel $\Qss$, fix $v\in V$. The value at vertex $v$
% is exchanged only with the value at the same vertex in the other independent
% set. Therefore, whenever $\Qss(\sigma,\sigma';\tau,\tau')>0$,
% \[
% \tau_v+\tau'_v=\sigma_v+\sigma'_v.
% \]
% Applying~\eqref{eqn:obs} with $\kernelQ=\Qss$ and $F(\sigma)=\sigma_v$ gives
% \[
% m_v(p\circ q)=\frac12 m_v(p)+\frac12 m_v(q).
% \]
% This proves \eqref{eq:pre:single-site-conservation}. Taking $q=p$ gives
% $\vect m(\Tss(p))=\vect m(p)
% $.
\end{proof}

\subsection{Reversibility and Gibbs stationary measures}

Let $\kernelQ$ be a collision kernel on $\OmegaG\times\OmegaG$. We say that
a measure $\mu\in\Pdist(\OmegaG)$ is \emph{reversible} with respect to $\kernelQ$ if $\mu \otimes \mu$ satisfies the following detailed balance condition
\begin{equation}\label{eq:pre:detailed-balance}
\mu(\sigma)\mu(\sigma')
\kernelQ(\sigma,\sigma';\tau,\tau')
=
\mu(\tau)\mu(\tau')
\kernelQ(\tau,\tau';\sigma,\sigma')
\end{equation}
for all $\sigma,\sigma',\tau,\tau'\in\OmegaG$. Equivalently, we also say
that $\kernelQ$ is reversible with respect to $\mu$. In this case,
$\mu\otimes\mu$ is reversible for the Markov chain on the pair space
$\OmegaG\times\OmegaG$ with transition kernel $\kernelQ$.

We say that $\mu\in\Pdist(\OmegaG)$ is \emph{stationary} for the nonlinear dynamics
associated with $\kernelQ$ if
\[
\mu\circ\mu=\mu.
\]

If $\mu$ is reversible with respect to $\kernelQ$, then $\mu$ is stationary
for the nonlinear dynamics associated with $\kernelQ$. Indeed, reversibility
implies that $\mu\otimes\mu$ is stationary for the Markov chain on
$\OmegaG\times\OmegaG$ with transition kernel $\kernelQ$. Therefore, for
every $\tau\in\OmegaG$,
\[
(\mu\circ\mu)(\tau)
=
\sum_{\tau'\in\OmegaG} \sum_{\sigma, \sigma' \in \OmegaG}
\mu(\sigma)\mu(\sigma')
\kernelQ(\sigma,\sigma';\tau,\tau')
=
\sum_{\tau'\in\OmegaG}
\mu(\tau)\mu(\tau')
=
\mu(\tau).
\]
Thus $\mu\circ\mu=\mu$.\\

In the next lemma, we prove that the relevant hard-core Gibbs measures are stationary for our nonlinear dynamics.
For any given partition $\cA$ of $V$, we use the notation $\Gamma_\cA$ for the set of vector fugacities that are constant ever the components of the partition, namely
\[
\Gamma_\cA = \{\vect\lambda\in(0,\infty)^V:\; \lambda_u = \lambda_v \text{ for all } u,v\in A\,,  \text{ for all } A\in \cA
\}.
\]
If $\cA$ is a the partition $\cA=\{V\}$, then $\Gamma_\cA$ is the set of all scalars $\lambda>0$, while if $\cA=\{\{u\},\,u\in V\}$, then $\Gamma_\cA$ is the set of all fugacity vectors   $\vect\lambda\in(0,\infty)^V$.
\begin{lemma}\label{lem:pre:gibbs-stationary-measures}
For any fixed partition $\cA$ of $V$,
 $\vect\lambda\in\Gamma_\cA$, the measure $\mu_{G,\vect\lambda}$ is reversible with respect to $\cQ^\cA$.
    In particular, $\mu_{G,\vect\lambda}$ is stationary for the
 $\cA$-multi-component mean-field dynamics for all $\vect\lambda\in\Gamma_\cA$.
%  \begin{enumerate}[label=(\roman*)]
%     \item For every $\lambda>0$, the measure $\mu_{G,\lambda}$ is reversible
%     with respect to $\Qmf$. In particular, $\mu_{G,\lambda}$ is stationary
%     for the mean-field dynamics.

%     \item For every $\vect\lambda\in(0,\infty)^V$, the measure
%     $\mu_{G,\vect\lambda}$ is reversible with respect to $\Qss$.
%     In particular, $\mu_{G,\vect\lambda}$ is stationary for the
%     single-site dynamics.
% \end{enumerate}
% 
\end{lemma}

\begin{proof}
We first observe that, for every $v,u\in V$, the local exchange kernel is
symmetric on the pair space:
\[
Q_{v,u}(\sigma,\sigma';\tau,\tau')
=
Q_{v,u}(\tau,\tau';\sigma,\sigma').
\]
Note that the exchange is its own inverse:
$(\tau,\tau')=(\sigma,\sigma')^{v,u}$ if and only if
$(\sigma,\sigma')=(\tau,\tau')^{v,u}$. So, if either side of the equality is positive, the other side is also positive and has the same value, and hence the equality follows.
If $A\in\cA$, and $Q_{v,u}(\sigma,\sigma';\tau,\tau')>0$, $u,v\in A$, then
\(
|\sigma|_A+|\sigma'|_A=|\tau|_A+|\tau'|_A.
\)
Hence, for all $\vect\lambda\in\Gamma_\cA$, %for the uniform hard-core measure,
\[
\mu_{G,\vect\lambda}(\sigma)\mu_{G,\vect\lambda}(\sigma')
=
\mu_{G,\vect\lambda}(\tau)\mu_{G,\vect\lambda}(\tau').
\]
Together with the symmetry of $Q_{v,u}$, this gives
\[
\mu_{G,\vect\lambda}(\sigma)\mu_{G,\vect\lambda}(\sigma')
Q_{v,u}(\sigma,\sigma';\tau,\tau')
=
\mu_{G,\lambda}(\tau)\mu_{G,\vect\lambda}(\tau')
Q_{v,u}(\tau,\tau';\sigma,\sigma').
\]
By the definition of $\cQ^\cA$ in \eqref{eq:QA} we have
\begin{align*}
\mu_{G,\vect\lambda}(\sigma)\mu_{G,\vect\lambda}(\sigma')
\cQ^\cA(\sigma,\sigma';\tau,\tau')
&=
%\frac{1}{n^2}
%\sum_{v,u\in V}
%\mu_{G,\lambda}(\sigma)\mu_{G,\lambda}(\sigma')
%Q_{v,u}(\sigma,\sigma';\tau,\tau') \\
%&=
%\frac{1}{n^2}
%\sum_{v,u\in V}
%\mu_{G,\lambda}(\tau)\mu_{G,\lambda}(\tau')
%Q_{v,u}%(\tau,\tau';\sigma,\sigma') \\
%&=
\mu_{G,\vect\lambda}(\tau)\mu_{G,\vect\lambda}(\tau')
\cQ^\cA(\tau,\tau';\sigma,\sigma').
\end{align*}
This proves that $\mu_{G,\vect\lambda}$ is reversible with respect to $\cQ^\cA$, for all $\vect\lambda\in\Gamma_\cA$, and so it is stationary.
% For the single-site kernel, if $Q_{v,v}(\sigma,\sigma';\tau,\tau')>0$, then
% the total occupation at every vertex is preserved across the pair:
% \[
% \sigma_x+\sigma'_x=\tau_x+\tau'_x,
% \qquad x\in V.
% \]
% Therefore, for any fugacity vector $\vect\lambda$, %for the nonuniform hard-core measure,
% \[
% \mu_{G,\vect\lambda}(\sigma)
% \mu_{G,\vect\lambda}(\sigma')
% =
% \mu_{G,\vect\lambda}(\tau)
% \mu_{G,\vect\lambda}(\tau').
% \]
% Together with the symmetry of $Q_{v,v}$, this gives
% \[
% \mu_{G,\vect\lambda}(\sigma)
% \mu_{G,\vect\lambda}(\sigma')
% Q_{v,v}(\sigma,\sigma';\tau,\tau')
% =
% \mu_{G,\vect\lambda}(\tau)
% \mu_{G,\vect\lambda}(\tau')
% Q_{v,v}(\tau,\tau';\sigma,\sigma').
% \]
% By the definition of $\Qss$ in \eqref{eq:pre:single-site-collision-kernel}, we have
% \begin{align*}
% &\mu_{G,\vect\lambda}(\sigma)
% \mu_{G,\vect\lambda}(\sigma')
% \Qss(\sigma,\sigma';\tau,\tau') %\\
% %&\qquad =
% % \frac1n
% % \sum_{v\in V}
% % \mu_{G,\vect\lambda}(\sigma)
% % \mu_{G,\vect\lambda}(\sigma')
% % Q_{v,v}(\sigma,\sigma';\tau,\tau') \\
% % &\qquad =
% % \frac1n
% % \sum_{v\in V}
% % \mu_{G,\vect\lambda}(\tau)
% % \mu_{G,\vect\lambda}(\tau')
% % Q_{v,v}(\tau,\tau';\sigma,\sigma') \\
% % &\qquad =
% = \mu_{G,\vect\lambda}(\tau)
% \mu_{G,\vect\lambda}(\tau')
% \Qss(\tau,\tau';\sigma,\sigma').
% \end{align*}
% This proves that $\mu_{G,\vect\lambda}$ is reversible with respect to
% $\Qss$, and hence it is stationary for the single-site dynamics.
\end{proof}

\section{Qualitative Convergence}

\subsection{A general convergence criterion}

We first state a general convergence criterion for nonlinear dynamics of
the form defined in the previous section. Let $\kernelQ$ be an arbitrary probability kernel on
$\OmegaG\times\OmegaG$. We interpret $\kernelQ$ as a generic  collision kernel, and let $T$ be the corresponding one step map from \eqref{eq:pre:one-step-map}.

We say that $\kernelQ$ is \emph{balanced} with respect to
$\mu\in\Ppos(\OmegaG)$ if $\mu\otimes\mu$ is stationary for the Markov chain
on $\OmegaG\times\OmegaG$ with transition kernel $\kernelQ$, that is,
\begin{equation}\label{eqn:balanced}
\sum_{\sigma,\sigma'\in\OmegaG}
\mu(\sigma)\mu(\sigma')
\kernelQ(\sigma,\sigma';\tau,\tau')
=
\mu(\tau)\mu(\tau')
\qquad
\forall \tau,\tau'\in\OmegaG .
\end{equation}
In particular, if $\kernelQ$ is reversible with respect to $\mu$, then
$\kernelQ$ is balanced with respect to $\mu$.

We say that $p\in\Pdist(\OmegaG)$ is \emph{irreducible} for the nonlinear
dynamics induced by $\kernelQ$ if there exist $t_0\ge 0$ and $\eps>0$ such
that
\[
T^t(p)(\sigma)\ge \eps
\qquad
\forall t\ge t_0,\ \forall \sigma\in\OmegaG .
\]

We will use the following convergence criterion from
\cite[Theorem~2.8]{caputo2024nonlineardynamicsisingmodel}. 

\begin{theorem}\label{thm:qual:general-convergence}
Let $\kernelQ$ be a collision kernel on $\OmegaG\times\OmegaG$. Assume that
$\kernelQ$ is balanced with respect to some $\mu\in\Ppos(\OmegaG)$, and that
$\kernelQ$ satisfies
\begin{equation}\label{eq:qual:strong-pair-symmetry}
\kernelQ(\sigma,\sigma';\tau,\tau')
=
\kernelQ(\sigma,\sigma';\tau',\tau)
=
\kernelQ(\sigma',\sigma;\tau,\tau')
\end{equation}
for all $\sigma,\sigma',\tau,\tau'\in\OmegaG$. Assume also that
\begin{equation}\label{eq:qual:positive-diagonal}
\kernelQ(\sigma,\sigma';\sigma,\sigma')>0
\qquad
\forall \sigma,\sigma'\in\OmegaG .
\end{equation}
Then, for every irreducible initial distribution $p\in\Pdist(\OmegaG)$,
the sequence $\{T^t(p)\}_{t\ge 0}$ converges %in total variation distance 
to a measure
$\nu\in\Ppos(\OmegaG)$ which is stationary for the nonlinear dynamics.
\end{theorem}
We shall apply the above criterion to prove a convergence theorem for the multi-component mean-field dynamics defined in \eqref{eq:QA}.

Our exchange kernels satisfy the mild exchange symmetry
\eqref{eq:pre:mild-exchange-symmetry}, rather than the stronger symmetry
\eqref{eq:qual:strong-pair-symmetry}. However, this is enough since we can replace our 
kernel $\kernelQ=\cQ^\cA$ by the symmetrized kernel
\[
\overline{\kernelQ}(\sigma,\sigma';\tau,\tau')
:=
\frac12
\left(
\kernelQ(\sigma,\sigma';\tau,\tau')
+
\kernelQ(\sigma',\sigma;\tau,\tau')
\right).
\]
Using 
\eqref{eq:pre:mild-exchange-symmetry}, one can easily show that
$\overline{\kernelQ}$ satisfies \eqref{eq:qual:strong-pair-symmetry}.
% Indeed, if
% $\kernelQ$ is balanced with respect to $\mu$, then for every
% $\tau,\tau'\in\OmegaG$,
% \begin{align*}
% &\sum_{\sigma,\sigma'\in\OmegaG}
% \mu(\sigma)\mu(\sigma')
% \overline{\kernelQ}(\sigma,\sigma';\tau,\tau') 
% %\\
% %&\qquad 
% =
% % \frac12
% % \sum_{\sigma,\sigma'\in\OmegaG}
% % \mu(\sigma)\mu(\sigma')
% % \kernelQ(\sigma,\sigma';\tau,\tau') \\
% % &%\qquad
% % \quad+
% \frac12
% \sum_{\sigma,\sigma'\in\OmegaG}
% \mu(\sigma)\mu(\sigma')\left[
% \kernelQ(\sigma',\sigma;\tau,\tau') %\\
% %&\qquad =
% %\sum_{\sigma,\sigma'\in\OmegaG}
% %\mu(\sigma)\mu(\sigma')
% + \kernelQ(\sigma,\sigma';\tau,\tau')\right]
% %\\& \quad 
% =
% \mu(\tau)\mu(\tau').
% \end{align*}
% % where in the third equality we exchange the variables $\sigma$ and
% % $\sigma'$ in the second sum. 
Moreover, $\kernelQ$ and $\overline{\kernelQ}$ induce the same nonlinear map:
\[
T_{\overline{\kernelQ}}(p)=T_{\kernelQ}(p),
\qquad p\in\Pdist(\OmegaG).
\]
Indeed, for every $\tau\in\OmegaG$,
\begin{align*}
T_{\overline{\kernelQ}}(p)(\tau)
&=
\sum_{\sigma,\sigma',\tau'\in\OmegaG}
p(\sigma)p(\sigma')
\overline{\kernelQ}(\sigma,\sigma';\tau,\tau') \\
&=
\frac12
\sum_{\sigma,\sigma',\tau'\in\OmegaG}
p(\sigma)p(\sigma')
\left[\kernelQ(\sigma,\sigma';\tau,\tau') +
\kernelQ(\sigma',\sigma;\tau,\tau')\right] \\
&=
\sum_{\sigma,\sigma',\tau'\in\OmegaG}
p(\sigma)p(\sigma')
\kernelQ(\sigma,\sigma';\tau,\tau')
=
T_{\kernelQ}(p)(\tau).
\end{align*}
% where in the third equality we exchange the variables $\sigma$ and
% $\sigma'$ in the second sum. 
Similarly, one has that $\overline{\kernelQ}$ is balanced with respect to $\mu$ whenever $\kernelQ$ is. Also, if $\kernelQ$ satisfies the positive
diagonal condition \eqref{eq:qual:positive-diagonal}, then so does
$\overline{\kernelQ}$.
Thus we may apply
\Cref{thm:qual:general-convergence} to the symmetrized kernel without
changing the nonlinear dynamics.

\subsection{Stationary distribution and positive diagonal}

In order to apply \Cref{thm:qual:general-convergence} to our dynamics, we  show that our exchange kernels satisfy the conditions of this theorem. We first check the positive diagonal condition in
\eqref{eq:qual:positive-diagonal}.

\begin{lemma}\label{lem:qual:positive-diagonal}
Assume that $E\ne\emptyset$. Then, for any partition $\cA$,  $\cQ=\cQ^\cA$ satisfies 
% both $\Qmf$ and $\Qss$ satisfy 
the positive
diagonal condition:
\[
\cQ^\cA(\sigma,\sigma';\sigma,\sigma')>0,
% \qquad
% \Qss(\sigma,\sigma';\sigma,\sigma')>0
\]
for all $\sigma,\sigma'\in\OmegaG$. Consequently, the corresponding
symmetrized kernels also satisfy the positive diagonal condition.
\end{lemma}

\begin{proof}
%We first prove the positive diagonal condition for the single-site kernel.
Fix $\sigma,\sigma'\in\OmegaG$. If there exists $v\in V$ such that
$\sigma_v=\sigma'_v$, then the single-site exchange at $v$ leaves the pair
unchanged. Hence
\(
Q_{v,v}(\sigma,\sigma';\sigma,\sigma')=1.
\)
Otherwise, $\sigma_v\ne\sigma'_v$ for every $v\in V$, so the two independent
sets are complements of each other. Since $E\ne\emptyset$, choose an edge
$\{v,u\}\in E$. As both $\sigma$ and $\sigma'$ are independent sets, one
endpoint of this edge belongs to $\sigma$ and the other belongs to $\sigma'$.
Assume, without loss of generality, that $v\in\sigma$ and $u\in\sigma'$.
Then exchanging the value at vertex $v$ would put $v$ into the second
independent set, which already contains $u$. This would violate independence.
Thus the exchange at $v$ is not admissible, so the pair is left unchanged. 
Therefore, in either case we have shown that 
there exists $v\in V$ such that
\[
Q_{v,v}(\sigma,\sigma';\sigma,\sigma')=1.
\]
From \eqref{eq:QA} it follows that 
$\cQ^\cA(\sigma,\sigma';\sigma,\sigma')\ge \frac1{n^2}$. 
% and
% The mean-field case follows from the same argument, since the mean-field
% kernel includes the single-site exchanges. More precisely, in both cases
% above there exists $v\in V$ such that
% \[
% Q_{v,v}(\sigma,\sigma';\sigma,\sigma')=1.
% \]
% Therefore,
% \(
% \Qmf(\sigma,\sigma';\sigma,\sigma')
% \ge
% \frac1{n^2}\).
% \sum_{u,w\in V}
% Q_{u,w}(\sigma,\sigma';\sigma,\sigma')
% \ge \frac1{n^2}>0.
% \]
% Finally, if a kernel $\kernelQ$ satisfies the positive diagonal condition,
% then its symmetrization satisfies it as well, since
% \[
% \overline{\kernelQ}(\sigma,\sigma';\sigma,\sigma')
% \ge
% \frac12\kernelQ(\sigma,\sigma';\sigma,\sigma')>0.
% \]
This completes the proof.
\end{proof}

Another condition in \Cref{thm:qual:general-convergence} is the balance property~\eqref{eqn:balanced}.
By \Cref{lem:pre:gibbs-stationary-measures}, the kernel
$\cQ^\cA$ is reversible with respect to %$\mu_{G,\lambda}$ for every
%$\lambda>0$, and $\Qss$ is reversible with respect to
$\mu_{G,\vect\lambda}$, for every
$\vect\lambda\in\Gamma_\cA$. Hence, these kernels are balanced with
respect to the corresponding Gibbs measures. 
% Moreover, by the symmetrization
% argument in the previous subsection, the same balancedness property holds for
% the corresponding symmetrized kernels.

The next lemma identifies the stationary distributions with full support.
%of the two nonlinear dynamics. 
This will allow us, after proving irreducibility, to
identify the limit given by \Cref{thm:qual:general-convergence} as the
hard-core Gibbs measure selected by the conserved quantity.

\begin{lemma}\label{lem:qual:positive-stationary-states}
Assume that $E \neq \emptyset$ and fix a partition $\cA$ of $V$. If $p\in\Ppos(\OmegaG)$ is stationary for $\cQ^\cA$,
    then $p=\mu_{G,\vect\lambda}$ for some    $\vect\lambda\in\Gamma_\cA$.
% \begin{enumerate}[label=(\roman*)]
%     \item If $p\in\Ppos(\OmegaG)$ is stationary for the mean-field dynamics,
%     then $p=\mu_{G,\lambda}$ for some $\lambda>0$.

%     \item If $p\in\Ppos(\OmegaG)$ is stationary for the single-site dynamics,
%     then $p=\mu_{G,\vect\lambda}$ for some
%     $\vect\lambda\in(0,\infty)^V$.
% \end{enumerate}
\end{lemma}

\begin{proof}
% We first state a consequence of stationarity that will be used in both
% cases. 
Let $\overline{\kernelQ}$ denote the symmetrized
kernel corresponding to $\cQ^\cA$.
% one of $\Qmf$ or $\Qss$. 
Since
symmetrization does not change the nonlinear map, a stationary distribution for
the original dynamics is also stationary for the symmetrized dynamics. By the symmetrization argument above and by
\Cref{lem:qual:positive-diagonal}, the symmetrized kernel satisfies the
symmetry and positive diagonal assumptions of
\cite[Proposition~2.10]{caputo2024nonlineardynamicsisingmodel}. %Moreover,
%balancedness is preserved under symmetrization. 
Moreover, the symmetrized kernel is reversible with respect to the
corresponding Gibbs measures. Indeed, the original kernel is reversible by
\Cref{lem:pre:gibbs-stationary-measures}, and the exchange symmetry
\eqref{eq:pre:mild-exchange-symmetry} implies that the same detailed balance
identity holds after symmetrization.
Thus, by \cite[Proposition~2.10]{caputo2024nonlineardynamicsisingmodel},
%applied to the corresponding symmetrized kernel, 
$p\otimes p$ is reversible
with respect to $\overline{\kernelQ}$ 
%that symmetrized kernel 
and it satisfies the detailed balance condition \eqref{eq:pre:detailed-balance}, whenever $p$ is stationary. 

Whenever an admissible exchange sends
$(\sigma,\sigma')$ to $(\tau,\tau')$, the symmetrized kernel satisfies
\[
\overline{\kernelQ}(\sigma,\sigma';\tau,\tau')
=
\overline{\kernelQ}(\tau,\tau';\sigma,\sigma')>0.
\]
% The positivity follows because the original kernel assigns positive
% probability to this admissible exchange, and the equality follows from the
% mild exchange symmetry of the kernel.
Since $p\otimes p$ is reversible with respect to $\overline{\kernelQ}$, when \(\overline{\kernelQ}(\s,\s';\t,\t')>0\) we must also have
% \[
% p(\sigma)p(\sigma')\overline{\kernelQ}(\sigma,\sigma';\tau,\tau')
% =
% p(\tau)p(\tau')\overline{\kernelQ}(\tau,\tau';\sigma,\sigma').
% \]
% Cancelling the common positive transition probability gives
\begin{equation}\label{eq:qual:stationary-multiplicative-identity}
    p(\sigma)p(\sigma') = p(\tau)p(\tau').
\end{equation}

We now show how \eqref{eq:qual:stationary-multiplicative-identity} implies the desired claim. 
% Now, we prove the statement for the single-site dynamics. 
For each
$v\in V$, define
\[
\lambda_v:=\frac{p(\{v\})}{p(\emptyset)}.
\]
This is well-defined since $p\in\Ppos(\OmegaG)$.

Let $\sigma\in\OmegaG$, and write $\sigma=\{v_1,\ldots,v_k\}$. Define
\[
\sigma^{(0)}:=\sigma,
\qquad
\sigma^{(i)}:=\sigma^{(i-1)}\setminus\{v_i\},
\qquad 1\le i\le k.
\]
For each $i=0,\ldots,k-1$, the single-site exchange at $v_{i+1}$ is
admissible and
\[
\Phi_{v_{i+1},v_{i+1}}(\sigma^{(i)},\emptyset)
=
(\sigma^{(i+1)},\{v_{i+1}\}),
\] 
where we use the notation from \eqref{eq:phi_def}. Applying
\eqref{eq:qual:stationary-multiplicative-identity}, we get
\[
p(\sigma^{(i)})p(\emptyset)
=
p(\sigma^{(i+1)})p(\{v_{i+1}\}).
\]
Multiplying these identities gives
\[
p(\sigma)
=
p(\emptyset)\prod_{v\in\sigma}\frac{p(\{v\})}{p(\emptyset)}
=
p(\emptyset)\prod_{v\in\sigma}\lambda_v.
\]
Since $p$ is a probability measure,
\[
1
=
\sum_{\sigma\in\OmegaG}p(\sigma)
=
p(\emptyset)
\sum_{\sigma\in\OmegaG}\prod_{v\in\sigma}\lambda_v
=
p(\emptyset)Z_G(\vect\lambda).
\]
Therefore $p(\emptyset)=1/Z_G(\vect\lambda)$, and hence
\[
p(\sigma)
=
\frac{1}{Z_G(\vect\lambda)}
\prod_{v\in\sigma}\lambda_v
=
\mu_{G,\vect\lambda}(\sigma).
\]

It remains to prove that $\vect\lambda\in\Gamma_\cA$.
% This proves the single-site statement.
% %We now prove the mean-field statement. 
% The argument above also applies to
% the mean-field dynamics, since the mean-field kernel includes all
% single-site exchanges. Thus there exist positive numbers
% $(\lambda_v)_{v\in V}$ such that
% \[
% p(\sigma)
% =
% p(\emptyset)\prod_{v\in\sigma}\lambda_v.
% \]
% It remains to show that the numbers $\lambda_v$ are all equal. Fix
% $v,u\in V$. For the mean-field exchange, note that
We note that
\[
\Phi_{v,u}(\{v\},\emptyset)
=
(\emptyset,\{u\}).
\]
Therefore, if $u,v\in A$, for some $A\in\cA$, by
\eqref{eq:qual:stationary-multiplicative-identity},
\[
p(\{v\})p(\emptyset)
=
p(\emptyset)p(\{u\}).
\]
Since $p(\emptyset)>0$, we get $p(\{v\})=p(\{u\})$. Hence
 $\lambda_v=\lambda_u$ for all $u,v\in A$, for all $A\in\cA$. %Writing this common value as
% $\lambda>0$, we obtain
% \[
% p(\sigma)
% =
% p(\emptyset)\lambda^{|\sigma|}.
% \]
% Normalizing gives $p(\emptyset)=1/Z_G(\lambda)$, and therefore
% $p=\mu_{G,\lambda}$. This proves the mean-field statement.
\end{proof}

%\subsection{Convergence of the mean-field dynamics}
\subsection{Proof of irreducibility}
The remaining condition to be checked in
order to apply \Cref{thm:qual:general-convergence} is irreducibility. The following lemma
verifies this condition for all multi-component mean-field dynamics.
\begin{lemma}\label{lem:multi:irr}
Fix a partition $\cA$ of $V$. Let \(p\in\cP(\O_G)\) have \(\cA\)-density vector $\vect\alpha_\cA(p)$ satisfying
\[
0<\a_A(p)<\frac{1}{\D+1}\qquad\text{for all } A\in\cA.
\]
Then \(p\) is irreducible with respect to the  dynamics with kernel \(\cQ^\cA\). More precisely, there exists \(\eps=\eps(n,\D,\vect\a_\cA)>0\) such that
\[
T^t(p)(\eta)\ge \eps
\qquad
\text{for all } t\ge 2n-1 \text{ and all } \eta\in\O_G .
\]
% In particular:
% \begin{itemize}
%     \item if \(p\) has density \(\a=\a(p)\) with \(0<\a<\frac{1}{\D+1}\), then \(p\) is irreducible for the mean-field dynamics;
%     \item if \(p\) has marginals \(0<m_v(p)<\frac{1}{\D+1}\) for all \(v\in V\), then \(p\) is irreducible for the single-site dynamics.
% \end{itemize}
\end{lemma}
% In particular, the previous condition on the $\cA$-density vector \(\vect\a_\cA(p)\) is the natural extension to the multi-mean-field case and it returns the original conditions for single-site or mean field.

\begin{proof}
We first show that there is a positive probability of reaching the empty configuration after $n-1$ steps of the  evolution, and then show that, starting from the empty configuration, any prescribed configuration can be reached with positive probability in $n$ further steps. Let \(\delta,\g>0\) be such that
\[ \gamma\le \min_{A\in \cA} \a_A(p) \qquad\text{and}\qquad \max_{A\in \cA} \a_A(p) \le \frac{1-\delta}{\Delta+1}. \]
\medskip

\noindent\textbf{Step 1.} For \(A\in\cA\), denote by \(u[A]\in V\) a uniformly distributed random vertex in \(A\) and we write \(\bbE_{u[A]}\) for the expectation with respect to \(u[A]\). We claim that for every \(t\ge 0\) and all \(A\in\cA\),
\begin{equation}
\bbE_{u[A]}\bigl[\bbP_{\tau\sim T^t(p)}(\tau_{N^+(u[A])}=0)\bigr]\ge\delta.
\end{equation}
Indeed, since \(\alpha_A(T^t(p))=\alpha_A(p)\) by the conservation law, we have
\begin{align}
&\bbE_{u[A]}\bigl[\bbP_{\tau\sim T^t(p)}(\exists\, w\in N^+(u[A]) : \tau_w=1)\bigr]
\\&\qquad \le\frac{1}{|A|}\sum_{v\in A}\sum_{w\in N^+(v)}\bbP_{\tau\sim T^t(p)}(\tau_w=1)
\le(\Delta+1)\,\alpha_A(p)\le 1-\delta,
\end{align}
and the claim follows by taking complements.
\medskip

\noindent\textbf{Step 2.} For every \(m\in\{1,\dots,n\}\), let \(A_m\) be a uniformly random subset of \(V\) of cardinality \(m\), and write \(\bbE_{A_m}\) for the expectation with respect to \(A_m\). For \(B\subseteq V\), we write \(\t_B=(\t_v)_{v\in B}\). We claim that for all \(m\in\{1,\dots,n\}\),
\begin{equation}\label{eq:0irr}
\bbE_{A_m}\bigl[\bbP_{\tau\sim T^{m-1}(p)}(\tau_{A_m}=0)\bigr]
\ge\Bigl(\frac{\delta}{n}\Bigr)^{m-1}m!\,\frac{\Delta+\delta}{\Delta+1}.
\end{equation}
In particular, since \(A_n=V\), taking \(m=n\) gives
\begin{equation}\label{eq:0irr-full}
\bbP_{\tau\sim T^{n-1}(p)}(\tau=0)
\ge\Bigl(\frac{\delta}{n}\Bigr)^{n-1}n!\,\frac{\Delta+\delta}{\Delta+1}.
\end{equation}
We prove~\eqref{eq:0irr} by induction on \(m\). For \(m=1\),
\begin{align*}
\bbE_{A_1}\bigl[\bbP_{\tau\sim p}(\tau_{A_1}=0)\bigr]
&=1-\frac{1}{n}\sum_{v\in V}m_v(p)=1-\frac{1}{n}\sum_{A\in\cA}\sum_{v\in A}m_v(p)\\
&=1-\frac{1}{n}\sum_{A\in\cA}|A|\,\a_A(p)\ge\frac{\Delta+\delta}{\Delta+1},
\end{align*}
which is~\eqref{eq:0irr} for \(m=1\).

\noindent Now let \(\tau\sim T(p)=p\circ p\), and let \(\sigma,\sigma'\sim p\) be independent. By definition of the kernel \(\cQ^{\cK_\cA}\), the configuration \(\tau\) is obtained from \(\sigma\) by replacing \(\sigma_v\) with \(\sigma'_w\), where \(v\) is chosen uniformly at random in \(V\) and, given \(v\), \(w\) is chosen uniformly at random in \(A(v)\). The replacement is performed only if it is admissible; otherwise, \(\tau=\sigma\).

\medskip

\noindent Fix \(B\subseteq V\). If \(v\in B\) and \(w\in A(v)\), then the event \(\{\tau_B=0\}\) contains \(\{\sigma_{B\setminus\{v\}}=0\}\cap\{\sigma'_{N^+(w)}=0\}\). Since, given \(v\), the vertex \(w\) is distributed as \(u[A(v)]\) (see Step~1), averaging over \(v\) and \(w\) yields
\begin{equation}
\bbP_{\tau\sim T(p)}(\tau_B=0)\ge\frac{1}{n}\sum_{v\in B}
\bbP_{\sigma\sim p}(\sigma_{B\setminus\{v\}}=0)\,
\bbE_{u[A(v)]}\bigl[\bbP_{\sigma'\sim p}(\sigma'_{N^+(u[A(v)])}=0)\bigr].
\end{equation}
By Step~1,
\begin{equation}\label{eq:1irr}
\bbP_{\tau\sim T(p)}(\tau_B=0)\ge\frac{\delta}{n}\sum_{v\in B}
\bbP_{\sigma\sim p}(\sigma_{B\setminus\{v\}}=0).
\end{equation}

\medskip

\noindent We now combine~\eqref{eq:1irr} with the inductive hypothesis. Fix \(1\le m<n\) and assume that~\eqref{eq:0irr} holds at level \(m\). Let \(q=T^{m-1}(p)\); by the conservation law, \(\vect\a_\cA(q)=\vect\a_\cA(p)\). Let \(\tau\sim T(q)=T^m(p)\). Applying~\eqref{eq:1irr} with \(q\) in place of \(p\) and with \(B=A_{m+1}\), and then taking expectations, we obtain
\[
\bbE_{A_{m+1}}\bigl[\bbP_{\tau\sim T^m(p)}(\tau_{A_{m+1}}=0)\bigr]
\ge\delta\,\frac{m+1}{n}\,
\bbE_{A_{m+1}}\!\left[\frac{1}{m+1}\sum_{v\in A_{m+1}}
\bbP_{\sigma\sim T^{m-1}(p)}(\sigma_{A_{m+1}\setminus\{v\}}=0)\right].
\]
Removing a uniformly random element from a uniformly random subset of size \(m+1\) yields a uniformly random subset of size \(m\). Hence
\[
\bbE_{A_{m+1}}\bigl[\bbP_{\tau\sim T^m(p)}(\tau_{A_{m+1}}=0)\bigr]
\ge\frac{\delta}{n}(m+1)\,
\bbE_{A_m}\bigl[\bbP_{\sigma\sim T^{m-1}(p)}(\sigma_{A_m}=0)\bigr].
\]
By the inductive hypothesis, this is at least
\[
\frac{\delta}{n}(m+1)\Bigl(\frac{\delta}{n}\Bigr)^{m-1}m!\,\frac{\Delta+\delta}{\Delta+1}
=\Bigl(\frac{\delta}{n}\Bigr)^m(m+1)!\,\frac{\Delta+\delta}{\Delta+1}.
\]
This is~\eqref{eq:0irr} at level \(m+1\), which completes the induction.
\medskip

\noindent \textbf{Step 3.} For all \(\eta\in\O_G\),
\begin{equation}\label{eq:step3}
T^{2n-1}(p)(\eta) > n! \left(\frac{\g\delta}{n^2}\right)^n\frac{\Delta + \delta}{\Delta + 1}.
\end{equation}
We first show that for every \(t\ge 1\) and all \(\eta\in\O_G\),
\begin{equation}\label{eq:term1}
T^t(p)(S_u(\eta,1))
\ge
\frac{\g}{n} T^{t-1}(p)(\eta),
\qquad
\forall u\in V:\ S_u(\eta,1)\in\OmegaG,\ \eta_u=0.
\end{equation}
To see this, let \(\s,\s'\) be independent samples from \(T^{t-1}(p)\), let \(v_1\) be chosen uniformly at random in \(V\) and, given \(v_1\), let \(v_2\) be chosen uniformly at random in \(A(v_1)\). Then \(\tau=S_{v_1}(\s,\s'_{v_2})\) (with \(\tau=\s\) if this configuration is not admissible) has distribution \(T^t(p)\). The event \(\{\tau=S_u(\eta,1)\}\) contains the event \(\{v_1=u,\ \s=\eta,\ \s'_{v_2}=1\}\), whose probability is \(\frac{1}{n}\,T^{t-1}(p)(\eta)\,\a_{A(u)}(p)\) by the conservation law. Combined with \(\a_A(p)\geq \g\) for all \(A\in\cA\), this gives~\eqref{eq:term1}.

Iterating~\eqref{eq:term1}, we obtain, for every \(\eta\in\O_G\),
   \[
   T^{2n-1}(p)(\eta)\ge \left(\frac{\g}{n}\right)^{|\eta|}T^{2n-1-|\eta|}(p)(0).
   \]
Indeed, removing the vertices of \(\eta\) one at a time, each removal costs a factor of at least \(\g/n\), and after \(|\eta|\) steps we reach the empty configuration. On the other hand, since \(\vect\a_\cA(T^s(p))=\vect\a_\cA(p)\) for all \(s\ge0\), applying the Step~2 with \(p\) replaced by $T^{n-|\eta|}(p)$
\[
T^{2n-1-|\eta|}(p)(0)=T^{n-1}(T^{n-|\eta|}(p))(0)\ge \Bigl(\frac{\delta}{n}\Bigr)^{n-1} n!\, \frac{\Delta + \delta}{\Delta + 1}.
\] 
Combining the previous inequalities, we obtain
\[
T^{2n-1}(p)(\eta)
\ge
\left(\frac{\g}{n}\right)^{|\eta|}
\Bigl(\frac{\delta}{n}\Bigr)^{n-1}
n!\, \frac{\Delta+\delta}{\Delta+1}
>
n!\left(\frac{\g\delta}{n^2}\right)^n
\frac{\Delta+\delta}{\Delta+1},
\]
which proves~\eqref{eq:step3}.
\medskip

To prove the lemma, we extend the bound \eqref{eq:step3}
 to all times \(t\ge 2n-1\). Set
$s=t-(2n-1)$, and note that by the  conservation law,
$\vect \a_\cA(T^s(p))=\vect \a_\cA (p)$.
Hence \(T^s(p)\) satisfies the same assumptions as \(p\). Applying
\eqref{eq:step3} with \(T^s(p)\) in place of \(p\), we obtain
\begin{equation}\label{eq:tsemi}
T^t(p)(\eta)
=
T^{2n-1}(T^s(p))(\eta)
\ge
n!\left(\frac{\g\delta}{n^2}\right)^n
\frac{\Delta+\delta}{\Delta+1}.
\end{equation}
Thus the irreducibility condition  for distribution \(p\)  holds with
\[
\eps
:=
n!\left(\frac{\delta\g}{n^2}\right)^{n-1}
\frac{\D+\delta}{\D+1}. \qedhere
\]
\end{proof}
\subsection{Proof of \Cref{thm:intro-convergence}}
The statement is an immediate consequence of the following more general result, which holds for any partition $\cA$ of $V$. Specializing to the two extreme cases corresponding to single-site and mean-field respectively, yields \Cref{thm:intro-convergence}.
\begin{theorem}
\label{th:multi-mean-field-convergence}
For any partition $\cA$ of $V$, for any $p\in\cP(\O_G)$ with $\cA$-density vector $\vect \alpha_\cA(p)$ satisfying 
\[
0 < \alpha_A(p) < \frac{1}{\Delta+1}\,,\qquad  A \in \cA,
\] one has the convergence
\begin{align*}
\lim_{t \to \infty} T_\cA^t(p) = \mu_{G,\vect\lambda},
\end{align*}
where $\vect\lambda$ is the
unique fugacity vector $\vect\lambda\in\Gamma_\cA$ such that $\vect\alpha_\cA(\mu_{G,\vect\lambda})=\vect\alpha_\cA(p)$.
\end{theorem}
\begin{proof}
Let $p\in\Pdist(\OmegaG)$ have $\cA$-density vector $\vect\alpha_\cA(p)=(\alpha_A(p),\,A\in\cA)$ as in the statement. If $E=\emptyset$, then $\OmegaG=\{0,1\}^V$ and every proposed exchange is admissible. In this case, for any distribution $p$, one step of the $\cA$-multi-component mean-field dynamics has the following law: sample $\sigma\sim p$, choose a vertex $v$ uniformly at random, and replace $\sigma_v$ with an independent Bernoulli random variable with parameter $\alpha_{A(v)}(p)$.
This is one step of an ordinary (linear) Markov chain equivalent to the Glauber dynamics with respect to the the product of Bernoulli random variables ${\rm Be}(\alpha_A(p))$, $A\in\cA$. 
%whose evolution is determined by the parameter vector $\vect\alpha_\cA(p)$. Since $\alpha_A(T_\cA^{t}(p))=\alpha_A(p)$ for all $t\ge0$ and $A\in\cA$, the evolution $\{T_\cA^{t}(p)\}_{t\ge0}$ coincides with that of this linear chain started from $p$ and driven by $\vect\alpha_\cA(p)$.
The chain thus converges to this product measure, which  coincides with $\mu_{G,\vect\lambda}$, with fugacities $\lambda_v=\frac{\alpha_{A(v)}(p)}{1-\alpha_{A(v)}(p)}$, $v\in V$. Since $\vect\lambda$ is constant on each block of $\cA$, we have $\vect\lambda\in\G_\cA$, and the statement holds when $E=\emptyset$.

\medskip
\noindent We may therefore assume that $E\neq\emptyset$. We apply \Cref{thm:qual:general-convergence} to the symmetrized kernel $\overline{\cQ^\cA}$. By construction, $\overline{\cQ^\cA}$ satisfies the symmetry assumption \eqref{eq:qual:strong-pair-symmetry}. Moreover, the collision kernel $\cQ^\cA$ is balanced with respect to $\mu_{G,\vect\lambda_0}$ for every $\vect\lambda_0\in\G_\cA$ (\Cref{lem:pre:gibbs-stationary-measures}) and satisfies the positive diagonal condition (\Cref{lem:qual:positive-diagonal}). By \Cref{lem:multi:irr}, $p$ is irreducible for the $\cA$-multi-component dynamics. It follows that $\overline{\cQ^\cA}$ and $p$ have all the required properties. Hence, by \Cref{thm:qual:general-convergence}, there exists a stationary distribution $\nu\in\Ppos(\OmegaG)$ such that $T_\cA^t(p)\to\nu$, as $t\to\infty$. By \Cref{lem:qual:positive-stationary-states}, every everywhere-positive stationary distribution of $T_\cA$ is a hard-core Gibbs measure whose fugacity is constant on each block of $\cA$; that is,
\[
\nu=\mu_{G,\vect\lambda}, %\qquad .
\]
for some $\vect\lambda\in\G_\cA$.  By \Cref{lem:pre:conservation-laws},
\(
\vect\alpha_\cA(T^t(p))=\vect\alpha_\cA(p), 
\) for all $t\ge 0$.
Therefore,  \[
\vect\alpha_\cA(p)=\vect\alpha_\cA(\mu_{G,\vect\lambda}).
\]
It remains to show that the latter relation uniquely identifies the fugacity vector $\vect\lambda\in\G_\cA$. To this end, it suffices show that the map
\[
\vect\lambda\in\G_\cA\longmapsto\vect\alpha_\cA(\mu_{G,\vect\lambda})
\]
is injective. We introduce the new variables \(\vect\theta_\cA=(\theta_A)_{A\in\cA}\), where $\theta_A=\sum_{v\in A}\log\lambda_v$, or equivalently \(\theta_A=|A|\log \l_u\) for any \(u\in A\), and we denote, with slight abuse of notation,  by \(\m_{G,\vect \theta_\cA}\) the corresponding hard-core measure parametrized by \(\vect\theta_\cA\): \[
\m_{G,\vect\theta_\cA}(\s)=\frac{e^{\sum_{A\in\cA}\theta_A s_A(\sigma)}}{Z_{G,\vect\theta_\cA}}, \qquad s_A(\sigma):=\frac{1}{|A|}\sum_{v\in A}\sigma_v, %A\in\cA,
\qquad Z_{G,\vect\theta_\cA}=\sum_{\sigma\in\O_G} e^{\sum_{A\in\cA}\theta_A s_A(\sigma)}.
\]
In this setting, we need to prove that
\(
\vect\theta_\cA\in\bbR^\cA\longmapsto f(\vect\theta_\cA)=\vect\a_{\cA}(\m_{G,\vect\theta_\cA})
\) is injective. Since \(f=\nabla \log(Z_{G,\vect\theta_\cA})\), the Jacobian $\nabla f=\Cov_{\vect\theta_\cA}(\vect s_\cA)$ is the covariance matrix of the vector $\vect s_\cA(\s)=(s_A(\sigma))_{A\in\cA}$ with $\sigma\sim\mu_{G,\vect\theta_\cA}$. We claim that this matrix is   positive definite for every $\vect\theta_\cA\in\bbR^\cA$. Indeed, for every \(\vect x\in\bbR^\cA\),
\[
\vect x^\transpose \Cov_{\vect\theta_\cA}(\vect s_\cA(\s))  \vect x=\Var_{\vect \theta_\cA} (\vect x^\transpose \vect s_\cA(\s)).
\] 
If \(x_A\neq0\) for some \(A\in \cA\), then, taking \(u\in A\), we have \begin{align*}
    &\bbP_{\vect\theta_\cA}(s_A(\sigma)>0)\ge \bbP_{\vect\theta_\cA}(\s_u=1,\sigma_{V\setminus\{u\}}=0)>0,\\
    &\bbP_{\vect\theta_\cA}(s_A(\sigma)=0)\ge \bbP_{\vect\theta_\cA}(\s=0)>0,
\end{align*}
and therefore $\Var_{\vect \theta_\cA} (\vect x^T\cdot \vect s_\cA)>0$. Thus $f$ is injective, and there is at most one $\vect\lambda\in\G_\cA$ with $\vect\alpha_\cA(\mu_{G,\vect\lambda})=\vect\alpha_\cA(p)$. 
This completes the proof.
\end{proof}

\section{Quantitative Convergence}
\subsection{Wasserstein contraction}
In this subsection, we prove \Cref{th:informal1}.
% the following quantitative convergence theorem for the mean-field and single-site dynamics. We will use the notation of multi-component mean-field dynamics in order to provide a common proof of both cases. Recall that the partition \(\cA=\{\{u\}:u\in V\}\) gives the single-site dynamics and \(\cA=\{V\}\) the mean-field dynamics; we denote these partitions by \(\cA_{\mathrm{ss}}\) and \(\cA_{\mathrm{mf}}\), respectively.
% \begin{theorem}\label{thm:wass:unif}
%     Consider the multi-mean-field dynamics on a graph \(G=(V,E)\) with \(n\) vertices and maximum degree \(\D\ge0\). Let \(\d\in(0,1)\), and let \(p\in\cP(\O_G)\) have \(\cA\)-density vector \(\vect\a_\cA(p)\) satisfying, for all \(A\in\cA\), \[
%     0<\a_A(p)\leq \frac{1-\d}{3(\D+1)}.
%     \]
% Suppose that \(\cA=\cA_{\mathrm{ss}}\) or \(\cA=\cA_{\mathrm{mf}}\), and let \(\vect\l\in\G_\cA\) be the unique fugacity vector such that \(\vect\a_\cA(p)=\vect\a_\cA(\mu_{G,\vect\l})\). Then, for every \(t\ge0\) we have
% \[
% \|T_\cA^t(p)-\mu_{G,\vect\l}\|_{\TV}\le n\,e^{-\d\,t/n}.
% \]
% \end{theorem}
% \medskip
We will use the more general $\cA$-multi-component mean-field notation to provide  unified statements where possible.   We write  \(\cA_{\mathrm{ss}}=\{\{u\}:u\in V\}\) for the single-site partition and \(\cA_{\mathrm{mf}}=\{V\}\) for the mean-field partition.
The proof is based on the following contraction estimate. 

\begin{lemma}\label{lem:wass:unif} Let \(\cA=\cA_{\mathrm{mf}}\) or \(\cA=\cA_{\mathrm{ss}}\). Let \(\delta\in(0,1)\) and let $p,q\in\Pdist(\OmegaG)$ have the same \(\cA\)-density vector $\vect\alpha_\cA(p)=\vect\alpha_\cA(q)=(\alpha_A, A\in \cA)$, and  assume that $ \alpha_A\le \frac{1-\delta}{3(\Delta+1)}$ for all \(A\in\cA\).  Then,  
\[ W_1(T_\cA(p),T_\cA(q)) \le \left(1-\frac{\delta}{n}\right)W_1(p,q). 
\] 
\end{lemma} 
Before proving the lemma, we show how it implies \Cref{th:informal1}.
\begin{proof}[Proof of \Cref{th:informal1}.]
Let \(\cA=\cA_{\mathrm{mf}}\) or \(\cA=\cA_{\mathrm{ss}}\). By \Cref{lem:pre:conservation-laws}, the \(\cA\)-density vector is preserved along the nonlinear dynamics. Thus, for all $t\in\bbN$,
\(
\vect \a_\cA(T_\cA^t(p))=\vect\a_\cA(p).
\)
Let \(\vect\lambda\) be the unique fugacity vector \(\vect \l\in\G_\cA\) satisfying
\(
\vect\a_\cA(\mu_{G,\vect\lambda})=\vect\a_\cA(p),
\)
as in \Cref{th:multi-mean-field-convergence}.
%By \Cref{lem:pre:gibbs-stationary-measures}  the measure \(\mu_{G,\vect\lambda}\) is stationary. 
Thus, using the invariance $T_\cA^t(\mu_{G,\vect\lambda})=\mu_{G,\vect\lambda}$, $t\in\bbN$, we may apply \Cref{lem:wass:unif} iteratively to the pairs \((T^s_\cA(p),\mu_{G,\vect\lambda})\), $0\le s\le t$. This gives, for every \(t\in\bbN\),
\[
W_1(T^t_\cA(p),\mu_{G,\vect\lambda})
\le
\left(1-\frac{\delta}{n}\right)^t
W_1(p,\mu_{G,\vect\lambda}).
\]
Since the Hamming distance is at most \(n\), we have \(W_1(p,\mu_{G,\vect\lambda})\le n\). Moreover, the total variation distance is bounded above by the Wasserstein distance with respect to the Hamming metric. Hence
\[
\|T^t_\cA(p)-\mu_{G,\vect\lambda}\|_{\TV}
\le
n\left(1-\frac{\delta}{n}\right)^t \le
n\exp\left(-\frac{\delta t}{n}\right). \qedhere
\]
\end{proof}
The proof of  \Cref{lem:wass:unif} uses the coupling construction given in \Cref{alg:quant:multi-mean-field-coupling}.
\begin{algorithm}[H] \caption{Coupling for one step of the $\cA$-mean-field nonlinear dynamics} \label{alg:quant:multi-mean-field-coupling} \begin{algorithmic}[1] \Require Distributions $p,q\in\Pdist(\OmegaG)$ with $\vect\alpha_\cA(p)=\vect\alpha_\cA(q)$; an optimal coupling $\pi_{\mathrm{opt}}\in\Gamma(p,q)$ with respect to the expected Hamming distance. 
\Ensure Coupled outputs $\tau^p\sim T_\cA(p)$ and $\tau^q\sim T_\cA(q)$. 
\State Sample $(\sigma^p,\sigma^q)\sim\pi_{\mathrm{opt}}$. \State Sample $v\in V$ uniformly.
\State Sample $u\in A(v)$ uniformly.
\If{$\sigma^p_v=1$ \textbf{and} $\sigma^q_v=0$} \State Sample $\xi^p\sim p$. \State Set $\xi^q\gets \xi^p$. \ElsIf{$\sigma^p_v=0$} \State Sample $\xi^q\sim q$. \State Set $\xi^p\gets \xi^q$. \Else \Comment{$\sigma^p_v=\sigma^q_v=1$} \State Sample $(\xi^p,\xi^q)\sim\pi_{\mathrm{opt}}$. \EndIf \State Set $(\tau^p,\widehat{\tau}^p)\gets \Phi_{v,u}(\sigma^p,\xi^p)$ and output $\tau^p$. \State Set $(\tau^q,\widehat{\tau}^q)\gets \Phi_{v,u}(\sigma^q,\xi^q)$ and output $\tau^q$. \State \Return $(\tau^p,\tau^q)$. \end{algorithmic} \end{algorithm}  
It is not hard to see that this algorithm indeed provides a coupling of $T_\cA(p),T_\cA(q)$, for any fixed partition $\cA$ of $V$. The key observation is that if the configuration $\sigma$ to be updated at $v\in V$ satisfies $\sigma_v=0$, then the only information needed to compute the distribution after the update is contained in the density $\a_{A(v)}(p)=\a_{A(v)}(q)$, where $A(v)$ is the unique block of $\cA$ which contains $v$. More formally, this can be derived as follows.  
Fix a partition  \(\cA\) of \(V\),  and a distribution \(p\in\cP(\O_G)\). Given $v\in V$, $\sigma,\tau\in\O_G$, define  \begin{equation}\label{defSigmap}
\Sigma^\cA_p(\tau,\sigma,v)
:=
\frac{1}{|A(v)|}
\sum_{u\in A(v),\sigma'}
p(\sigma')
\Bigl(
\mathbf{1}\bigl((u,v)_{\sigma,\sigma'}\bigr)\,
\mathbf{1}\bigl(\tau=S_u(\sigma,\sigma'_v)\bigr)
+
\mathbf{1}\bigl(\widehat{(u,v)}_{\sigma,\sigma'}\bigr)\,
\mathbf{1}(\tau=\sigma)
\Bigr).
\end{equation}
%This is the probability, given \(\s\) and \(u\in V\),  to see \(\t\) by updating \(\s\) at \(u\)-coordinate. 
According to this definition, 
\begin{equation}\label{lemcoup1}
    T_\cA(p)(\t)=\frac{1}{n}\sum_{\stackrel{\s\in\O_G}{v\in V}} p(\s)\Sigma^\cA_p(\t,\s,v).
\end{equation}

\begin{lemma}\label{lemma coupling idea}
%Fix a\(\cA\) be a partition of \(V\) and let 
Let \(p,q\in\cP(\O_G)\) be such that \(\vect\a_\cA(p)=\vect\a_\cA(q)\). For any \(v\in V\), $\sigma\in\O_G$, if  \(\s_v=0\) then
\begin{equation}\label{lemcoup2}
\Sigma^\cA_p(\t,\s,v)=\Sigma^\cA_q(\t,\s,v), \quad\forall \t\in\O_G  .
\end{equation}
% In particular, 
% \begin{equation}\label{lemcoup3}
%     T_\cA(p)(\t)=\frac{1}{n}\sum_{\stackrel{\s\in\O_G}{v\in V}} p(\s)\Sigma^\cA_p(\t,\s,v)\mathbf{1}(\s_v=1) + \frac{1}{n}\sum_{\stackrel{\s\in\O_G}{v\in V}} p(\s)\Sigma^\cA_q(\t,\s,v) \mathbf{1}(\s_v=0).
% \end{equation}
In particular, \Cref{alg:quant:multi-mean-field-coupling} provides a coupling of $T_\cA(p),T_\cA(q)$. 
\end{lemma}

\begin{proof}
%Fix \(\sigma\in\Omega_G\) and \(v\in V\). 
If \(\s_v=0\), we distinguish two cases:
\begin{enumerate}
    \item\(N(v)\cap\sigma\neq\emptyset\), and \(\s_v=0\);
    \item\(N(v)\cap\sigma=\emptyset\) and \(\sigma_v=0\).
\end{enumerate}
The proof is completed by showing that \eqref{lemcoup2} holds in both cases.
In the first case, we necessarily have \(\tau=\sigma\), even if the \((v,u)\)-exchange occurs. This implies
\begin{align*}
\Sigma^\cA_p(\tau,\sigma,v)
&=\frac{1}{|A(v)|}\sum_{\sigma',u\in A(v)} p(\sigma')\mathbf{1}(\tau=\sigma)\\
&=\mathbf{1}(\sigma=\tau)=\Sigma^\cA_q(\tau,\sigma,v).
\end{align*}
In the second case, the event \((v,u)_{\sigma,\sigma'}\) holds independently of \(\sigma'\) and \(u\in A(v)\). Therefore,
\begin{align*}
\Sigma^\cA_p(\tau,\sigma,v)
&=
\frac{1}{|A(v)|}
\sum_{u\in A(v),\sigma'}
p(\sigma'\mid\sigma'_u=1)\,m_u(p)\,
\mathbf{1}\bigl(\tau=S_v(\sigma,1)\bigr) \\
&\quad+
\frac{1}{|A(v)|}
\sum_{u\in A(v),\sigma'}
p(\sigma'\mid\sigma'_u=0)\,(1-m_u(p))\,
\mathbf{1}\bigl(\tau=S_v(\sigma,0)\bigr) \\
&=
\frac{\mathbf{1}\bigl(\tau=S_v(\sigma,1)\bigr)}{|A(u)|}
\sum_{u\in A(v)}\,m_u(p)
\sum_{\sigma'}p(\sigma'\mid\sigma'_u=1) \\
&\quad+
\frac{\mathbf{1}\bigl(\tau=S_v(\sigma,0)\bigr)}{|A(u)|}
\sum_{u\in A(v)}\,(1-m_u(p))
\sum_{\sigma'}p(\sigma'\mid\sigma'_v=0) \\
&=
\a_{A(v)}(p)\,\mathbf{1}\bigl(\tau=S_v(\sigma,1)\bigr)
+
(1-\a_{A(v)}(p))\,\mathbf{1}\bigl(\tau=S_v(\sigma,0)\bigr)\\
&=\Sigma^\cA_q(\tau,\sigma,v),
\end{align*}
which concludes the proof of \eqref{lemcoup2}.

Using \eqref{lemcoup1} we conclude the statement about the coupling of $T_\cA(p),T_\cA(q)$. 
\end{proof}

We can now turn to the proof of \Cref{lem:wass:unif}.

\begin{proof}[Proof of \Cref{lem:wass:unif}]
We assume that all random variables are sampled as prescribed by \Cref{alg:quant:multi-mean-field-coupling}.
 For two independent sets $\sigma,\sigma'\in\OmegaG$, let \[ D_{\sigma,\sigma'}:=\{z\in V:\sigma_z\ne\sigma'_z\}. \] 
We claim that the following inequality holds with probability one: \begin{equation} \label{eq:quant:multi-field-coupling} \begin{aligned} |D_{\tau^p,\tau^q}|-|D_{\sigma^p,\sigma^q}| \le\;& -\mathbf 1\{\sigma^p_v=1 \land \sigma^q_v=0\} \mathbf 1\{\forall\, z\in\neighbor{u}:\xi^p_z=0\} \\[4pt] &- \mathbf 1\{\sigma^p_v=0 \land \sigma^q_v=1\} \mathbf 1\{\forall\, z\in\neighbor{u}:\xi^q_z=0\} \\[4pt] &+ \mathbf 1\{\xi^p_u=\xi^q_u=1\} \mathbf 1\{\exists\, z\in\neighbor{v}:\sigma^p_z\ne\sigma^q_z\} \\[4pt] &+ \mathbf 1\{\sigma^p_v=\sigma^q_v=1\} \mathbf 1\{\exists\, z\in\neighborCl{u}:\xi^p_z\ne\xi^q_z\}. \end{aligned} \end{equation} 
We consider four cases to prove \eqref{eq:quant:multi-field-coupling}.\\ 

\emph{Case 1:} $\sigma^p_v=\sigma^q_v=0$. Note that the first two terms on the right-hand side are $0$, so the right-hand side is nonnegative. Since the left-hand side is at most $1$, the only nontrivial case to consider is \[ |D_{\tau^p,\tau^q}|-|D_{\sigma^p,\sigma^q}|=1. \] In this case, $v\notin D_{\sigma^p,\sigma^q}$ and $v\in D_{\tau^p,\tau^q}$. This can happen only if, in exactly one of the two chains, the value at $v$ changes to $1$. Since the two chains use the same auxiliary independent set in this case, we must have $\xi^p_u=\xi^q_u=1$. Moreover, exactly one of the two chains accepts the insertion at $v$, so $\sigma^p$ and $\sigma^q$ must disagree on at least one neighbor of $v$. Therefore the third term on the right-hand side of \eqref{eq:quant:multi-field-coupling} is equal to one.\\

\emph{Case 2:} $\sigma^p_v=1$ and $\sigma^q_v=0$. In this case the left-hand side is at most $0$, and the right-hand side is at least $-1$. Hence the only nontrivial case to consider is \[ |D_{\tau^p,\tau^q}|-|D_{\sigma^p,\sigma^q}|=0. \] We consider two subcases. If $\xi^p_u=\xi^q_u=0$, then the discrepancy at $v$ remains only if the $p$-chain fails to change the value at $v$ from $1$ to $0$. This can happen only if there exists $z\in\neighbor{u}$ such that $\xi^p_z=1$. Hence \[ \mathbf 1\{\sigma^p_v=1 \land \sigma^q_v=0\} \mathbf 1\{\forall\, z\in\neighbor{u}:\xi^p_z=0\}=0, \] and the right-hand side is at least zero. If $\xi^p_u=\xi^q_u=1$, then the discrepancy at $v$ remains only if the $q$-chain fails to change the value at $v$ from $0$ to $1$. Since $\sigma^p_v=1$, all neighbors of $v$ are vacant in $\sigma^p$. Therefore this failure implies that there exists $z\in\neighbor{v}$ such that $\sigma^p_z\ne\sigma^q_z$. Thus \[ \mathbf 1\{\xi^p_u=\xi^q_u=1\} \mathbf 1\{\exists\, z\in\neighbor{v}:\sigma^p_z\ne\sigma^q_z\}=1, \] and the right-hand side is at least zero.\\

\emph{Case 3:} $\sigma^p_v=0$ and $\sigma^q_v=1$. This is symmetric to the previous case.\\

\emph{Case 4:} $\sigma^p_v=\sigma^q_v=1$. Similar to Case 1, the only nontrivial case to consider is \[ |D_{\tau^p,\tau^q}|-|D_{\sigma^p,\sigma^q}|=1. \] Then $v\notin D_{\sigma^p,\sigma^q}$ and $v\in D_{\tau^p,\tau^q}$. Hence, in exactly one chain, the value at $v$ changes. This can happen only if the two auxiliary independent sets differ somewhere in $\neighborCl{u}$; otherwise both chains behave in the same way at $v$. Consequently, \[ \mathbf 1\{\sigma^p_v=\sigma^q_v=1\} \mathbf 1\{\exists\, z\in\neighborCl{u}:\xi^p_z\ne\xi^q_z\}=1, \] and so the right-hand side is at least $1$. This proves \eqref{eq:quant:multi-field-coupling}.

So far, we have established the pointwise inequality \eqref{eq:quant:multi-field-coupling}. Now, we compute the expectation of each term in the right-hand side of the inequality, which gives us an upper bound on the expectation of the left-hand side.

\emph{First and second terms:} Here \(v\) is uniformly distributed in \(V\) and, given \(v\), \(u\) is uniformly distributed in \(A(v)\), independently of all configurations. For the first term, we have \begin{align*} &\E\!\left[ \mathbf 1\{\sigma^p_v=1 \land \sigma^q_v=0\} \mathbf 1\{\forall\, z\in\neighbor{u}:\xi^p_z=0\} \right] \\[4pt] ={}& \frac1{n} \sum_{v\in V}\frac{\Pr(\sigma^p_v=1 \land \sigma^q_v=0) }{|A(v)|}\sum_{u\in A(v)} \Pr\!\left(\forall\, z\in\neighbor{u}:\xi^p_z=0\right) \\[6pt] \ge{}& \frac1{n} \sum_{v\in V}\Pr(\sigma^p_v=1 \land \sigma^q_v=0)\left(1-\frac{1}{|A(v)|}\,\sum_{u\in A(v)} \sum_{z\in\neighbor{u}}m_z(p)\right).
\end{align*}
Exchanging the order of summation gives
\[
\frac{1}{|A(v)|}\sum_{u\in A(v)}\sum_{z\in N(u)} m_z(p)=\frac{1}{|A(v)|}\sum_{z\in V}m_z(p)|N(z)\,\cap A(v)|.
\]
When \(\cA=\cA_{\mathrm{ss}}\), we have \(A(v)=\{v\}\), and since \(m_z(p)=\a_{\{z\}}(p)\le\frac{1-\d}{3(\D+1)}\) for every \(z\in V\),
\[
\frac{1}{|A(v)|}\sum_{z\in V}m_z(p)|N(z)\,\cap A(v)|=\sum_{z\in V}m_z(p)\mathbf{1}(z\in N(v))\le \frac{(1-\delta)}3.
\]
When \(\cA=\cA_{\mathrm{mf}}\), we have \(A(v)=V\), and the density vector reduces to the single density value \(\a(p)=\frac{1}{n}\sum_{z\in V} m_z(p)\leq \frac{1-\d}{3(\D+1)}\); hence
\[
\frac{1}{|A(v)|}\sum_{z\in V}m_z(p)|N(z)\,\cap A(v)|=\frac{1}{n}\sum_{z\in V}m_z(p)|N(z)|\le \D\a(p)\le  \frac{(1-\delta)}3
\]
In both cases, we conclude that
\[
\E\!\left[ \mathbf 1\{\sigma^p_v=1 \land \sigma^q_v=0\} \mathbf 1\{\forall\, z\in\neighbor{u}:\xi^p_z=0\} \right] \ge  \frac{1}{n}\left[ 1-(1-\d)/3\right]\sum_{v\in V}\Pr(\sigma^p_v=1 \land \sigma^q_v=0).
\]
Similarly, since \(q\) satisfies the same density assumption as \(p\), the second term satisfies
\[
\E\!\left[ \mathbf 1\{\sigma^p_v=0 \land \sigma^q_v=1\} \mathbf 1\{\forall\, z\in\neighbor{u}:\xi^q_z=0\} \right] \ge \frac1n\left[ 1-(1-\d)/3\right] \sum_{v\in V} \Pr(\sigma^p_v=0 \land \sigma^q_v=1).
\]
Therefore, in both the single-site and the mean-field case, the sum of the first two terms is bounded from below by
\[ 
\frac1n\left[ 1-(1-\d)/3\right] \sum_{v \in V} \left[\Pr(\sigma^p_v = 1 \land \sigma^q_v = 0) + \Pr(\sigma^p_v = 0 \land \sigma^q_v = 1)\right] = 
\frac 1n\left[ 1-(1-\d)/3\right] \E\!\left[|D_{\sigma^p,\sigma^q}|\right]. 
\]

\emph{Third term:} Note that, on the event \(\neighbor{v}\cap D_{\sigma^p,\sigma^q}\ne\emptyset\), we cannot have \(\sigma^p_v=\sigma^q_v=1\). Hence, on this event, the coupling uses one of the first two branches, and the auxiliary independent sets coincide, i.e., \(\xi^p=\xi^q\). Moreover, this common auxiliary independent set has law \(p\) or \(q\), both of which have \(\cA\)-density vector \(\vect\alpha_\cA\). Since, given \(v\), the vertex \(u\) is uniformly distributed in \(A(v)\), it follows that, conditionally on \(v\) and on this event, \(\Pr(\xi^p_u=\xi^q_u=1)=\a_{A(v)}\). Thus \begin{align*} &\E\!\left[ \mathbf 1\{\xi^p_u=\xi^q_u=1\} \mathbf 1\{\exists\, z\in\neighbor{v}:\sigma^p_z\ne\sigma^q_z\} \right] \\[4pt]
&= \frac1n \sum_{v\in V} \alpha_{A(v)}\, \Pr\!\left(\neighbor{v}\cap D_{\sigma^p,\sigma^q}\ne\emptyset\right) \\[6pt]
&\le \frac{1-\d}{3(\D+1)}\, \E\!\left[ \frac1n \left| \bigcup_{z\in D_{\sigma^p,\sigma^q}}\neighbor{z} \right| \right] %\\[6pt] 
%&
\le \frac{1-\d}{3n} \E\!\left[|D_{\sigma^p,\sigma^q}|\right]. \end{align*}

\emph{Last term:} On the event \(\{\sigma^p_v=\sigma^q_v=1\}\), the coupling samples \((\xi^p,\xi^q)\) from the optimal coupling \(\pi_{\mathrm{opt}}\). Hence \begin{align*} 
&\E\!\left[ \mathbf 1\{\sigma^p_v=\sigma^q_v=1\} \mathbf 1\{\exists\, z\in\neighborCl{u}:\xi^p_z\ne\xi^q_z\} \right] \\[4pt] 
&=\frac1n \sum_{v\in V} \Pr(\sigma^p_v=\sigma^q_v=1)\,\frac{1}{|A(v)|}\sum_{u\in A(v)}\Pr_{(\xi^p,\xi^q)\sim\pi_{\mathrm{opt}}}(N^+(u)\,\cap D_{\xi^p,\xi^q}\not=\emptyset)  \\[4pt]
&=\frac1n \sum_{u\in V} \Pr_{(\xi^p,\xi^q)\sim\pi_{\mathrm{opt}}}(N^+(u)\,\cap D_{\xi^p,\xi^q}\not=\emptyset)  \,\frac{1}{|A(u)|}\sum_{v\in A(u)}\Pr(\sigma^p_v=\sigma^q_v=1)\\[4pt]
&= \frac1n \sum_{u\in V} \a_{A(u)}\Pr_{(\xi^p,\xi^q)\sim\pi_{\mathrm{opt}}}(N^+(u)\,\cap D_{\xi^p,\xi^q}\not=\emptyset)  \\[6pt] 
&\le \frac{(1-\d)}{3n(\D+1)} \E_{(\xi^p,\xi^q)\sim\pi_{\mathrm{opt}}} \left[ \frac1n \left| \bigcup_{z\in D_{\xi^p,\xi^q}}\neighborCl{z} \right| \right]\\[4pt]
&\le \frac{1-\d}{3n}\E_{(\xi^p,\xi^q)\sim\pi_{\mathrm{opt}}} \left[|D_{\xi^p,\xi^q}|\right].
\end{align*}
Combining the previous inequality with the identity \(\E_{(\xi^p,\xi^q)\sim\pi_{\mathrm{opt}}}\left[|D_{\xi^p,\xi^q}|\right]=\bbE_{\s^p,\s^q}\left[|D_{\s^p,\s^q}|\right]\), we obtain the following upper bound on the last term:
\[
\frac{1-\d}{3n} \bbE_{\s^p,\s^q}\left[|D_{\s^p,\s^q}|\right].
\]

\emph{Conclusion:} Combining the previous bounds with \eqref{eq:quant:multi-field-coupling}, we obtain
\[
\bbE\left[|D_{\t^p,\t^q}|-|D_{\s^p,\s^q}|\right]\le -\frac{\d}{n}\bbE\left[|D_{\s^p,\s^q}|\right]
\]
In particular,\[
 W_1(T_\cA(p), T_\cA(q))\le \bbE\left[|D_{\t^p,\t^q}|\right]\le (1-\d/n)\bbE\left[|D_{\s^p,\s^q}|\right]=(1-\d/n)W_1(p,q),
\]
which is the desired bound and concludes the proof.
\end{proof}
\medskip

\subsection{The particle system} \label{subsec:quant:particle-system} 
We consider a system of $N\in\bbN$ \emph{particles}
$\vect\s=\{\s^{(i)},\, i=1,\ldots,N\}$, where 
$\s^{(i)}=(\s^{(i)}_u,\, u\in V)\in\O_G$, for each $i\in[N]$. Equivalently,
$\vect\s\in\O_G^N$ can be seen as an independent set on the graph
$G_N:=G^{\sqcup N}$, the disjoint union of $N$ copies $G(i)$ of $G$. The vertex set of $G_N$ is $V_N=\{(i,v):\, i\in[N],v\in V\}$.
For $i\in[N]$ and a probability measure $\nu$ on $\OmegaG^N$, let $P_i\nu$
denote the $i$-th one-particle marginal of $\nu$, that is,
\[
(P_i\nu)(\eta)
:=
\nu\bigl(\{\vect\sigma\in\OmegaG^N:\sigma^{(i)}=\eta\}\bigr),
\qquad \eta\in\OmegaG.
\]  
%We fix a symmetric $n\times n$ stochastic matrix
% $\cK$. 
 The evolution of the particle system is a Markov chain with state space $\O_G^N$ where the occupation variables $\s^{(i)}_v$ and $\s^{(j)}_w$ are exchanged between vertices $(i,v),(j,w)\in V_N$, with the constraint that the resulting configuration is again a valid independent set of $G_N$. 
% The evolution of the particle system is the Markov chain with state space $\O_G^N$ defined as follows. 
%driven by a
%swapping mechanism. 

To define the Markov chain, we introduce some notation. 
%When the selected pair of particles is
%(i,j)$ and the pair of sites $(v,w)\in V^2$ is sampled, the process attempts to
%swap the occupation values at the sites $(i,v)$ and $(j,w)$. 
Given vertices $(i,v),(j,w)\in V_N$, let
$\widehat{\vect\s}$ be the configuration obtained from $\vect\s$ by the swap of $\s^{(i)}_v$ and $\s^{(j)}_w$,
namely
\[
\widehat{\s}^{(k)}_x
=
\begin{cases}
\s^{(j)}_w, & (k,x)=(i,v),\\
\s^{(i)}_v, & (k,x)=(j,w),\\
\s^{(k)}_x, & \text{otherwise},
\end{cases}
\]
for all $(k,x)\in V_N$, and define
\begin{equation}\label{defswapPS}
   \vect\s^{i,v;j,w}
:=
\begin{cases}
\widehat{\vect\s}, & \widehat{\vect\s}\in\O_G^N,\\
\vect\s, & \widehat{\vect\s}\notin\O_G^N,
\end{cases} 
\end{equation}
with the convention $\vect\s^{i,v;j,w}=\vect\s$ when $(i,v)=(j,w)$. For $i\neq j$
the move exchanges a spin of particle $i$ with a spin of particle $j$, whereas
for $i=j$ it is an internal swap within the single particle $\s^{(i)}$; in
both cases the update is accepted only if the resulting configuration remains a
valid independent set, i.e.\ $\widehat{\vect\s}\in\O_G^N$. The corresponding jump
is
\begin{equation}\label{defkacexch}
    \vect\s\to{\vect\s}^{i,v;j,w}.
\end{equation}
We consider two different evolutions: the mean-field dynamics, where $u,v$ are arbitrary sites, and the single-site dynamics, where $u=v$.

\subsubsection{Mean-field particle system}
Here both pairs $i,j\in[N]$ and $v,w\in[n]$ are picked uniformly at random. 
%If $\vect\s\in\O_G^N$ is the current state at time
%$t\in\bbN$, to sample the new state $\vect\s'$ at time $t+1$ we pick
%uniformly at random a pair of indices $i,j\in[N]$, identifying the two particles  $\s^{(i)},\s^{(j)}$, and select a pair of
%sites $v,w\in [n]$ with probability $\tfrac1n\cK(v,w)$. The state $\vect\s'$ is obtained by exchanging the occupation values at the the vertices $(i,v)$ and $(j,w)$, the move being accepted only if the resulting configuration is still a valid
%independent set. As we will see (cf.\ the mean-field case below), it is sometimes
%useful to allow $i=j$, in which case the swap exchanges spins between two sites
%of the same particle. Once the pair of particles has been selected, the swapping
%dynamics is the same in both models.We are in the setting presented above, with transport kernel $\cK=\frac1n$ and
%allowing exchanges between particles with the same label $i=j$. Since the
%dynamics defined above is a linear Markov chain, we can describe it through its
Thus, the stochastic matrix of the Markov chain is given by %$\vect\sigma,\vect\eta\in\Omega_G^N$, define
\[
\PmfN(\vect\sigma,\vect\eta)
:= \frac{1}{N^2 n^2} \sum_{i,j\in[N]} \sum_{v,w\in V}
\mathbf 1\left\{ \vect\sigma^{i,v;j,w}=\vect\eta \right\}.
\]
The dynamics conserves the total
occupation number
\begin{equation}\label{defmeanfieldconserved}
    M(\vect\s)=\sum_{i=1}^N|\s^{(i)}|,\qquad |\s^{(i)}|=\sum_{v\in V}\s^{(i)}_v.
\end{equation}
Hence, in order to obtain an irreducible chain, it is necessary to fix the number of occupied sites. Equivalently,
we restrict to a state space with a fixed density.
\begin{definition}\label{defOmegaNa}
Given a density $\a\in\Big(0,\frac{1}{\D+1}\Big]$,
\begin{equation}
    \OmegaNa=\left\{ \vect\sigma\in\OmegaG^N: M(\vect\sigma)=\lfloor Nn\a\rfloor \right\}.
\end{equation}
\end{definition}
We note that the restriction $\a\in(0,\frac{1}{\D+1}]$ is optimal, in the sense that it is the largest range ensuring that $\OmegaNa$ is nonempty for every graph of maximum degree $\D$ and every $N\ge 1$. For $\a\le \frac{1}{\D+1}$, the initialization procedure in \Cref{sec:alg-application-mean-field} constructs a configuration $\vect\sigma\in\OmegaNa$; see \Cref{alg:prescribed-density-particle-system-sampler}. Conversely, if $\a>\frac{1}{\D+1}$, then taking $G=K_{\D+1}$ makes $\OmegaNa$ empty for all sufficiently large $N$.

We define $\m_{N,\a}^{\mathrm{mf}}$ as the uniform distribution over $\OmegaNa$, 
\begin{equation}
    \m_{N,\a}^{\mathrm{mf}}=\mathrm{Unif}(\OmegaNa)\,.
\end{equation}
By symmetry of the
stochastic matrix $\PmfN$, the pair $(\PmfN,\m_{N,\a}^{\mathrm{mf}})$
is reversible, for any $\a\in (0,\tfrac{1}{\D+1}]$. 
When the density is sufficiently small, we have the following quantitative convergence
statement, expressed in terms of entropy contraction. 
%As we will see, the fast convergence to %equilibrium established in
%\Cref{thm-disc-partmf}, together with %the Kac chaos property in
%\Cref{lem:alg:fixed-density-kac-chaos}, %yields a valid  algorithm for the
%hard-core model when its occupancy function is known; see
%\Cref{thm:prescribed-density-particle-system-sampler}.
Recall that $\alphac$ is the critical density introduced in \eqref{eq:crit_dens}.
\begin{theorem}\label{thm-disc-partmf}
    Let $G = (V, E)$ be a graph on $n$ vertices with maximum degree $\D\ge 3$. For any
    $\g>0$ and $\d\in(0,1)$, consider a density $\a\in(\g,(1-\d)\a_c(\D))$. Then, there exist constants $r:=r(\g,\d,\D)>0$ and $N_0 :=N_0(\g,\d,\D)$ such that
    for any probability measure $\n\in\cP(\O^{\mathrm{mf}}_{N,\a})$ and any $N \geq N_0$, it holds that
    \begin{equation}\label{eq:ent_contr}
        H(\nu P_{N}^{\mathrm{mf}}\mid \m^{\mathrm{mf}}_{N,\a})
        \le \left(1-\frac{r}{Nn}\right) H(\n\mid \m^{\mathrm{mf}}_{N,\a}).
    \end{equation}
\end{theorem}

In \Cref{sec:alg-application-mean-field}, we use this entropy contraction result to prove the $\widetilde O(Nn)$ mixing-time bound for our mean-field particle-system sampler.

\begin{proof}
    The entropy contraction \eqref{eq:ent_contr} can be derived from 
    a logarithmic Sobolev inequality (LSI) for the pair
$(P^{\mathrm{mf}}_N,\m^{\mathrm{mf}}_{N,\a})$ established recently in \cite{jain2023optimalmixingdownupwalk,JainPhamVuong_EntropicSparseLocalization}.

    The LSI is formulated as follows. 
    The Dirichlet form associated to the reversible pair $(P_{N}^{\mathrm{mf}},\m^{\mathrm{mf}}_{N,\a})$ is given by
    \begin{equation}\label{eq:dirichlet-particle}
        \cE_{N,\a}(F,G)
        =
        \frac{1}{2N^2n^2}
        \sum_{i,j=1}^{N}\sum_{v,w\in V}
        \m^{\mathrm{mf}}_{N,\a}\!\left[
            \nabla^{i,v;j,w}F\,\nabla^{i,v;j,w}G
        \right],
    \end{equation}
    for all $F,G:\O^{\mathrm{mf}}_{N,\a}\to\bbR$, where the gradient operator
    $\nabla^{i,v;j,w}$ is defined by
    \begin{equation}\label{eq:undeHM}
        \nabla^{i,v;j,w}F(\vect\s) := F(\vect\s^{i,v;j,w}) - F(\vect\s).
    \end{equation}
    The LSI for the pair $(P^{\mathrm{mf}}_N,\m^{\mathrm{mf}}_{N,\a})$ with constant $\varrho\ge 0$  is the inequality
    \begin{equation}\label{eq:lsiHMa}
        \cE_{N,\a}(\sqrt{F},\sqrt{F}) \geq \varrho\,\Ent_{\m^{\mathrm{mf}}_{N,\a}}[F],
    \end{equation}
    for all
    $F:\O^{\mathrm{mf}}_{N,\a}\to\bbR_{\ge 0}$,
    where we use the notation \(\Ent_{\m^{\mathrm{mf}}_{N,\a}}[F]=\m^{\mathrm{mf}}_{N,\a}[F\log\left(F/\m^{\mathrm{mf}}_{N,\a}[F]\right)]\).
    
Set
\[
\beta:=\frac{\lfloor \a Nn\rfloor}{Nn}.
\]
Since $\a\in(\g,(1-\d)\a_c(\D))$, by taking
$N\ge N_0(\g,\d,\D)$ we may ensure that
\begin{equation}\label{eq:eqbeta}
\beta\in\left(\frac{\g}{2},(1-\d)\a_c(\D)\right).
\end{equation}
The inequality~\eqref{eq:lsiHMa} is established with 
    \begin{equation}\label{eq:r0}
    \varrho =\frac{r_0}{Nn},
    \end{equation}
    with a constant  $r_0=r_0(\g,\d,\D)>0$, 
    by applying 
    \cite[Theorem 2]{jain2023optimalmixingdownupwalk} and
    \cite[Equation 6.3]{JainPhamVuong_EntropicSparseLocalization} in the density window \eqref{eq:eqbeta},
%     \[
% \beta\in\left(\frac{\gamma}{2},(1-\delta)\alpha_c(\Delta)\right),
% \]
    together with a simple comparison of the Dirichlet form $\cE_{N,\a}(\cdot,\cdot)$ with the Dirichlet form of the down-up walk defined
    in \cite{jain2023optimalmixingdownupwalk}.

     Once the statement~\eqref{eq:lsiHMa} is available, the desired contraction \eqref{eq:ent_contr}, with $r=r_0$, follows from
    \cite[Lemma 2.5]{EntLecNotes} and \cite[Proposition 6]{miclo2008remarques}. We note that the argument in \cite[Lemma 2.5]{EntLecNotes} requires  the positive semidefiniteness of the
    stochastic matrix $P_N^{\mathrm{mf}}$.
The latter can be checked by verifying the stronger property that the holding probability satisfies
    \begin{equation}\label{eq:lazy}
        P_{N}^{\mathrm{mf}}(\vect\s,\vect\s)\ge \frac12,\qquad
        \forall\,\vect\s\in \O^{\mathrm{mf}}_{N,\a}.
    \end{equation}
To see this, note that $\b \in(0,1)$. Since exactly $M(\vect\s)=\b Nn$ sites are occupied and
    $Nn-M(\vect\s)=(1-\b)Nn$ are empty, the probability that both selected sites $(i,v),(j,w)$ are
    occupied equals $\b^2$, and the probability that both are empty equals
    $(1-\b)^2$. Hence
    \[
        P_{N}^{\mathrm{mf}}(\vect\s,\vect\s)\ge (1-\b)^2+\b^2 = 1-2\b(1-\b)\ge \frac12,
    \]
    where the last inequality uses $\b(1-\b)\le \tfrac14$,
    $\b\in[0,1]$.
\end{proof}

\subsubsection{Single-site particle system} 
Here, for $N \geq 2$, we pick $i,j\in[N]$ uniformly at random with the constraint $i\neq j$ and then we pick $v=w\in V$ uniformly at random. We use the notation, 
$\vect\sigma^{i,j;v} := \vect\sigma^{i,v;j,v}$, so that the moves of the chain are % and consider only swaps 
of the type
\begin{equation}\label{defswapsinglesite}
    \vect\s\to \vect\sigma^{i,j;v}.
\end{equation}
The   stochastic matrix of the Markov chain is given by
\begin{equation}\label{defPssNr}
     \PssNr(\vect\sigma,\vect\eta) := \frac{1}{N(N-1)n}
     \sum_{\substack{i,j\in[N]\\ i\ne j}} \sum_{v\in V}
     \mathbf 1\left\{ \vect\sigma^{i,j;v}=\vect\eta \right\}.
\end{equation}
As in the mean-field case, the dynamics is conservative. However, here there is a conserved quantity for every $v\in V$. For every
$\vect\s\in\O_G^N$ and every $v\in V$, let
$R_v(\vect\s)=\sum_{i\in[N]}\s_v^{(i)}$ be the number of ones in $\vect\s$ along
the stripe $\{(i,v):i\in[N]\}$. The vector
\begin{equation}\label{defconstonsingle}
    R(\vect\sigma):=(R_v(\vect\sigma))_{v\in V}
\end{equation}
is conserved by the dynamics defined in \eqref{defPssNr}, since all allowed
exchanges are of the type \eqref{defswapsinglesite}. 
%Note the difference with the
%mean-field case, where a single scalar quantity is preserved, namely the total
%occupation number $M(\vect\s)$ of \eqref{defmeanfieldconserved}: here the
%conserved object is the whole vector $R(\vect\s)$. In order to have any chance of
%convergence, 
To obtain an irreducible chain one must therefore restrict
the state space to a fixed value of \eqref{defconstonsingle}.
\begin{definition}
For $\vec r=(r_v)_{v\in V}\in\{0,\ldots,N\}^V$, define
\[
\OmegaNr := \left\{ \vect\sigma\in\OmegaG^N:
R_v(\vect\sigma)=r_v\ \text{for every }v\in V \right\}.
\]
\end{definition}
We note that $\OmegaNr$ is nonempty whenever
\[
(\Delta + 1) \max_{v \in V} r_v \leq N.
\]
Under this condition, in \Cref{sec:alg-application-single-site}, we provide an explicit algorithm for constructing an initial configuration $\vect \sigma \in \OmegaNr$ (see \Cref{alg:initial-particle-configuration}).

By symmetry of the chain $\PssNr$, whenever
$\OmegaNr\neq\emptyset$, the uniform measure
\begin{equation}\label{defunifmeassingle}
    \muNr := \mathrm{Unif}(\OmegaNr)
\end{equation}
is reversible with respect to $\PssNr$. We now prove a rapid mixing bound for the chain $\PssNr$ by comparing it with Glauber dynamics for proper colorings of an auxiliary graph. 

For a graph $H$ and an integer $q\ge 1$, let $\mathcal C_q(H)$ denote the
set of proper $q$-colorings of $H$. The Glauber dynamics on
$\mathcal C_q(H)$ is the Markov chain which is defined as follows: at each step, it chooses a vertex
$x\in V(H)$ uniformly at random and a color $j\in[q]$ uniformly at random,
and recolors $x$ with $j$ if the resulting coloring is proper; otherwise it
stays put. This chain is reversible with respect to the uniform distribution
on $\mathcal C_q(H)$.

\begin{assumption}
\label{assumption:coloring-mixing}
    We will use the following hypothesis about mixing of colorings. It is widely conjectured that Glauber dynamics for proper $q$-colorings of a graph $H$ mixes rapidly whenever $q \geq \Delta(H) + 2$. In this paper, we assume the following form of this conjecture: there exist absolute constants $c_{\mathrm{col}} > 1$ and $\Delta_0 > 0$ such that, for every graph $H$ of maximum degree $\Delta (H) \ge \Delta_0$, the Glauber dynamics on $\mathcal C_q(H)$ mixes in

\[
O
\left(
|V(H)|\log\frac{|V(H)|}{\zeta}
\right)
\]
steps within total variation distance $\zeta$ of the uniform distribution on $\mathcal C_q(H)$ whenever $q \geq c_{\mathrm{col}} \Delta (H)$. 
Note that we assume the constant hidden in big-O is absolute (determined by $c_{\mathrm{col}}$ and $\Delta_0$) and independent of $\Delta (H)$ or $q$.
The classical coupling argument of Jerrum proves such a bound for any fixed $c_{\mathrm{col}} > 2$ and $\Delta_0 = 0$; see \cite{Jer95}.
\end{assumption}

Let $\gamma \in (0, 1)$. Let $\vec r\in \mathbb Z_{\ge 0}^{V}$ satisfy the following assumption

\begin{equation}
    \label{eq:particle-system-vertex-count-assumption}
    \gamma \leq \frac{r_v}{N} \leq \frac{1}{c_{\mathrm{col}}(\Delta + 1)} \qquad \forall v \in V.
\end{equation}

\noindent Note that by the discussion above, $\OmegaNr \neq \emptyset$. We associate to $\vec r$ an
auxiliary graph $H_{\vec r}$ as follows. For each $v\in V$, let
\[
C_v:=\set{(v,a):a\in[r_v]},
\]
and define $H_{\vec r}$ to be the graph with vertex set
\[
V(H_{\vec r})=\bigcup_{v\in V}C_v,
\]
where each $C_v$ induces a clique, and for every edge $\{u,v\}\in E$, all
vertices of $C_u$ are adjacent to all vertices of $C_v$.

Let $\Pcol$ denote the Glauber dynamics on $\mathcal C_N(H_{\vec r})$. We
compare $\PssNr$ with the Glauber dynamics for proper $N$-colorings of
$H_{\vec r}$.

For a coloring $c\in\mathcal C_N(H_{\vec r})$, define
\[
\Psi_{\vec r}(c):=(\sigma^{(1)},\ldots,\sigma^{(N)}),
\]
where
\[
\sigma^{(i)}
:=
\set{
v\in V:
\text{there exists }a\in[r_v]\text{ such that }c(v,a)=i
}.
\]

\begin{lemma}
\label{lem:configurations-to-colorings}
Let
$H_{\vec r}$ and $\Psi_{\vec r}$ be the graph and the map defined above.

First, for every $c\in\mathcal C_N(H_{\vec r})$, one has
\[
\Psi_{\vec r}(c)\in\Omega^{\mathrm{ss}}_{N,\vec r}.
\]

Second, for every $\vect\sigma\in\Omega^{\mathrm{ss}}_{N,\vec r}$,
\[
\size{\Psi_{\vec r}^{-1}(\vect\sigma)}
=
\prod_{v\in V} r_v!.
\]

Consequently, if $c$ is sampled uniformly from $\mathcal C_N(H_{\vec r})$,
then $\Psi_{\vec r}(c)$ is distributed according to
$\mu^{\mathrm{ss}}_{N,\vec r}$.
\end{lemma}

\begin{proof}
Fix $c\in\mathcal C_N(H_{\vec r})$, and let
\[
\Psi_{\vec r}(c)=(\sigma^{(1)},\dots,\sigma^{(N)}).
\]
Since $C_v$ is a clique of size $r_v$, the coloring $c$ uses $r_v$ distinct
colors on $C_v$. Therefore $v$ belongs to exactly $r_v$ of the sets
$\sigma^{(1)},\dots,\sigma^{(N)}$, and hence
\[
\sum_{i=1}^{N}\sigma^{(i)}=\vec r.
\]
It remains to show that each $\sigma^{(i)}$ is an independent set. Suppose
that $u,v\in \sigma^{(i)}$ for some $i\in[N]$. Then there exist
$a\in[r_u]$ and $b\in[r_v]$ such that
$
c(u,a)=c(v,b)=i.
$
Since $c$ is proper, $(u,a)$ and $(v,b)$ are not adjacent in $H_{\vec r}$.
In particular, one cannot have $\{u,v\}\in E$. Thus $\sigma^{(i)}$ is an
independent set, and therefore
$
\Psi_{\vec r}(c)\in\Omega^{\mathrm{ss}}_{N,\vec r}.
$

Now fix
$\vect\sigma=(\sigma^{(1)},\dots,\sigma^{(N)})
\in\Omega^{\mathrm{ss}}_{N,\vec r}$.
For each $v\in V$, let
\[
I_v:=\set{i\in[N]:v\in \sigma^{(i)}}.
\]
Note that $\size{I_v}=r_v$ for every $v\in V$. Moreover, if
$\{u,v\}\in E$, then $I_u\cap I_v=\emptyset$, since each $\sigma^{(i)}$ is
an independent set.

A coloring $c$ satisfies $\Psi_{\vec r}(c)=\vect\sigma$ if and only if, for
every $v\in V$, the colors used on $C_v$ are exactly the elements of $I_v$.
Since $C_v$ is a clique, this is equivalent to choosing a bijection from
$C_v$ onto $I_v$. The number of such choices is $r_v!$, independently for
each $v\in V$. Hence
\[
\size{\Psi_{\vec r}^{-1}(\vect\sigma)}
=
\prod_{v\in V} r_v!.
\]

The final claim follows immediately, since every element of
$\Omega^{\mathrm{ss}}_{N,\vec r}$ has the same number of preimages under
$\Psi_{\vec r}$.
\end{proof}

The following theorem proves the chain $\PssNr$ mixes rapidly by comparing it with $\Pcol$.

\begin{theorem}
\label{thm:particle-system-mixing-via-coloring}
Let $\gamma \in (0, 1)$, and let $c_{\mathrm{col}} > 1$ and $\Delta_0 > 0$ be the coloring-mixing constants from \Cref{assumption:coloring-mixing}. Let $G = (V, E)$ be a graph on $n$ vertices with maximum degree $\Delta \geq \Delta_0$. Then there exists $C_0=C_0(\Delta,c_{\mathrm{col}})$ such that,
for every $N\ge C_0$, if $\vec r\in \mathbb Z_{\ge 0}^{V}$ satisfies
 \eqref{eq:particle-system-vertex-count-assumption}, then 
 the particle system chain $\PssNr$ mixes
within total variation distance $\zeta$ of $\muNr$ in
\[
O_{\gamma,\Delta, c_{\mathrm{col}}}
\left(
Nn\log\frac{Nn}{\zeta}
\right)
\]
steps.
\end{theorem}

\begin{proof}
For every
$(v,a)\in V(H_{\vec r})$,
\[
\deg_{H_{\vec r}}(v,a)
=
(r_v-1)+\sum_{u\in\neighbor{v}}r_u
\le
\sum_{u\in\neighborCl{v}}r_u.
\]
Now we have
\[
\deg_{H_{\vec r}}(v,a)
\le
(\Delta+1)\max_{u\in V} r_u
\le
\frac{N}{c_{\mathrm{col}}},
\]
where the last inequality follows from
\eqref{eq:particle-system-vertex-count-assumption}. Hence
\[
N\ge c_{\mathrm{col}}\Delta(H_{\vec r}).
\]
Also, $\Delta (H_{\vec r}) \geq \Delta \geq \Delta_0$. Therefore, by \Cref{assumption:coloring-mixing}, $\Pcol$ mixes within total
variation distance $\zeta$ of the uniform distribution on
$\mathcal C_N(H_{\vec r})$ in
\[
O_{c_{\mathrm{col}}, \Delta_0}
\left(
|V(H_{\vec r})|\log\frac{|V(H_{\vec r})|}{\zeta}
\right)
\]
steps.

We now compare $\PssNr$ with $\Pcol$. Let
\[
m_{\vec r}:=|V(H_{\vec r})|=\sum_{v\in V}r_v.
\]
Define an auxiliary coloring chain $\widehat \Pcol$ on
$\mathcal C_N(H_{\vec r})$ as follows. At each step, choose
$v\in V$ uniformly at random and choose an ordered pair of distinct colors
$i,j\in[N]$ uniformly at random. If exactly one of the colors $i,j$ appears
on the clique $C_v$, recolor the unique vertex of $C_v$ carrying the present
color with the absent color, provided the resulting coloring is proper;
otherwise stay put. If either both colors appear on $C_v$ or neither appears
on $C_v$, stay put.

We claim that the projection of $\widehat \Pcol$ under $\Psi_{\vec r}$ is
exactly $\PssNr$. Suppose that $\Psi_{\vec r}(c)=\vect\sigma$. For each
$v\in V$, the colors used on $C_v$ are exactly the indices $k\in[N]$ such
that $v\in\sigma^{(k)}$. Thus, choosing $v,i,j$ in the coloring chain
corresponds to choosing the same $v,i,j$ in the particle system chain. If
exactly one of $i,j$ appears on $C_v$, then exactly one of the particles
$\sigma^{(i)}$ and $\sigma^{(j)}$ contains $v$. Recoloring on $C_v$ replaces
the present color by the absent color, which is the same as swapping the site
$v$ between $\sigma^{(i)}$ and $\sigma^{(j)}$. The recoloring is proper if and
only if the resulting two particle configurations are independent sets.
Therefore the projected update is exactly the update of $\PssNr$.

It remains to compare $\widehat \Pcol$ with $\Pcol$.
Consider a nontrivial recoloring move $c\to c'$ in which a vertex of some
clique $C_v$ with current color $i$ is recolored with color $j$. Under the
standard Glauber chain,
\[
\Pcol(c,c')=\frac{1}{m_{\vec r}N}.
\]
Under $\widehat \Pcol$, the same move is proposed by the two ordered choices
$(i,j)$ and $(j,i)$, and hence
\[
\widehat \Pcol(c,c')
=
\frac{2}{nN(N-1)}.
\]
Define $\beta := \frac{2 m_{\vec r}}{n (N - 1)}$. By \eqref{eq:particle-system-vertex-count-assumption}, we have
\[
\frac{2 m_{\vec r}}{n (N - 1)}\le \frac{2\max_{v\in V}r_v }{N - 1}
\le
\frac{N}{c_{\mathrm{col}}(\Delta+1)(N - 1)}.
\]
Thus if we choose $C_0 = C_0(\Delta, c_{\mathrm{col}})$ sufficiently large, for $N \geq C_0$, 
\[
\beta \leq 1.
\]
Now, note that for off-diagonal transition,
\[
\widehat \Pcol(c,c')
=
\beta\, \Pcol(c,c').
\]
Thus
\[
\widehat \Pcol=(1-\beta)I+\beta \Pcol.
\]

The lower bound $\frac{r_v}{N}\ge \gamma$ also gives
\[
m_{\vec r}=\sum_{v\in V}r_v\ge Nn \gamma.
\]

Thus
$
\beta=\Theta_{\gamma,\Delta,c_{\mathrm{col}}}(1).
$
Consequently, $\widehat \Pcol$ is a constant-factor lazy version of $\Pcol$,
and hence it mixes within total variation distance $\zeta$ of the uniform
distribution on $\mathcal C_N(H_{\vec r})$ in
\[
O_{\gamma,\Delta,c_{\mathrm{col}}}
\left(
m_{\vec r}\log \frac{m_{\vec r}}{\zeta}
\right)
\]
steps.

Since $\PssNr$ is the projection of $\widehat \Pcol$ under $\Psi_{\vec r}$,
total variation distance cannot increase under this projection. Moreover, by
\Cref{lem:configurations-to-colorings}, the uniform distribution on
$\mathcal C_N(H_{\vec r})$ is mapped by $\Psi_{\vec r}$ to $\muNr$.
Therefore $\PssNr$ has stationary distribution $\muNr$ and mixes within total
variation distance $\zeta$ of $\muNr$ in
\[
O_{\gamma,\Delta,c_{\mathrm{col}}}
\left(
m_{\vec r}\log \frac{m_{\vec r}}{\zeta}
\right)
\]
steps. Finally, since
\[
m_{\vec r}=|V(H_{\vec r})|
=
\sum_{v\in V}r_v
=
O_{\gamma, \Delta,c_{\mathrm{col}}}(Nn),
\]
the claimed bound follows.
\end{proof}
In \Cref{sec:alg-application-single-site}, we apply \Cref{thm:particle-system-mixing-via-coloring} to analyze the mixing time of our single-site particle-system sampler.

\subsection{Mean-field entropy contraction through Kac's program}
In this section we consider the continuous time nonlinear dynamics \eqref{eq:pre:continuous-nonlinear-dynamics} in the
mean-field case, i.e.\ with collision kernel $\cQ=\Qmf$, and we prove its
exponential decay in relative entropy for all sub-critical densities. Specifically, we prove the following.
\begin{theorem}[Entropy decay]
\label{thm-entropydecaymf}
Let $G = (V, E)$ be a graph on $n$ vertices with maximum degree $\D \ge 3$. For any
$\g>0$ and $\d\in(0,1)$, and any initial state $p_0$ with density
$\a(p_0) \in (\g, (1-\d)\a_c(\D))$, the solution $p_t$
of the mean-field dynamics \eqref{eq:pre:continuous-nonlinear-dynamics} satisfies
\[
  H(p_t \mid \m_{G,\l}) \;\le\; e^{-r\,t/n}\,H(p_0 \mid \m_{G,\l}),\qquad t\ge0,
\]
for some constant $r = r(\g,\d,\D) > 0$ independent of $n$, where
$\l = \l(p_0)<\l_c(\D)$ is the unique fugacity matching the
density of $p_0$.
\end{theorem}
\Cref{thm-entropydecaymf} is established by completing Kac's program: we
deduce the entropy decay of the nonlinear mean-field dynamics from a modified
log-Sobolev inequality (MLSI) for the mean-field
particle system, which in turn follows from \Cref{thm-disc-partmf}, and transfer it to
the nonlinear process by taking the limit $N\to\infty$. We start by discussing the functional inequality associated to the entropy decay of the continuous time nonlinear dynamics.
\subsubsection{Nonlinear Modified Log-Sobolev}
For any \(\l>0\), taking  $\m=\m_{G,\l}$, we consider the functional
\begin{equation}
    \cD_\m(f,f)=\frac{1}{4n^2}\sum_{u,v\in V}\sum_{\s,\s',\t,\t'} \m(\s)\m(\s')Q_{u,v}(\s,\s';\t,\t')\left(f(\s)f(\s')-f(\t)f(\t')\right)\log\left(\frac{f(\s)f(\s')}{f(\t)f(\t')}\right),
\end{equation}
for all $f:\O_G\to\bbR_+$. 
This is known as the entropy dissipation functional. Indeed, 
a direct computation using the reversibility of 
$\Qmf$ with respect to $\mu$, see \Cref{lem:pre:gibbs-stationary-measures}, shows that
the solution $p_t$ of
\eqref{eq:pre:continuous-nonlinear-dynamics} with kernel $\Qmf$, for any  initial state
$p_0\in\cP(\O_G)$, satisfies
\begin{equation}\label{h-th-mf}
\frac{d}{dt} H(p_t\mid \mu)
%\Ent_\m[f_t]
=-\cD_\m(f_t,f_t),
\end{equation}
where $f_t=p_t/\m$ is the relative density of $p_t$ with respect to $\m$. 
We refer, e.g.,  to \cite[Lemma 2.3]{caputo2025kac} for a detailed derivation. 

%Here, we use the notation \(\Ent_{\m}[f]=\m\left[f\log \left(f/\log\m[f]\right)\right]\),
%for all \(f:\O_G\to\bbR_+\). In  particular, when \(\m[f]=1\) the functional returns the relative entropy \(\Ent_{\m}[f]=H(f\m\mid \m)\).
Since 
$(x-y)\log(\tfrac{x}{y})\ge0$ for all $x,y\ge0$, we have $\cD_\m(f,f)\ge 0$, for all $f$, so that $H(p_t\mid \mu)$ %$\Ent_\m[f_t]$ 
is monotone non increasing as a function of $t\ge0$. This is the analogue of Boltzmann's  H-theorem in our setting. 

We reformulate
\Cref{thm-entropydecaymf} in terms of the following nonlinear modified
log-Sobolev inequality. We use the notation \[\Ent_{\m}[f]=\m\left[f\log \left(f/\log\m[f]\right)\right],\]
for all \(f:\O_G\to\bbR_+\). In  particular, when \(\m[f]=1\) the functional returns the relative entropy \(\Ent_{\m}[f]=H(f\m\mid \m)\).
\begin{definition}[Nonlinear Modified Log Sobolev (NMLSI)]
For all $\l>0$, let $\cS_\l$ be the set of functions $f:\O_G\to\bbR_+$ with
$\m_{G,\l}[f]=1$ such that $f\m_{G,\l}\in\cP(\O_G)$ has the same density as
$\m_{G,\l}$,
\begin{equation}
    \a\left(f\m_{G,\l}\right)=\a\left(\m_{G,\l}\right).
\end{equation}
The nonlinear dynamics is said to satisfy the
NMLSI with profile function  $\cI:(0,\infty)\to [0,\infty)$ if for all $\l>0$, 
\begin{equation}\label{def:NMLSIine}
    \cD_{\m_{G,\l}}(f,f)\ge \cI(\l)\Ent_{\m_{G,\l}}[f], \quad \forall f\in\cS_\l.
\end{equation}
\end{definition}
\medskip
Since the nonlinear mean-field dynamics conserves the density, any solution $p_t$
satisfies $\a(p_t)=\a(p_0)$ for all $t\ge 0$. As a consequence, for any initial datum $p_0\in\cP(\O_G)$, one has $f_t=p_t/\mu_{G,\l}\in \cS_\l$ for all $t\ge0$, where $\l$ is the unique value of the fugacity matching the density $\a(p_0)$. Therefore, \eqref{h-th-mf} implies that,  if
the NMLSI holds with profile $\cI(\cdot)$, then every 
solution $p_t$ satisfies
\[
H(p_t\mid\m_{G,\l})\le e^{-\cI(\l) t}H(p_0\mid \m_{G,\l}),
\]
where $\l$ is determined by $\a(p_0)=\a(\mu_{G,\l})$; see also \cite[Lemma 4.2]{caputo2025kac} for a  derivation of this statement in a slightly different context.  
In particular, \Cref{thm-entropydecaymf} is established by proving that the nonlinear system satisfies the NMLSI with a profile $\cI$ such that $\cI(\l)\ge r/n$, with a constant $r=r(\g,\d,\D)>0$, for any $\l>0$ such that $\a(\mu_{G,\l})\in(\g,(1-\d)\a_c(\D))$.

\subsubsection{Proof of \Cref{thm-entropydecaymf}}
%here we recall the main steps without proof, referring to the relevant parts
%therein.
Fix $\g>0$, $\d\in(0,1)$ and $\a\in (\g,(1-\d)\a_c(\D))$, and let
$\l>0$ be such that $\a=\a(\m_{G,\l})$.
%for a hard-core model $\m_{G,\l}$ with fugacity
%$\l\in (0,\l_c(\D))$. 
Consider the particle system defined by the chain
$\PmfN$ on the state space $\O^{\rm mf}_{N,\a_N}$, as in \Cref{defOmegaNa}, with
density
\begin{equation}
    \a_{N}=\frac{\lfloor Nn\a\rfloor}{Nn},
\end{equation}
and note that $\a_N\in (\g/2,(1-\d)\a_c(\D))$ for all $N$ large enough. Recall the
unique invariant measure $\m^{\mathrm{mf}}_{N,\a_N}$, and observe that
$\m^{\mathrm{mf}}_{N,\a_N}=\m_{G,\l}^{\otimes N}(\cdot\mid \O^{\mathrm{mf}}_{N,\a_N})$, that is the $N$-fold product of $\m_{G,\l}$ restricted to $\O^{\rm mf}_{N,\a_N}$. Pick
$f\in\cS_\l$ with $f>0$. By definition, $\nu=f\m_{G,\l}\in\cP_+(\O_G)$ with
$\a(\nu)=\a$; define the function $F_N$ on $\O_{N,\a_N}^{\mathrm{mf}}$ by
\begin{equation}\label{defobskac}
    F_N:=\frac{\nu^{\otimes N}(\cdot\mid \O_{N,\a_N}^{\mathrm{mf}})}{\m^{\mathrm{mf}}_{N,\a_N}}.
\end{equation}
Using $(x-y)\log(x/y)\ge 4 (\sqrt x - \sqrt y)^2$ one has
\[
\cE_{N,\a_N}(F_N,\log F_N)\ge 4 \cE_{N,\a_N}(\sqrt F_N,\sqrt F_N)\,.
\]
The LSI \eqref{eq:lsiHMa} with $\varrho=r_0/(Nn)$ from the proof of \Cref{thm-disc-partmf} implies  
\begin{equation}\label{in:MLSIN}
    \cE_{N,\a_N}(F_N,\log F_N)\ge \frac{4r_0}{Nn}\Ent_{\m_{N,\a_N}^{\mathrm{mf}}}[F_N]
    \,,
\end{equation}
with $r_0=r_0(\g/2,\d,\D)>0$ independent of $N,n$ as in \eqref{eq:r0}. 
%In particular, for 
%$N\ge N(\g,\d,\D)$ we apply 
%\eqref{inmlsiproof} with $\tilde\a=\a_N$ and observable $F_N$, that gives
%\begin{equation}\label{in:MLSIN}
 %    \cE_{N,\a_N}(F_N,\log F_N)\ge \frac{r}{Nn}\Ent_{\m_{N,\a_N}^{\mathrm{mf}}}[F_N], \quad \forall N\ge N(\g,\d,\D).
%\end{equation}

From now on, we proceed as
in \cite[Section~4]{caputo2025kac}. Since $\nu=f\m_{G,\l}\in\cP_+(\O_G)$ and $|\a_N-\a|\le \tfrac{1}{N}$, 
by adapting the arguments in
\cite[Proposition 4.9]{caputo2025kac}
one proves the so-called entropic
chaos, that is
\begin{align}\label{eq:ent_chaos}
    &\lim_{N\to \infty} \frac 1 N\,\Ent_{\m^{\mathrm{mf}}_{N,\a_N}} [F_N]= \Ent_{\m_{G,\l}}[f].
\end{align}
Similarly, as in \cite[Proposition 4.10]{caputo2025kac},
one finds the so-called Fisher chaos, that is
\begin{align}\label{eq:fish_chaos}
    &\lim_{N\to\infty}\,{\cE}_{N,\a_N}(F_N, \log F_N)=2\cD_\m(f,f).
\end{align}
Therefore, passing to the limit $N\to\infty$ in~\eqref{in:MLSIN} yields~\eqref{def:NMLSIine}
with
\[
\cI(\l)=\frac{2r_0}{n}\,,
\]
for every $f\in\cS_\l$ with full support in $\O_G$. It remains to extend the estimate to all  $f\in \cS_\l$. To this end, given $f\in\cS_\l$, consider
\[
 f_\e(\sigma)=(1-\e)f(\sigma)+\e, \qquad \e\in (0,1), \; \sigma\in\O_G.
\]
Note that $f_\e\in\cS_\l$ for all \(\e\in (0,1)\), since
$\m_{G,\l}[f_\e]=1$ and
$\a(f_\e\m_{G,\l})=(1-\e)\a(f\m_{G,\l})+\e\,\a(\m_{G,\l})=\a(\m_{G,\l})$. Since $f$ has full support on $\O_G$,
\[
 \cD_{\m_{G,\l}}(f_\e,f_\e)\ge \,\frac{2r_0}{n}\,\Ent_{\m_{G,\l}}[f_\e]\,.
\]
Letting $\e\to 0$, by continuity one has  \[
\cD_{\m_{G,\l}}(f_\e,f_\e)\to \cD_{\m_{G,\l}}(f,f), \quad \Ent_{\m_{G,\l}}[f_\e]\to \Ent_{\m_{G,\l}}[f]\,,
\]
which conclude the proof of the NMLSI with a profile $\cI$ such that $\cI(\l)\ge r/n$, with  $r=2r_0(\g/2,\d,\D)>0$, for any $\l>0$ such that $\a(\mu_{G,\l})\in(\g,(1-\d)\a_c(\D))$.
\qed
\section{Algorithmic Applications}
\label{sec:alg}

We now turn to the algorithmic problem of sampling from the hard-core model
when the desired conserved quantities are prescribed in advance. %There are
%two natural versions of this problem. 
In the mean-field setting, a target
density $\alpha$ is given, and we want to sample from the unique
uniform-fugacity Gibbs measure with density $\alpha$. In the single-site setting, a target marginal vector $(m_v)_{v\in V}$ is given, and we
want to sample from the unique hard-core Gibbs measure whose vertex marginals
are given by $(m_v)_{v \in V}$.

\begin{remark}\label{rem:sample}
Drawing one sample from a distribution with prescribed vertex marginals
can be done efficiently. Given a graph $G=(V,E)$ and a marginal vector
$\vect m=(m_v)_{v\in V}$ satisfying
\[
\max_{v\in V}m_v\le \frac{1}{\Delta+1},
\]
%the following %algorithm 
\Cref{alg:alg:prescribed-marginals-initialization} below outputs an independent set $\sigma\in\OmegaG$ with
\[
\Pr[\sigma_v=1]=m_v
\qquad \forall v\in V.
\]

\begin{algorithm}[ht]
\caption{Initialization with prescribed marginals}
\label{alg:alg:prescribed-marginals-initialization}
\begin{algorithmic}[1]
\Require Graph $G=(V,E)$ with maximum degree $\Delta$, and marginals
$\vect m=(m_v)_{v\in V}$ satisfying
$\max_{v\in V}m_v\le 1/(\Delta+1)$.
\Ensure A random independent set $\sigma\in\OmegaG$ with
$\Pr[\sigma_v=1]=m_v$ for all $v\in V$.

\State Find a proper coloring $c:V\to[\Delta+1]$.
\State Choose a color $i\in[\Delta+1]$ uniformly at random.
\State Initialize $\sigma_v\gets 0$ for all $v\in V$.
\ForAll{$v\in c^{-1}(i)$}
    \State Set $\sigma_v\gets 1$ with probability $(\Delta+1)m_v$, independently.
\EndFor
\State \Return $\sigma$.
\end{algorithmic}
\end{algorithm}

Each color class is an independent set, so the output belongs to
$\OmegaG$. Moreover, for every $v\in V$,
\[
\Pr[\sigma_v=1]
=
\frac1{\Delta+1}\cdot(\Delta+1)m_v
=
m_v.
\]
For prescribed density $\alpha$, the same algorithm applied to the constant
marginal vector $m_v=\alpha$ gives an initial distribution with density $\alpha$,
provided $\alpha\le 1/(\Delta+1)$.
\end{remark}

In order to sample from the maximum entropy distributions with the prescribed marginals or density, one possible approach is to use the nonlinear dynamics directly. First, we
choose an initial distribution with the prescribed conserved quantity as in \Cref{rem:sample}:
density $\alpha$ in the mean-field case, or vertex marginals $(m_v)_{v \in V}$ in the
single-site case. Then we run the corresponding nonlinear dynamics until it
is close to its limiting Gibbs measure. The issue is that this dynamics is not
easy to implement efficiently. For example, one can use the derivation tree
construction of \cite{caputo2024nonlineardynamicsisingmodel}: to sample from
the distribution obtained after $T$ nonlinear steps, draw $2^T$ independent
samples from an initial distribution $\mu_0\in\Pdist(\OmegaG)$ with the
prescribed conserved quantity, place them at the leaves of a binary tree of
height $T$, and apply the exchange kernel level by level from the leaves to
the root. The configuration at the root has the law obtained after $T$ steps
of the nonlinear dynamics. However, this method already requires $2^T$
initial samples, and is therefore not efficient.

For this reason, we do not use the nonlinear dynamics itself as the algorithm. 
Instead, we use the particle system introduced in 
\Cref{subsec:quant:particle-system}. In the prescribed-density case, we work on
the fixed total-occupation space $\OmegaNa$ and use the mean-field particle
chain $\PmfN$. In the prescribed-marginal case, we work on the fixed vertex-count
space $\OmegaNr$ and use the single-site particle chain $\PssNr$. In both cases,
we run the corresponding particle-system chain close to stationarity and output
one particle.

\subsection{Sampling with prescribed density} 
\label{sec:alg-application-mean-field}
Here we consider the mean-field inverse sampling problem, and prove \Cref{thm:intro-alg}. 
% We now consider the mean-field inverse sampling problem, where the prescribed
% conserved quantity is the density $\alpha$. 
Fix $\gamma, \delta \in (0,1)$. Throughout this subsection we assume that
\begin{equation}
\label{eq:alg:prescribed-density-assumption}
\gamma \leq \alpha \leq (1 - \delta) \alpha_{c}(\Delta).    
\end{equation}
%One can show that there exists a unique fugacity 
Let $\lambda_\alpha > 0$ be the unique fugacity for which the density of  $\mu_{G, \lambda_{\alpha}}$ is exactly $\alpha$. One can show that $\lambda_{\alpha}$ lies in the regime 
\begin{equation}\label{eq:lala}
\gamma \leq \lambda_{\alpha} \leq (1 - \delta) \lambda_{c}(\Delta).
\end{equation}
See the beginning of \Cref{app:bisection-prescribed-density} for the detailed proof.\\
The goal is to approximately sample from $\mu_{G, \lambda_{\alpha}}$. Recall the notation 
\[
\mu_G^{\alpha} = \mu_{G, \lambda_{\alpha}}.
\]
For sufficiently large $N$, we write 
\[
 M=\floor{Nn\alpha},
\]
and consider the space $\OmegaNa$ with uniform measure $\muNm$. We first observe that, if $\vect \sigma$ is a uniform draw from $\muNm$, then the law of $\sigma^{(1)}$ is a good approximation for $\mu_G^{\alpha}$. After that, we show that the mean-field particle system chain $\PmfN$
mixes rapidly to $\muNm$.

The next lemma proves the first step. 

\begin{lemma}\label{lem:alg:fixed-density-kac-chaos}
Assume $\Delta \geq 3$, and $\alpha$ satisfies \eqref{eq:alg:prescribed-density-assumption}. Then there exist constants
$C=C(\gamma,\delta,\Delta)>0$ and
$N_0=N_0(\gamma,\delta,\Delta)$ such that, for every $N\ge N_0$,
\[
\left\|
P_1\muNm-\mu_G^{\alpha}
\right\|_{\TV}
\le
\frac{C}{N}.
\]
\end{lemma}
\begin{proof}
We write
$
\mu:=\mu_G^{\alpha}.
$
%By \eqref{eq:lala}
% the bounds on
% $\lambda_\alpha$ stated at the beginning of this subsection,
% \[
% \gamma\le \lambda_\alpha\le (1-\delta)\lambda_c(\Delta).
% \]
All constants below depend only on $\gamma,\delta$, and $\Delta$.
Note that
\[
\muNm=\mu^{\otimes N}(\,\cdot \mid \OmegaNm).
\]
By duality of total variation distance, it suffices to show that for every
$f:\OmegaG\to\mathbb R$,
\begin{equation}\label{eq:alg:kac-chaos-goal}
\left|
\E_{\vect\sigma\sim\muNm}\!\left[f(\sigma^{(1)})\right]
-
\E_{\sigma\sim\mu}[f(\sigma)]
\right|
\le
\frac{C}{N}\|f\|_\infty .
\end{equation}

Let
\[
X_N:=M(\vect\sigma)=\sum_{i=1}^N|\sigma^{(i)}|
\qquad
\text{when } \vect\sigma\sim\mu^{\otimes N}.
\]
Set
\[
a_N:=\E_{\vect\sigma\sim\mu^{\otimes N}}[X_N]=Nn\alpha,
\qquad
v_N^2:=\Var_{\vect\sigma\sim\mu^{\otimes N}}(X_N).
\]
Since $\OmegaNm=\{\vect\sigma:M(\vect\sigma)=M\}=\{X_N=M\}$, we have
\begin{equation}\label{eq:alg:kac-chaos-ratio}
\E_{\vect\sigma\sim\muNm}\!\left[f(\sigma^{(1)})\right]
-
\E_{\sigma\sim\mu}[f(\sigma)]
=
\frac{\mathcal N(f)}
{\mu^{\otimes N}(X_N=M)},
\end{equation}
where
\[
\mathcal N(f)
:=
\E_{\vect\sigma\sim\mu^{\otimes N}}\!\left[
f(\sigma^{(1)})\mathbf 1\{X_N=M\}
\right]
-
\E_{\sigma\sim\mu}[f(\sigma)]\,
\mu^{\otimes N}(X_N=M).
\]
We estimate the denominator and the numerator in
\eqref{eq:alg:kac-chaos-ratio} separately.

\medskip
\noindent\textbf{Denominator.}
Let
\(
Y:=|\sigma|\), 
when  $\sigma\sim\mu$,
and set
\(
v_1^2:=\Var_{\sigma\sim\mu}(Y).
\)
Applying the variance estimate
\cite[Lemma~3.2]{jain2021approximatecountingsamplinglocal} %applies with
%additive gap $\delta\lambda_c(\Delta)$ from the threshold. Since also
%$\lambda_\alpha\ge \gamma$, 
there exist constants $c_0,C_0>0$, depending
only on $\gamma,\delta$, and $\Delta$, such that
\[
c_0 n
\le
v_1^2
\le
C_0 n.
\]
Equivalently, since $v_N^2=Nv_1^2$,
\[
c_0 Nn
\le
v_N^2
\le
C_0 Nn.
\]

Recall that 
$
G_N
$
is the disjoint union of $N$ copies of $G$. Thus $G_N$ has $Nn$ vertices
and maximum degree at most $\Delta$. Moreover, the hard-core Gibbs measure
on $G_N$ with fugacity $\lambda_\alpha$ is exactly $\mu^{\otimes N}$, and
$X_N$ is the size of the independent set sampled from this measure.
Therefore, since 
\[
|M-a_N|\le 1,
\]
by the local central limit theorem
\cite[Corollary~13]{jain2023optimalmixingdownupwalk}, applied to $G_N$,
we have
\[
\left|
\mu^{\otimes N}(X_N=M)
-
\frac{1}{\sqrt{2\pi}\,v_N}
\exp\left(
-\frac{(M-a_N)^2}{2v_N^2}
\right)
\right|
\le
C_1(Nn)^{-3/2},
\]
for some constant $C_1>0$.

% Since $M=\floor{Nn\alpha}$ and $a_N=Nn\alpha$, we have
% \[
% |M-a_N|\le 1.
% \]
We first lower bound the Gaussian leading term. Using $v_N^2\le C_0Nn$,
we get
\[
\frac{1}{v_N}
\ge
\frac{1}{\sqrt{C_0Nn}}.
\]
Using $v_N^2\ge c_0Nn$ and $|M-a_N|\le 1$, we also get
\[
\frac{(M-a_N)^2}{2v_N^2}
\le
\frac{1}{2c_0Nn}.
\]
Hence, for all sufficiently large $N$,
% \[
% \exp\left(
% -\frac{(M-a_N)^2}{2v_N^2}
% \right)
% \ge
% \frac12.
% \]
% Therefore,
\[
\frac{1}{\sqrt{2\pi}\,v_N}
\exp\left(
-\frac{(M-a_N)^2}{2v_N^2}
\right)
\ge
\frac{1}{2\sqrt{2\pi C_0}}\frac{1}{\sqrt{Nn}}.
\]
Therefore, 
% On the other hand,
% \[
% C_1(Nn)^{-3/2}
% =
% \frac{C_1}{Nn}\frac{1}{\sqrt{Nn}},
% \]
% and hence, 
for all sufficiently large $N$,
% \[
% C_1(Nn)^{-3/2}
% \le
% \frac{1}{4\sqrt{2\pi C_0}}\frac{1}{\sqrt{Nn}}.
% \]
% Combining this with the local central limit theorem gives
\begin{align}
\mu^{\otimes N}(X_N=M)
&\ge
\frac{1}{\sqrt{2\pi}\,v_N}
\exp\left(
-\frac{(M-a_N)^2}{2v_N^2}
\right)
-
C_1(Nn)^{-3/2} \nonumber \\
&\ge
\frac{1}{2\sqrt{2\pi C_0}}\frac{1}{\sqrt{Nn}}
-
\frac{1}{4\sqrt{2\pi C_0}}\frac{1}{\sqrt{Nn}} 
=
\frac{c_1}{\sqrt{Nn}},
\label{eq:alg:kac-chaos-denominator}
% \frac{1}{4\sqrt{2\pi C_0}}\frac{1}{\sqrt{Nn}}.
\end{align}
where
\(
c_1:=\frac{1}{4\sqrt{2\pi C_0}}.
\)
% we obtain
% \begin{equation}\label{eq:alg:kac-chaos-denominator}
% \mu^{\otimes N}(X_N=M)
% \ge
% \frac{c_1}{\sqrt{Nn}}.
% \end{equation}

\medskip
\noindent\textbf{Numerator.}
Define the normalized characteristic function
\[
\psi(t)
:=
\E_{\sigma\sim\mu}\left[
\exp\left(
\frac{it(Y-n\alpha)}{v_1}
\right)
\right].
\]
Using again the additive gap $\delta\lambda_c(\Delta)$ from the threshold,
the characteristic function estimate
\cite[Lemma~3.5]{jain2021approximatecountingsamplinglocal}, applied with
$\lambda=\lambda_\alpha$, gives a constant $c>0$, depending only on
$\delta$ and $\Delta$, such that
\[
|\psi(t)|
\le
\exp\left(
-\frac{c\lambda_\alpha n t^2}{v_1^2}
\right)
\qquad
\forall\,t\in[-\pi v_1,\pi v_1].
\]
Since $\lambda_\alpha\ge \gamma$, setting $c_2:=c\gamma$ gives
\[
|\psi(t)|
\le
\exp\left(
-\frac{c_2 n t^2}{v_1^2}
\right)
\qquad
\forall\,t\in[-\pi v_1,\pi v_1].
\]

We shall estimate the numerator $\mathcal N(f)$ with $f$ replaced by a recentered function $\tilde f$, and then control the recentering term separately. 

To this end, note that for any function $h:\Omega_G\to\mathbb R$, Fourier inversion gives
\begin{equation}\label{eq:alg:kac-chaos-fourier-improved}
\mathcal N(h)
=
\frac{1}{2\pi v_N}
\int_{-\pi v_N}^{\pi v_N}
e^{-it(a_N-M)/v_N}
\psi(t/\sqrt N)^{N-1}
\Cov_{\mu}\left(
h,
\exp\left(
\frac{it(Y-n\alpha)}{v_N}
\right)
\right)
dt .
\end{equation}
For $t\in[-\pi v_N,\pi v_N]$, we have
$t/\sqrt N\in[-\pi v_1,\pi v_1]$. Hence, for all sufficiently large $N$,
\[
|\psi(t/\sqrt N)|^{N-1}
\le
\exp\left(
-\frac{c_2 n t^2}{2v_1^2}
\right).
\]

We now recenter $f$ by removing its component in the direction of the conserved
quantity. Set
\[
\beta_f
:=
\frac{\Cov_\mu(f,Y)}{v_1^2},
\qquad
g(\sigma)
:=
\beta_f\,(|\sigma|-n\alpha),
\qquad
\widetilde f
:=
f-g.
\]
Then
\[
\Cov_\mu(\widetilde f,Y)=0,
\qquad
\E_\mu[g]=0.
\]
Moreover, by Cauchy--Schwarz,
\[
|\beta_f|
\le
\frac{\sqrt{\Var_\mu(f)}}{v_1}
\le
\frac{\|f\|_\infty}{v_1}.
\]
Therefore
\[
\Var_\mu(g)
=
\beta_f^2 v_1^2
\le
\|f\|_\infty^2,
\]
and hence
\[
\Var_\mu(\widetilde f)
\le
2\Var_\mu(f)+2\Var_\mu(g)
\le
C\|f\|_\infty^2.
\]

By the cumulant bound of \cite[Lemma 11]{jain2023optimalmixingdownupwalk}, for every fixed $d \geq 1$,

\[
    |\kappa_d(Y)|\le O_{\D,\d,d}(n), \quad \kappa_d(Y)=\frac{d^d}{dt^d}\log \bbE_{\mu}\left[e^{tY}\right]\Big|_{t=0}.
\]
In particular, taking $d =4$, there exists a constant $C$ such that
\begin{equation}
\label{in:boundcumulants}
    |\kappa_4(Y)|\leq Cn.
\end{equation}
Recall also that $v_1^2 = \Var_{\mu}(Y) \leq C_0n$ for some constant $C_0$. Now, note that
\[
\bbE_{\m}\left[(Y-n\a)^4\right]=\kappa_4(Y)+3\Var_{\mu}(Y)^2 \leq |\kappa_4(Y)| + 3 \Var_{\mu}(Y)^2.
\]
Thus, by applying \eqref{in:boundcumulants} and the variance bound, there exists a constant $C_4$ such that
\begin{equation}\label{eq:alg:kac-chaos-fourth-moment}
\E_\mu\left[(Y-n\alpha)^4\right]
\le
C_4 n^2.
\end{equation}

Since $\Cov_\mu(\widetilde f,Y)=0$, we have
\[
\Cov_\mu\left(
\widetilde f,
\exp\left(
\frac{it(Y-n\alpha)}{v_N}
\right)
\right)
=
\Cov_\mu\left(
\widetilde f,
\exp\left(
\frac{it(Y-n\alpha)}{v_N}
\right)
-
1
-
\frac{it(Y-n\alpha)}{v_N}
\right).
\]
Using
\[
|e^{i\theta}-1-i\theta|
\le
\frac{\theta^2}{2},
\]
Cauchy--Schwarz, and \eqref{eq:alg:kac-chaos-fourth-moment}, we get
\begin{align}
&
\left|
\Cov_\mu\left(
\widetilde f,
\exp\left(
\frac{it(Y-n\alpha)}{v_N}
\right)
\right)
\right|
\notag \\
&\qquad
\le
\sqrt{\Var_\mu(\widetilde f)}
\left(
\E_\mu\left[
\left|
\exp\left(
\frac{it(Y-n\alpha)}{v_N}
\right)
-
1
-
\frac{it(Y-n\alpha)}{v_N}
\right|^2
\right]
\right)^{1/2}
\notag \\
&\qquad
\le
C\|f\|_\infty
\frac{t^2}{v_N^2}
\left(
\E_\mu\left[(Y-n\alpha)^4\right]
\right)^{1/2}
%\notag 
% \\
% &\qquad
\le
C\|f\|_\infty
\frac{t^2}{N}.
\label{eq:alg:kac-chaos-improved-covariance}
\end{align}
In the last step we used $v_N^2=Nv_1^2$ and $v_1^2\ge c_0 n$. The value of the constant $C$ may change from line to line but it remains independent of $N,n$.

Combining
\eqref{eq:alg:kac-chaos-fourier-improved}
and
\eqref{eq:alg:kac-chaos-improved-covariance}, we obtain
\begin{align}
|\mathcal N(\widetilde f)|
&\le
\frac{1}{2\pi v_N}
\int_{-\pi v_N}^{\pi v_N}
\exp\left(
-\frac{c_2 n t^2}{2v_1^2}
\right)
C\|f\|_\infty
\frac{t^2}{N}
\,dt
% \notag \\
% &\le
% \frac{C\|f\|_\infty}{N v_N}
% \notag \\
\le
\frac{C\|f\|_\infty}{N\sqrt{Nn}}.
\label{eq:alg:kac-chaos-improved-numerator}
\end{align}

Combining
\eqref{eq:alg:kac-chaos-ratio},
\eqref{eq:alg:kac-chaos-denominator}, and
\eqref{eq:alg:kac-chaos-improved-numerator}, gives
\[
\left|
\E_{\vect\sigma\sim\muNm}
\left[
\widetilde f(\sigma^{(1)})
\right]
-
\E_{\sigma\sim\mu}
\left[
\widetilde f(\sigma)
\right]
\right|
\le
\frac{C\|f\|_\infty}{N}.
\]

It remains to estimate the linear correction $g$. Since $\E_\mu[g]=0$ and,
by symmetry under $\muNm$,
\[
\E_{\vect\sigma\sim\muNm}
\left[
|\sigma^{(1)}|
\right]
=
\frac{M}{N},
\]
we have
\begin{align*}
\left|
\E_{\vect\sigma\sim\muNm}
\left[
g(\sigma^{(1)})
\right]
-
\E_{\sigma\sim\mu}[g(\sigma)]
\right|
% &=
% |\beta_f|
% \left|
% \frac{M}{N}
% -
% n\alpha
% \right|  \\
=
|\beta_f|
\frac{|M-Nn\alpha|}{N}  
\le
\frac{\|f\|_\infty}{v_1}\frac{1}{N}  
\le
\frac{C\|f\|_\infty}{N\sqrt n}.
\end{align*}
Here we used $|M-Nn\alpha|\le 1$ and $v_1^2\ge c_0 n$.

Since $f=\widetilde f+g$, we conclude that
\[
\left|
\E_{\vect\sigma\sim\muNm}
\left[
f(\sigma^{(1)})
\right]
-
\E_{\sigma\sim\mu}
\left[
f(\sigma)
\right]
\right|
\le
\frac{C\|f\|_\infty}{N}.
\]
Taking the supremum over $\|f\|_\infty\le 1$ gives
\[
\left\|
P_1\muNm-\mu_G^{\alpha}
\right\|_{\TV}
\le
\frac{C}{N}. \qedhere
\]

\end{proof}

It remains to show that we can efficiently sample from $\muNm$ by running the chain $\PmfN$. Recall that the pair $(\PmfN,\m_{N,\a}^{\mathrm{mf}})$
is reversible. The following lemma is a direct implication of \Cref{thm-disc-partmf}.

\begin{lemma}\label{lem:alg:fixed-density-particle-mixing}
Assume $\Delta\ge 3$, and suppose $\alpha$ satisfies
\eqref{eq:alg:prescribed-density-assumption}. Then, there exists a constant $N_0 := N_0(\gamma, \delta, \Delta)$ such that 
for every $\zeta\in(0,1)$ and every $ N \geq N_0$, the chain $\PmfN$ mixes within
total variation distance $\zeta$ of $\muNm$ in
\[
O_{\gamma,\delta,\Delta}\!\left(Nn\log\frac{Nn}{\zeta}\right)
\]
steps.
\end{lemma}

\begin{proof}
By \Cref{thm-disc-partmf}, applied with parameters $\gamma/2$ and $\delta/2$, there exist constants
$r=r(\gamma,\delta,\Delta)>0$ and $N_0=N_0(\gamma,\delta,\Delta)$ such that, for every
$N\ge N_0$ and every probability measure $\nu$ on $\Omega^{\mathrm{mf}}_{N,\alpha}$,
\[
H(\nu P_N^{\mathrm{mf}}\mid \mu^{\mathrm{mf}}_{N,\alpha})
\le
\left(1-\frac{r}{Nn}\right)H(\nu\mid \mu^{\mathrm{mf}}_{N,\alpha}).
\]
Iterating the entropy contraction gives
\[
H(\nu(\PmfNM)^t\mid\mu^{\mathrm{mf}}_{N,\alpha})
\le
\exp\left(-\frac{rt}{Nn}\right)
H(\nu\mid\mu^{\mathrm{mf}}_{N,\alpha}).
\]
Since $\mu^{\mathrm{mf}}_{N,\alpha}$ is uniform on $\OmegaNa$, for every initial distribution $\nu$
we have
\[
H(\nu\mid\mu^{\mathrm{mf}}_{N,\alpha})
\le
\log|\OmegaNa|
\le
Nn\log 2.
\]
Thus
\[
H(\nu(\PmfNM)^t\mid \mu^{\mathrm{mf}}_{N,\alpha})
\le
Nn\log 2\,
\exp\left(-\frac{rt}{Nn}\right).
\]
By Pinsker's inequality,
\[
\left\|
\nu(\PmfNM)^t-\mu^{\mathrm{mf}}_{N,\alpha}
\right\|_{\TV}
\le
\sqrt{
\frac{Nn\log 2}{2}
\exp\left(-\frac{rt}{Nn}\right)
}.
\]
Choosing the implicit constant sufficiently large in
\[
t=
O_{\gamma,\delta,\Delta}
\left(
Nn\log\frac{Nn}{\zeta}
\right)
\]
makes the right-hand side at most $\zeta$. This proves the claimed mixing
bound.
\end{proof}

We now describe the resulting algorithm. By
\Cref{lem:alg:fixed-density-kac-chaos}, if $\vect\sigma$ is sampled from
$\muNm$, then the first particle $\sigma^{(1)}$ has law close to
$\mu_G^{\alpha}$. Thus it remains to approximately sample from $\muNm$. For this, we run the chain $\PmfNM$, which mixes rapidly to $\muNm$ by
\Cref{lem:alg:fixed-density-particle-mixing}, and then output the first
particle.

We first explain how to initialize the chain. Since
\[
\alpha\le (1-\delta)\alphac < \frac{1}{\Delta+1},
\]
we have
\[
M=\floor{Nn\alpha}<\frac{Nn}{\Delta+1}.
\]
Find a proper $(\Delta + 1)$-coloring of $G$ by the greedy algorithm, and use
the same coloring on each of the $N$ copies of $G$ in $G_N$. Then some color
class of $G_N$ has size at least $Nn/(\Delta+1)$. Thus, choosing any $M$
vertices of this color class gives an independent set of $G_N$ of size $M$,
or equivalently an initial configuration in $\OmegaNm$.

\begin{algorithm}[ht]
\caption{Prescribed-Density Particle System Sampler}
\label{alg:prescribed-density-particle-system-sampler}
\begin{algorithmic}[1]
    \Require Graph $G=(V,E)$ with maximum degree $\Delta$, target density
    $\alpha$ satisfying the assumptions above, and accuracy $\eps\in(0,1)$.
    \Ensure A random independent set whose law is within total variation
    distance $\eps$ from $\mu_G^{\alpha}$.

    \State Choose $N=O_{\gamma,\delta,\Delta}(\frac{1}{\eps})$ sufficiently large.
    \State Set $M:=\floor{Nn\alpha}$.
    \State Find a proper $(\Delta+1)$-coloring of $G$ by the greedy algorithm,
    and use the same coloring on each of the $N$ copies of $G$ in $G_N$.
    \State Choose a color class of $G_N$ of size at least $Nn/(\Delta+1)$,
    and choose any $M$ vertices from this color class.
    \State Let $\vect\sigma_0\in\OmegaNm$ be the corresponding initial
    particle configuration.
    \State Run the particle system chain $\PmfNM$ from $\vect\sigma_0$ for
    \[
    O_{\gamma,\delta,\Delta}
    \left(
    Nn\log\frac{Nn}{\eps}
    \right)
    \]
    steps.
    \State \Return $\sigma^{(1)}$, where
    $
    \vect\sigma=(\sigma^{(1)},\ldots,\sigma^{(N)})
    $
    is the resulting particle configuration.
\end{algorithmic}
\end{algorithm}
\Cref{thm:intro-alg}
can be rephrased as follows. \begin{theorem}
\label{thm:prescribed-density-particle-system-sampler}
Let $\gamma,\delta\in(0,1)$. Let $G=(V,E)$ be a graph on $n$ vertices with maximum degree $\Delta \geq 3$, and let the target
density $\alpha$ satisfy
\[
\gamma \le \alpha \le (1-\delta)\alpha_c(\Delta).
\]
Then, for every $\eps\in(0,1)$,
\Cref{alg:prescribed-density-particle-system-sampler} outputs a random
independent set whose law is within total variation distance $\eps$ from
$\mu_G^{\alpha}$. Moreover, the running time is
\[
\widetilde O_{\gamma,\delta,\Delta}
\left(
\frac{n}{\eps}
\right).
\]
\end{theorem}

\begin{proof}
Let $C$ and $N_1$ be the constants from
\Cref{lem:alg:fixed-density-kac-chaos}, and let $N_2$ be the constant from
\Cref{lem:alg:fixed-density-particle-mixing}. We choose $N$ sufficiently
large so that
\[
N\ge
\max\left\{
N_1,N_2,\frac{2C}{\eps}
\right\}.
\]
In particular,
\[
N=O_{\gamma,\delta,\Delta}(\frac{1}{\eps}).
\]

Let $\widehat\mu$ be the law of the particle configuration after running
$\PmfNM$ for the number of steps specified in
\Cref{alg:prescribed-density-particle-system-sampler}. By
\Cref{lem:alg:fixed-density-particle-mixing}, choosing the implicit constant
in the running time sufficiently large gives
\[
\left\|
\widehat\mu-\muNm
\right\|_{\TV}
\le
\frac{\eps}{2}.
\]
Let $\nu$ denote the law of the output $\sigma^{(1)}$. Since total variation
distance cannot increase under projection to the first particle,
\[
\left\|
\nu-P_1\muNm
\right\|_{\TV}
=
\left\|
P_1\widehat\mu-P_1\muNm
\right\|_{\TV}
\le
\frac{\eps}{2}.
\]
On the other hand, by \Cref{lem:alg:fixed-density-kac-chaos},
\[
\left\|
P_1\muNm-\mu_G^{\alpha}
\right\|_{\TV}
\le
\frac{C}{N}
\le
\frac{\eps}{2}.
\]
Therefore,
\[
\|\nu-\mu_G^{\alpha}\|_{\TV}
\le
\left\|
\nu-P_1\muNm
\right\|_{\TV}
+
\left\|
P_1\muNm-\mu_G^{\alpha}
\right\|_{\TV}
\le
\eps.
\]

It remains to bound the running time. The initialization takes
$O_\Delta(Nn)$ time. The chain is run for
\[
O_{\gamma,\delta,\Delta}
\left(
Nn\log\frac{Nn}{\eps}
\right)
\]
steps, and each step can be implemented in $O_\Delta(1)$ time. Since
$N=O_{\gamma,\delta,\Delta}(\frac{1}{\eps})$, the total running time is
\[
O_{\gamma,\delta,\Delta}
\left(
\frac{n}{\eps}
\log\frac{n}{\eps}
\right)
=
\widetilde O_{\gamma,\delta,\Delta}
\left(
\frac{n}{\eps}
\right).
\]
This completes the proof.
\end{proof}

% \begin{remark}
%     There is also a classical approach to answer to the same algorithmic problem: first  estimate the fugacity $\lambda_{\alpha}$ and then sample from $\mu_{G, \lambda_{\alpha}}$ using existing samplers for the hard-core model. We present this approach in \Cref{app:bisection-prescribed-density}.
% \end{remark}

\begin{remark}
    The same approach can be used to generate $k$ independent samples from $\mu_{G}^{\alpha}$. Let $G'$ be the disjoint union of $k$ copies of $G$. The Gibbs measure on $G'$ with density $\alpha$ is
\[
    \mu_{G'}^\alpha = (\mu_G^\alpha)^{\otimes k}.
\] 
We can apply \Cref{alg:prescribed-density-particle-system-sampler} to $G'$, which has $kn$ vertices and the same maximum degree as $G$, and produce a sample whose distribution is within total variation distance $\eps$ from $(\mu_{G}^{\alpha})^{\otimes k}$ in time $$\widetilde O_{\gamma, \delta, \Delta}\left(\frac{kn}{\eps}\right).$$ 
For comparison, there is also a classical approach to address the same algorithmic problem: first  estimate the fugacity $\lambda_{\alpha}$ and then sample from $\mu_{G, \lambda_{\alpha}}$ using existing samplers for the hard-core model. We present this approach in \Cref{app:bisection-prescribed-density}. As shown there, the bisection
approach requires time
\[
\widetilde{O}_{\gamma,\delta,\Delta}\left(\frac{n}{\varepsilon^2}\right)
\]
for one sample ($k=1$). Applied to $G'$, for all $k\ge 1$, it requires time
\[
\widetilde{O}_{\gamma,\delta,\Delta}\left(\frac{kn}{\eps^2}\right)
\]
to approximate $(\mu_G^\alpha)^{\otimes k}$. Thus, the particle-system
approach improves the dependence on $\eps$ for both one sample and
$k$ independent samples.
\end{remark}

\subsection{Hardness of prescribed-density sampling} 
\label{sec:hardness-prescribed-density}
To complete our understanding of prescribed-density sampling problem, we show that the threshold $\alphac$ is also a computational barrier for this problem. 

Throughout this section we use the notion of \emph{polynomial time approximate sampler}: a polynomial time approximate sampler for $\mu_{G}^{\alpha}$ is an algorithm that, given a graph $G$ and an accuracy parameter $\eps$, in time $\mathrm{poly}(|G|, 1/\eps)$ outputs a random independent set whose law is within $\eps$ total variation distance of $\mu_{G}^{\alpha}$. 

The following is the hardness result we prove in this section. 

\begin{theorem}
    \label{thm:hardness-prescribed-density}
    Fix $\Delta \geq 3$ and $\alpha \in (\alphac, \frac{1}{\Delta + 1})$. Unless $\mathrm{NP} = \mathrm{RP}$, there is no polynomial time approximate sampler for $\mu_{G}^{\alpha}$ on graphs with maximum degree $\Delta$. 
\end{theorem}

This hardness result should be compared with the known hardness of approximately counting independent sets of size $\floor{\alpha n}$ above $\alphac$ \cite{davies2023approximatelycountingindependentsets}. Our argument is a direct adaptation of that of \cite{davies2023approximatelycountingindependentsets}. The main difference is that our target is $\mu_{G}^{\alpha}$, not a uniform distribution on fixed-size independent sets. For a suitable choice of $\lambda$, we reduce the hard problem of approximately sampling from $\mu_{G, \lambda}$ on graphs with maximum degree $\Delta$ to our problem \cite{SlySun14}.  As in \cite{davies2023approximatelycountingindependentsets}, we add many copies of cliques $K_{\Delta + 1}$ to $G$ in order to bring down the density of the resulting graph $H$ close to $\alpha$. Then, by a variance comparison we show that approximately sampling from $\mu_{H}^{\alpha}$ gives an approximate sample from $\mu_{G, \lambda}$.

\begin{lemma}
    \label{lem:hardness-density-to-tv-approx}
    Let $G$ be a graph on $n$ vertices and $\hat \alpha \geq  \alpha > 0$. Assume there exist constants $c, C$ such that for every $\beta \in [\alpha, \hat \alpha]$, 
    \[
    cn \leq \Var_{\sigma \sim \mu_{G}^{\beta}}(|\sigma|) \leq Cn.
    \]
    Then, we have
    \[
    \|\mu_{G}^{\alpha} - \mu_{G}^{\hat \alpha}\|_{\TV} \leq \frac{\sqrt{Cn}}{2c} |\hat \alpha - \alpha|.
    \]
\end{lemma}

We provide the proof of \Cref{lem:hardness-density-to-tv-approx} in \Cref{app:hardness-missing-proofs}.

\begin{lemma}
    \label{lem:hardness-clique-construction}
    Fix $\Delta \geq 3$ and $\alpha \in (\alphac, \frac{1}{\Delta + 1})$, and let 
    \[
    \lambda := \frac{\alpha}{1 - \alpha (\Delta + 1)}.
    \]
    If there is a polynomial time approximate sampler for $\mu_{G}^{\alpha}$ on graphs with maximum degree $\Delta$, then there is a polynomial time approximate sampler for $\mu_{G, \lambda}$ on graphs with maximum degree $\Delta$. 
\end{lemma}

\begin{proof}
Assume $\mathcal{O}_{\alpha}$ is an oracle that, given a graph $H$ with
maximum degree $\Delta$ and $\eps>0$, in time
$\mathrm{poly}(|H|,1/\eps)$ outputs a random independent set whose law is
within $\eps$ total variation distance of $\mu_H^{\alpha}$.

Let $G=(V,E)$ be a graph with maximum degree $\Delta$, and write $n=|V|$. We construct a graph $H$ by adding $M$ disjoint copies of complete graphs on $\Delta + 1$ vertices to $G$, i.e., 
\[
H:=G\sqcup M K_{\Delta + 1},
\]
where $M$ will be chosen later as a polynomial in $n$ and $1/\eps$. Note that the graph $H$ also has maximum degree
$\Delta$. 

Note also that $\alpha (\mu_{K_{\Delta + 1}, \lambda}) = \alpha$. 
Let
\[
\hat\alpha:=\alpha(\mu_{H,\lambda}).
\]
Since $H$ is a disjoint union, and the lowest density of $\mu_{G', \lambda}$ occurs when $G' = K_{\Delta + 1}$, we conclude $\hat \alpha \geq \alpha$. Also, 
we have
\[
\hat\alpha
=
\frac{n}{|V(H)|}\alpha(\mu_{G,\lambda})
+
\frac{M(\Delta + 1)}{|V(H)|}\alpha.
\]
Therefore,
\[
\hat\alpha
\le
\alpha + \frac{n}{|V(H)|}.
\]

We next verify that the variance assumption in
\Cref{lem:hardness-density-to-tv-approx} holds for all $\beta \in [\alpha, \hat \alpha]$. Let $\lambda_\beta$ be the fugacity such that
\[
\mu_H^\beta=\mu_{H,\lambda_\beta}.
\]
Since $H$ is a disjoint union,
\[
\Var_{\sigma\sim\mu_H^\beta}(|\sigma|)
=
\Var_{\sigma\sim\mu_{G,\lambda_\beta}}(|\sigma|)
+
M\Var_{\sigma\sim\mu_{K_{\Delta +1},\lambda_\beta}}(|\sigma|).
\]

We first prove the upper bound. Since $0\le |\sigma|\le n$ on $G$,
\[
\Var_{\sigma\sim\mu_{G,\lambda_\beta}}(|\sigma|)
\le
\frac{n^2}{4}.
\]
Also, on $K_{\Delta + 1}$, the size of an independent set is either $0$ or $1$, so
\[
\Var_{\sigma\sim\mu_{K_{\Delta + 1},\lambda_\beta}}(|\sigma|)
\le
\frac14.
\]
Therefore,
\[
\Var_{\sigma\sim\mu_H^\beta}(|\sigma|)
\le
\frac{n^2}{4}+\frac{M}{4}.
\]
If $M\ge n^2$, then
\[
\Var_{\sigma\sim\mu_H^\beta}(|\sigma|)
\le
\frac{M}{2}
\le
\frac{|V(H)|}{2}.
\]
This gives the desired upper bound with $C=1/2$.

We now prove the lower bound. Since
\[
\beta
=
\frac{n}{|V(H)|}\alpha(\mu_{G,\lambda_\beta})
+
\frac{M(\Delta + 1)}{|V(H)|}\alpha(\mu_{K_{\Delta + 1}, \lambda_{\beta}}),
\]
with a similar argument as above, we have
\[
\alpha(\mu_{K_{\Delta + 1}, \lambda_{\beta}}) \leq \beta \leq \alpha(\mu_{K_{\Delta + 1}, \lambda_{\beta}}) + \frac{n}{|V(H)|}.
\]
Moreover, since $\beta$ lies between $\alpha$ and $\hat\alpha$, and
$|\hat\alpha-\alpha|\le \frac{n}{|V(H)|}$, we get
\[
|\beta-\alpha|\le \frac{n}{|V(H)|}.
\]
Hence
\[
|\alpha(\mu_{K_{\Delta + 1}, \lambda_{\beta}})-\alpha|
\le
\frac{2n}{|V(H)|}.
\]
Since $\alpha\in(\alphac, \frac{1}{\Delta + 1})$ is fixed, by taking $M\ge A_{\alpha,\Delta} n$
for a sufficiently large constant $A_{\alpha,\Delta}$, we can ensure that
\[
\alpha(\mu_{K_{\Delta + 1}, \lambda_{\beta}})\in
\left[
\frac{\alpha}{2},
\frac{\alpha+\frac{1}{\Delta +1}}{2}
\right].
\]
On $K_{\Delta + 1}$, the random variable $|\sigma|$ is Bernoulli with parameter
\[
x_\beta
=
\Pr_{\sigma \sim \mu_{K_{\D+1},\lambda_\beta}}[|\sigma|=1]
=
(\Delta + 1) \alpha(\mu_{K_{\Delta + 1}, \lambda_{\beta}}).
\]
The previous display implies that $x_\beta$ is bounded away from both $0$
and $1$ by constants depending only on $\alpha$ and $\Delta$. Therefore, there
exists a constant $b=b(\alpha,\Delta)>0$ such that
\[
\Var_{\sigma\sim\mu_{K_{\Delta + 1},\lambda_\beta}}(|\sigma|)
=
x_\beta(1-x_\beta)
\ge b.
\]
Consequently,
\[
\Var_{\sigma\sim\mu_H^\beta}(|\sigma|)
\ge
Mb.
\]
If $M\ge n$, then
\[
|V(H)|=n+M(\Delta+1) \le M(\Delta + 2),
\]
and hence
\[
\Var_{\sigma\sim\mu_H^\beta}(|\sigma|)
\ge
\frac{b}{\Delta + 2}|V(H)|.
\]
Thus the variance lower bound in
\Cref{lem:hardness-density-to-tv-approx} holds with
\[
c=\frac{b}{\Delta + 2}.
\]

Combining the upper and lower bounds, \Cref{lem:hardness-density-to-tv-approx}
gives
\[
\|\mu_H^\alpha-\mu_H^{\hat\alpha}\|_{\TV}
\le
\frac{\sqrt{C|V(H)|}}{2c}|\hat\alpha-\alpha|
\le
\frac{\sqrt{C|V(H)|}}{2c}\cdot \frac{n}{|V(H)|}
=
\frac{\sqrt C}{2c}\frac{n}{\sqrt{ |V(H)|}}.
\]
But by definition of $\hat\alpha$, we have
\[
\mu_H^{\hat\alpha}=\mu_{H,\lambda}.
\]
Choose
\[
M=\left\lceil A_{\alpha,\Delta}\frac{n^2}{\eps^2}\right\rceil
\]
for a sufficiently large constant $A_{\alpha,\Delta}$. Then
\[
\|\mu_H^\alpha-\mu_{H,\lambda}\|_{\TV}
\le
\frac{\eps}{2}.
\]

Now run the oracle $\mathcal O_\alpha$ on $H$ with accuracy $\eps/2$, and
let $\nu$ be the law of its output. Then
\[
\|\nu-\mu_H^\alpha\|_{\TV}\le \frac{\eps}{2}.
\]
We output the restriction of this independent set to the original vertex set
$V(G)$. Let $P_G$ denote this projection. Since total variation distance
cannot increase under projection,
\[
\|P_G\nu-P_G\mu_H^\alpha\|_{\TV}\le \frac{\eps}{2}.
\]
Also,
\[
\|P_G\mu_H^\alpha-P_G\mu_{H,\lambda}\|_{\TV}
\le
\|\mu_H^\alpha-\mu_{H,\lambda}\|_{\TV}
\le
\frac{\eps}{2}.
\]
Finally, because $H=G\sqcup M K_{\Delta + 1}$ is a disjoint union,
\[
P_G\mu_{H,\lambda}=\mu_{G,\lambda}.
\]
Therefore,
\[
\|P_G\nu-\mu_{G,\lambda}\|_{\TV}
\le
\eps.
\]

Note that by the choice of $M$, $|H|$ is polynomial in $|G|$, 
and the oracle runs in time $\mathrm{poly}(|H|,1/\eps)$. Thus the whole
procedure runs in time $\mathrm{poly}(|G|,1/\eps)$.
\end{proof}

\begin{proof}[Proof of \Cref{thm:hardness-prescribed-density}]
    Set \[
    \lambda := \frac{\alpha}{1 - \alpha (\Delta + 1)}.
    \]
    Since 
    \[
    \alphac =     \frac{\lambda_c(\Delta)}{1+(\Delta+1)\lambda_c(\Delta)},
    \]
    and the map $x \mapsto \frac{x}{1 + (\Delta + 1)x}$ is strictly increasing on $(0, \infty)$, the assumption $\alpha > \alphac$ implies that 
    \[
    \lambda > \lambdac.
    \]
    Assume for contradiction that there is a polynomial time approximate sampler for $\mu_{G}^{\alpha}$ on graphs with maximum degree $\Delta$.
    By \Cref{lem:hardness-clique-construction}, this gives a polynomial time approximate sampler for $\mu_{G, \lambda}$ on graphs with maximum degree $\Delta$. But hard-core approximate sampling at fugacity $\lambda > \lambdac$ on $\Delta$-regular graphs is known to be hard \cite{SlySun14}. Thus, unless $\mathrm{NP} = \mathrm{RP}$, there is no polynomial approximate sampler for $\mu_{G}^{\alpha}$ on graphs with maximum degree $\Delta$. 
\end{proof}
\subsection{Sampling with prescribed vertex marginals}
\label{sec:alg-application-single-site}
Here we focus on the single-site inverse sampling problem and prove \Cref{thm:intro-ss}. Let
$\vect m=(m_v)_{v\in V}$ be a target marginal vector, and let
$\vect\lambda_{\vect m}\in(0,\infty)^V$ be the unique fugacity vector such
that \footnote{See \Cref{app:marginal-polytope} for a discussion of existence and uniqueness of the fugacity vector realizing a prescribed marginal vector.}
\[
\Pr_{\sigma\sim\mu_{G,\vect\lambda_{\vect m}}}[\sigma_v=1]=m_v
\qquad \forall v\in V.
\]
For brevity, throughout this subsection we write
\[
\mu_G^{\vect m}:=\mu_{G,\vect\lambda_{\vect m}}.
\]

The particle system approach is as follows. We choose $N$ sufficiently large,
set
\[
r_v=\floor{Nm_v}
\qquad \forall v\in V,
\]
and consider the space $\OmegaNr$ with uniform
measure $\muNr$. We then run the single-site particle chain $\PssNr$ close to
its stationary distribution $\muNr$ and output the first particle. Thus the law of the output is approximately 
$P_1\muNr$.

The analysis has two parts. First, we prove that $P_1\muNr$ is close to the
target measure $\mu_G^{\vect m}$. Second, we show how to sample efficiently
from $\muNr$ by running the particle chain $\PssNr$.

To prove the first step, we use a relative entropy argument, along with concentration inequalities. One could also proceed as in the mean-field case, by using a direct estimation via the local central limit theorem (\Cref{lem:alg:fixed-density-kac-chaos}). In this case, however, one has to deal with an $n$-dimensional event and the Fourier transform argument used there appears to lead to a substantially worse dependence of $N$ on $n$.\footnote{On the other hand, we note that a relative entropy argument, along with concentration inequalities, could also be used in the mean-field case, but it gives a worse dependence on $N$ than the one-dimensional local central limit theorem.}

The next lemmas provide the proof of the first step. 
\begin{lemma}\label{lem:conditioning-one-marginal}
With the above notation, assume $\OmegaNr \neq \emptyset$. Then we have
\[
D_{\mathrm{KL}}\!\left(P_1 \muNr \,\middle\|\, \mu_G^{\vect m}\right)
\le
\frac{1}{N}\log\frac{1}{\bigl(\mu_G^{\vect m}\bigr)^{\otimes N}(\OmegaNr)}.
\]
Consequently,
\[
\left\|
P_1\muNr-\mu_G^{\vect m}
\right\|_{\TV}
\le
\sqrt{
\frac{1}{2N}\log\frac{1}{\bigl(\mu_G^{\vect m}\bigr)^{\otimes N}(\OmegaNr)}
}.
\]
\end{lemma}

\begin{proof}
First, note that for $\vect\sigma \in \OmegaNr$ and
$\vect\lambda=\vect\lambda_{\vect m}$, we have
\[
\mu_{G,\vect\lambda}^{\otimes N}(\vect\sigma)
\propto \prod_{i \in [N]} \mu_{G,\vect\lambda}(\sigma^{(i)})
\propto \prod_{i \in [N]} \prod_{v \in V} \lambda_{v}^{\sigma_v^{(i)}}
=
\prod_{v \in V} \lambda_{v}^{\sum_{i \in [N]}\sigma_v^{(i)}}
=
\prod_{v \in V} \lambda_v^{r_v}.
\]
Hence, all elements of $\OmegaNr$ have the same measure under
$\bigl(\mu_G^{\vect m}\bigr)^{\otimes N}$, which means that
\[
\muNr = \bigl(\mu_G^{\vect m}\bigr)^{\otimes N}\!\left(\,\cdot\, \middle|\, \OmegaNr\right).
\]
We now have
\[
\muNr(\vect\sigma)
=
\frac{\bigl(\mu_G^{\vect m}\bigr)^{\otimes N}(\vect\sigma)}
{\bigl(\mu_G^{\vect m}\bigr)^{\otimes N}(\OmegaNr)}
\,\mathbf{1}(\vect\sigma\in\OmegaNr),
\qquad
\vect\sigma\in \OmegaG^{N}.
\]
Hence, for every $\vect\sigma\in\OmegaNr$,
\[
\frac{\muNr(\vect\sigma)}
{\bigl(\mu_G^{\vect m}\bigr)^{\otimes N}(\vect\sigma)}
=
\frac{1}{\bigl(\mu_G^{\vect m}\bigr)^{\otimes N}(\OmegaNr)}.
\]
Therefore,
\begin{equation}
\label{eq:entropy-conditioned-law}
D_{\mathrm{KL}}(\muNr\,\|\,\bigl(\mu_G^{\vect m}\bigr)^{\otimes N})
=
\sum_{\vect\sigma\in \OmegaNr}
\muNr(\vect\sigma)
\log\frac{1}{\bigl(\mu_G^{\vect m}\bigr)^{\otimes N}(\OmegaNr)}
=
\log\frac{1}{\bigl(\mu_G^{\vect m}\bigr)^{\otimes N}(\OmegaNr)}.
\end{equation}

Now, by the chain rule for relative entropy,
\[
D_{\mathrm{KL}}(\muNr\,\|\,\bigl(\mu_G^{\vect m}\bigr)^{\otimes N})
=
\sum_{i=1}^{N}
\E_{\vect\sigma \sim \muNr}\!\left[
D_{\mathrm{KL}}
\bigl(
\muNr(\sigma^{(i)}\mid \sigma^{(1)},\dots,\sigma^{(i-1)})
\,\|\,\mu_G^{\vect m}
\bigr)
\right].
\]
By convexity of relative entropy in the first argument,
\[
\E_{\vect\sigma \sim \muNr}\!\left[
D_{\mathrm{KL}}
\bigl(
\muNr(\sigma^{(i)}\mid \sigma^{(1)},\dots,\sigma^{(i-1)})
\,\|\,\mu_G^{\vect m}
\bigr)
\right]
\ge
D_{\mathrm{KL}}(P_i\muNr\,\|\,\mu_G^{\vect m}).
\]
Thus,
\[
D_{\mathrm{KL}}(\muNr\,\|\,\bigl(\mu_G^{\vect m}\bigr)^{\otimes N})
\ge
\sum_{i=1}^{N}D_{\mathrm{KL}}(P_i\muNr\,\|\,\mu_G^{\vect m})
=
N \cdot D_{\mathrm{KL}}(P_1\muNr\,\|\,\mu_G^{\vect m}).
\]

Combining this with \eqref{eq:entropy-conditioned-law} gives
\[
D_{\mathrm{KL}}(P_1\muNr\,\|\,\mu_G^{\vect m})
\le
\frac{1}{N}\log\frac{1}{\bigl(\mu_G^{\vect m}\bigr)^{\otimes N}(\OmegaNr)}.
\]

Finally, by Pinsker's inequality,
\[
\|P_1\muNr -\mu_G^{\vect m}\|_{\TV}
\le
\sqrt{
\frac{1}{2}
D_{\mathrm{KL}}(P_1\muNr\,\|\,\mu_G^{\vect m})
},
\]
and therefore
\[
\|P_1\muNr-\mu_G^{\vect m}\|_{\TV}
\le
\sqrt{
\frac{1}{2N}\log\frac{1}{\bigl(\mu_G^{\vect m}\bigr)^{\otimes N}(\OmegaNr)}
}. \qedhere
\]
\end{proof}

By \Cref{lem:conditioning-one-marginal}, to find an upper bound on the total
variation distance, it suffices to provide a lower bound on the probability
$\bigl(\mu_G^{\vect m}\bigr)^{\otimes N}(\OmegaNr)$. The next two lemmas establish such
a lower bound for the case where $\vec r$ is an integer vector such that
$\|\vec r-N\vect m\|_{\infty}$ is $O(\sqrt{N\log n})$.

\begin{lemma}
\label{lem:column-size-counting-comparison}
Assume that
\[
\gamma \le m_u
\le \frac{1-\delta}{\Delta+1} \qquad \forall u \in V
\]
for some constants $\gamma>0$ and $\delta\in(0,1)$. Then there exists
$
C = C(\gamma, \delta, \Delta)
$
such that for all $N\ge C\log(4n)$, the following hold.

First, for every $\vec s\in \mathbb Z_{\geq 0}^V$ satisfying
\[
\|\vec s-N\vect m\|_{\infty}\le \sqrt{\frac N2 \log(4n)},
\]
one has
$
\Omega^{\mathrm{ss}}_{N,\vec s}\neq \emptyset.
$

Second, let $\vec r,\vec r'\in \mathbb Z^{V}_{\geq 0}$ satisfy
\[
\|\vec r-N\vect m\|_{\infty}\le \sqrt{\frac N2 \log(4n)},
\qquad
\|\vec r'-N\vect m\|_{\infty}\le \sqrt{\frac N2 \log(4n)},
\]
and suppose that for some $v\in V$,
\[
r'_v=r_v+1,
\qquad
r'_u=r_u \quad \text{for all }u\in V\setminus\{v\}.
\]
Then
\[
\frac{\gamma}{2}
\le
\frac{\size{\Omega^{\mathrm{ss}}_{N,\vec r}}}
{\size{\Omega^{\mathrm{ss}}_{N,\vec r'}}}
\le
\frac{2}{\delta}.
\]
\end{lemma}

\begin{proof}
Let $C=C(\gamma, \delta, \Delta)$ be sufficiently large so that for every
$N \geq C\log(4n)$ one has
\[
2(\Delta+1)\sqrt{\frac N2\log(4n)}\le \delta N
\qquad\text{and}\qquad
\sqrt{\frac N2 \log(4n)} \le \frac{\gamma}{2}N.
\]

We now prove the first statement. Fix $\vec s\in\mathbb Z_{\geq 0}^V$ satisfying
$\|\vec s-N\vect m\|_{\infty}\le \sqrt{\frac N2\log(4n)}$.
We show that $\Omega^{\mathrm{ss}}_{N,\vec s}\neq\emptyset$. Indeed, since $G$ has maximum degree at most $\Delta$, it admits a proper
$(\Delta+1)$-coloring. Let
$
V=V_1\cup\cdots\cup V_{\Delta+1}
$
be the corresponding partition of $V$ into independent sets. For each
$j\in[\Delta+1]$, define
\[
M_j:=\max_{u\in V_j} s_u,
\]
with the convention that $M_j = 0$ if $V_j = \emptyset$.
Since $\|\vec s-N\vect m\|_{\infty}\le \sqrt{\frac N2\log(4n)}$, we have for every
$u\in V$,
\[
s_u\le Nm_u+\sqrt{\frac N2\log(4n)}
\le \frac{1-\delta}{\Delta+1}N+\sqrt{\frac N2\log(4n)}.
\]
Hence
\[
M_j\le \frac{1-\delta}{\Delta+1}N+\sqrt{\frac N2\log(4n)}
\qquad
\text{for all }j\in[\Delta+1].
\]
Therefore,
\[
\sum_{j=1}^{\Delta+1} M_j
\le
(1-\delta)N+(\Delta+1)\sqrt{\frac N2\log(4n)}
\le N.
\]

Now, partition $[N]$ into pairwise disjoint sets
$
I_1,\dots,I_{\Delta+1},I_0
$
such that $\size{I_j}=M_j$ for each $j\in[\Delta+1]$. For each
$j\in[\Delta+1]$ and each $u\in V_j$, since $s_u\le M_j$, we may choose a set
\[
J_u\subseteq I_j
\qquad\text{with}\qquad
\size{J_u}=s_u.
\]
Now, for each $t\in I_j$, define
\[
\sigma^{(t)}:=\set{u\in V_j: t\in J_u},
\]
and for each $t\in I_0$, set $\sigma^{(t)}:=\emptyset$.

Since each $V_j$ is an independent set, every $\sigma^{(t)}$ is an independent
set. Moreover, for each $u\in V_j$, the vertex $u$ belongs to exactly the sets
$\sigma^{(t)}$ with $t\in J_u$, and hence appears exactly $s_u$ times among
$\sigma^{(1)},\dots,\sigma^{(N)}$. Therefore
\[
\sum_{t=1}^{N}\sigma^{(t)}=\vec s,
\]
which shows that
$
(\sigma^{(1)},\dots,\sigma^{(N)})\in\Omega^{\mathrm{ss}}_{N,\vec s}.
$
Hence $\Omega^{\mathrm{ss}}_{N,\vec s}\neq\emptyset$.

Note that the first part of the lemma shows that the ratio we want to bound in
the second part is well-defined, and thus now we can prove the second statement.

Let
\[
L:=\Omega^{\mathrm{ss}}_{N,\vec r},
\qquad
R:=\Omega^{\mathrm{ss}}_{N,\vec r'},
\]
and define a bipartite graph $H=(L\cup R,F)$ as follows.
For $\vect\sigma=(\sigma^{(1)},\dots,\sigma^{(N)})\in L$ and
$\vect\tau=(\tau^{(1)},\dots,\tau^{(N)})\in R$, we place an edge between
$\vect\sigma$ and $\vect\tau$ if there exists $j\in[N]$ such that
\[
\tau^{(j)}=\sigma^{(j)}\cup\{v\},
\qquad
\tau^{(k)}=\sigma^{(k)} \quad \text{for all } k\neq j,
\]
and $\sigma^{(j)}\cup\{v\}$ is an independent set. Equivalently,
$\vect\sigma$ is obtained from $\vect\tau$ by deleting $v$ from exactly one
coordinate $j$ for which $v\in\tau^{(j)}$. Clearly, every vertex of $H$ has degree at most $N$.

We first bound the degree of vertices in $R$ from below. Fix
$\vect\tau\in R$. Then
$\deg_H(\vect\tau) = r'_v$,
since each coordinate $i\in[N]$ with $v\in \tau^{(i)}$ yields one neighbor in
$L$, obtained by deleting $v$ from $\tau^{(i)}$. Therefore,
\[
\deg_H(\vect\tau)
=
r'_v
\ge
Nm_v-\sqrt{\frac N2 \log(4n)}
\ge
\gamma N-\sqrt{\frac N2 \log(4n)}
\ge
\frac{\gamma}{2}N.
\]

Next, we bound the degree of vertices in $L$ from below. Fix
$\vect\sigma\in L$, and let
\[
M
:=
\sum_{i\in[N]} \size{\sigma^{(i)}_{\neighborCl{v}}}.
\]
Let
\[
B:=\set{i\in[N] : \sigma^{(i)}_{\neighborCl{v}}\neq\emptyset }.
\]
Note that if $i\notin B$, then
$v$ can be added to $\sigma^{(i)}$ while preserving independence. Thus
\[
\deg_H(\vect\sigma) = N-\size{B}.
\]
Since each $i\in B$ contributes at least $1$ to the sum defining $M$, we have
$\size{B}\le M$, and therefore
\[
\deg_H(\vect\sigma) \ge N-M.
\]

On the other hand, since $\vect\sigma\in\Omega^{\mathrm{ss}}_{N,\vec r}$,
\[
M
=
\sum_{u\in \neighborCl{v}} r_u.
\]
Using $\size{\neighborCl{v}}\le \Delta+1$ and the bound
$r_u\le Nm_u+\sqrt{\frac N2 \log(4n)}$ for every $u\in V$, we obtain
\[
M
\le
(\Delta+1)\left(\max_{u\in V} m_u \cdot N + \sqrt{\frac N2 \log(4n)}\right)
\le
(\Delta+1)\left(\frac{1-\delta}{\Delta+1}N + \sqrt{\frac N2 \log(4n)}\right).
\]
Hence
\[
M
\le
(1-\delta)N+(\Delta+1)\sqrt{\frac N2 \log(4n)},
\]
and so
\[
\deg_H(\vect\sigma)
\ge
N-M
\ge
\delta N-(\Delta+1)\sqrt{\frac N2 \log(4n)}
\ge
\frac{\delta}{2}N.
\]

We now double-count the edges of $H$. Since every vertex has degree at most $N$,
while every vertex in $L$ has degree at least $\frac{\delta}{2}N$ and every
vertex in $R$ has degree at least $\frac{\gamma}{2}N$, we have
\[
\frac{\delta}{2}N\,\size{L}
\le
\size{F}
\le
N\,\size{R},
\quad \text{and}\quad 
\frac{\gamma}{2}N\,\size{R}
\le
\size{F}
\le
N\,\size{L}.
\]
Therefore,
\[
\frac{\size{L}}{\size{R}}
\le
\frac{2}{\delta}
\qquad\text{and}\qquad
\frac{\size{L}}{\size{R}}
\ge
\frac{\gamma}{2}.
\]
Since $L=\Omega^{\mathrm{ss}}_{N,\vec r}$ and
$R=\Omega^{\mathrm{ss}}_{N,\vec r'}$, this gives
\[
\frac{\gamma}{2}
\le
\frac{\size{\Omega^{\mathrm{ss}}_{N,\vec r}}}
{\size{\Omega^{\mathrm{ss}}_{N,\vec r'}}}
\le
\frac{2}{\delta},
\]
as claimed.
\end{proof}

\begin{lemma}
\label{lem:min-box-mass}
Assume that
\[
\gamma \le m_u \leq \frac{1-\delta}{\Delta+1} \qquad \forall u \in V
\]
for some constants $\gamma>0$ and $\delta\in(0,1)$. Then there exists
$
C=C(\gamma, \delta, \Delta)
$
such that for all $N\ge C\log(4n)$, if we define
\[
L:=\sqrt{\frac N2 \log(4n)}
\qquad\text{and}\qquad
\mathcal{B}:=\set{\vec r\in\mathbb Z^V:\|\vec r-N\vect m\|_{\infty}\le L},
\]
then
\[
\min_{\vec r\in\mathcal{B}}
\bigl(\mu_G^{\vect m}\bigr)^{\otimes N}(\Omega^{\mathrm{ss}}_{N,\vec r})
\ge
\frac{1}{2}(2L+1)^{-n}
\left(\frac{\delta\gamma}{2}\right)^{2nL}.
\]
\end{lemma}
\begin{proof}
Let $C=C(\gamma, \delta, \Delta)$ be as in
\Cref{lem:column-size-counting-comparison}, and fix
$N\ge C\log(4n)$. For $\vec r\in \mathcal B$,
we first show that
\[
\sum_{\vec r\in\mathcal B} \bigl(\mu_G^{\vect m}\bigr)^{\otimes N}(\Omega^{\mathrm{ss}}_{N,\vec r})\ge \frac12.
\]
Let $\vect\sigma=(\sigma^{(1)},\dots,\sigma^{(N)})\sim \bigl(\mu_G^{\vect m}\bigr)^{\otimes N}$, and set
$
\vec R:=\sum_{i=1}^N \sigma^{(i)}.
$
Fix $v\in V$. Then
$
R_v=\sum_{i=1}^N \sigma_v^{(i)},
$
where $\sigma_v^{(1)},\dots,\sigma_v^{(N)}$ are independent $\{0,1\}$-valued random variables with
$
\E[\sigma_v^{(i)}]=m_v.
$
Hence, by Hoeffding's inequality,
\[
\bigl(\mu_G^{\vect m}\bigr)^{\otimes N}\!\left(|R_v-Nm_v|>L\right)
\le
2\exp\!\left(-\frac{2L^2}{N}\right)
=
2e^{-\log(4n)}
=
\frac{1}{2n}.
\]
Applying the union bound over all $v\in V$, we obtain
\[
\bigl(\mu_G^{\vect m}\bigr)^{\otimes N}\!\left(\|\vec R-N\vect m\|_\infty>L\right)
\le
\frac12.
\]
Therefore,
\[
\sum_{\vec r\in\mathcal B} \bigl(\mu_G^{\vect m}\bigr)^{\otimes N}(\Omega^{\mathrm{ss}}_{N,\vec r})
=
\bigl(\mu_G^{\vect m}\bigr)^{\otimes N}\!\left(\|\vec R-N\vect m\|_\infty\le L\right)
\ge
\frac12.
\]

We next compare the probabilities of neighboring vectors. Let $\vec r,\vec r'\in\mathcal B$ satisfy
\[
r'_v=r_v+1,
\qquad
r'_u=r_u \quad \text{for all }u\in V\setminus\{v\},
\]
for some $v\in V$. Let $\vect\lambda=\vect\lambda_{\vect m}$. Since all elements of
$\Omega^{\mathrm{ss}}_{N,\vec r}$ have the same measure under
$\bigl(\mu_G^{\vect m}\bigr)^{\otimes N}$, we have
\[
\frac{
\bigl(\mu_G^{\vect m}\bigr)^{\otimes N}(\Omega^{\mathrm{ss}}_{N,\vec r})
}{
\bigl(\mu_G^{\vect m}\bigr)^{\otimes N}(\Omega^{\mathrm{ss}}_{N,\vec r'})
}
=
\frac{|\Omega^{\mathrm{ss}}_{N,\vec r}|}{|\Omega^{\mathrm{ss}}_{N,\vec r'}|}
\prod_{u\in V}\lambda_u^{r_u-r_u'}
=
\frac{|\Omega^{\mathrm{ss}}_{N,\vec r}|}{|\Omega^{\mathrm{ss}}_{N,\vec r'}|}\cdot \frac{1}{\lambda_v}.
\]
By \Cref{lem:column-size-counting-comparison},
\[
\frac{\gamma}{2}
\le
\frac{|\Omega^{\mathrm{ss}}_{N,\vec r}|}{|\Omega^{\mathrm{ss}}_{N,\vec r'}|}
\le
\frac{2}{\delta}.
\]

We now bound the fugacities. For each $v\in V$, one has
\[
m_v
=
\lambda_v\,
\mu_G^{\vect m}\!\left(\sigma_u=0 \text{ for all }u\in \neighborCl{v}\right).
\]
Hence
$
\lambda_v\ge m_v\ge \gamma.
$
Also, by the union bound,
\[
\mu_G^{\vect m}\!\left(\sigma_u=0 \text{ for all }u\in \neighborCl{v}\right)
\ge
1-\sum_{u\in \neighborCl{v}} m_u
\ge
1-(\Delta+1)\cdot\frac{1-\delta}{\Delta+1}
=
\delta,
\]
and therefore
\[
\lambda_v
=
\frac{m_v}{\mu_G^{\vect m}(\sigma_u=0 \text{ for all }u\in \neighborCl{v})}
\le
\frac{1}{\delta}.
\]
Thus
$
\gamma\le \lambda_v\le \frac{1}{\delta},
$
which implies
$$
\frac{\delta\gamma}{2}
\le
\frac{
\bigl(\mu_G^{\vect m}\bigr)^{\otimes N}(\Omega^{\mathrm{ss}}_{N,\vec r})
}{
\bigl(\mu_G^{\vect m}\bigr)^{\otimes N}(\Omega^{\mathrm{ss}}_{N,\vec r'})
}
\le
\frac{2}{\delta\gamma}.
$$

Now let $\vec r_*\in\mathcal B$ be such that
\[
\bigl(\mu_G^{\vect m}\bigr)^{\otimes N}(\Omega^{\mathrm{ss}}_{N,\vec r_*})
=
\min_{\vec r\in\mathcal B}
\bigl(\mu_G^{\vect m}\bigr)^{\otimes N}(\Omega^{\mathrm{ss}}_{N,\vec r}).
\]
Fix any $\vec s\in\mathcal B$. There exists a path
$
\vec r_*=\vec r^{(0)},\vec r^{(1)},\dots,\vec r^{(m)}=\vec s
$
in $\mathcal B$ such that consecutive points differ in exactly one coordinate by $\pm 1$, and
$
m=\|\vec s-\vec r_*\|_1.
$
Applying the previous bound along this path gives
\[
\bigl(\mu_G^{\vect m}\bigr)^{\otimes N}(\Omega^{\mathrm{ss}}_{N,\vec s})
\le
\left(\frac{2}{\delta\gamma}\right)^m
\bigl(\mu_G^{\vect m}\bigr)^{\otimes N}(\Omega^{\mathrm{ss}}_{N,\vec r_*})
=
\left(\frac{2}{\delta\gamma}\right)^{\|\vec s-\vec r_*\|_1}
\bigl(\mu_G^{\vect m}\bigr)^{\otimes N}(\Omega^{\mathrm{ss}}_{N,\vec r_*}).
\]
Since $\vec s,\vec r_*\in\mathcal B$,
\[
\|\vec s-\vec r_*\|_1
\le
\sum_{u\in V}\bigl(|s_u-Nm_u|+|r_{*,u}-Nm_u|\bigr)
\le
2nL.
\]
Hence
\[
\bigl(\mu_G^{\vect m}\bigr)^{\otimes N}(\Omega^{\mathrm{ss}}_{N,\vec s})
\le
\left(\frac{2}{\delta\gamma}\right)^{2nL}
\bigl(\mu_G^{\vect m}\bigr)^{\otimes N}(\Omega^{\mathrm{ss}}_{N,\vec r_*}).
\]

Also, for each $u\in V$, the coordinate $r_u$ ranges over an interval of length $2L$, and hence has at most $2L+1$ possible integer values. Therefore
$
|\mathcal B|\le (2L+1)^n.
$
Summing over $\vec s\in\mathcal B$, we obtain
\[
\frac12
\le
\sum_{\vec s\in\mathcal B}
\bigl(\mu_G^{\vect m}\bigr)^{\otimes N}(\Omega^{\mathrm{ss}}_{N,\vec s})
\le
|\mathcal B|
\left(\frac{2}{\delta\gamma}\right)^{2nL}
\bigl(\mu_G^{\vect m}\bigr)^{\otimes N}(\Omega^{\mathrm{ss}}_{N,\vec r_*})
\le
(2L+1)^n
\left(\frac{2}{\delta\gamma}\right)^{2nL}
\bigl(\mu_G^{\vect m}\bigr)^{\otimes N}(\Omega^{\mathrm{ss}}_{N,\vec r_*}).
\]
Rearranging yields
\[
\bigl(\mu_G^{\vect m}\bigr)^{\otimes N}(\Omega^{\mathrm{ss}}_{N,\vec r_*})
\ge
\frac12(2L+1)^{-n}
\left(\frac{\delta\gamma}{2}\right)^{2nL}.
\]
we conclude that
\[
\min_{\vec r\in\mathcal B}
\bigl(\mu_G^{\vect m}\bigr)^{\otimes N}(\Omega^{\mathrm{ss}}_{N,\vec r})
=
\bigl(\mu_G^{\vect m}\bigr)^{\otimes N}(\Omega^{\mathrm{ss}}_{N,\vec r_*})
\ge
\frac12(2L+1)^{-n}
\left(\frac{\delta\gamma}{2}\right)^{2nL}.
\]
This proves the lemma.
\end{proof}

\begin{corollary}
\label{cor:particle-marginal-approximation}
By \Cref{lem:conditioning-one-marginal,lem:column-size-counting-comparison,lem:min-box-mass}, for every
$\vec r\in\mathcal B$,
\[
\|P_1\mu^{\mathrm{ss}}_{N,\vec r}-\mu_G^{\vect m}\|_{\TV}
\le
\sqrt{
\frac{1}{2N}
\left(
\log 2
+
n\log(2L+1)
+
2nL\log\frac{2}{\delta\gamma}
\right)
},
\]
where
\[
L=\sqrt{\frac N2\log(4n)}.
\]
In particular, there exists $C=C(\gamma, \delta, \Delta)$ such that if
\[
N\ge C\frac{n^2\log (4n)}{\eps^4},
\]
then
\[
\|P_1\mu^{\mathrm{ss}}_{N,\vec r}-\mu_G^{\vect m}\|_{\TV}\le \eps
\qquad \forall \vec r\in\mathcal B.
\]
Moreover, one has $\vec r=\floor{N\vect m}\in\mathcal B$, and hence
\[
\|P_1\mu^{\mathrm{ss}}_{N,\vec r}-\mu_G^{\vect m}\|_{\TV}\le \eps.
\]
\end{corollary}

We now consider sampling from $\mu^{\mathrm{ss}}_{N,\vec r}$ by running
the particle-system chain $\PssNr$. Let $\gamma \in (0, 1)$, and let $c_{\mathrm{col}} > 1$ and $\Delta_0$ be the coloring-mixing constants in \Cref{assumption:coloring-mixing}. Assume that
\begin{equation}
    \label{eq:alg:single-site-marginal-assumption}
    \Delta \geq \Delta_0, \qquad 
    \gamma \leq m_v \leq \frac{1}{c_{\mathrm{col}}(\Delta + 1)} \qquad \forall v \in V. 
\end{equation}
Note that
\[
\frac{r_v}{N} = \frac{\floor{Nm_v}}{N} \leq \frac{1}{c_{\mathrm{col}}(\Delta + 1)}.
\]
Also, if $N \geq \frac{2}{\gamma}$
\[
\frac{r_v}{N} = \frac{\floor{Nm_v}}{N} \geq m_v - \frac1N \geq \frac{\gamma}{2}.
\]
Thus, for $N \geq \frac{2}{\gamma}$, $\vec r$ satisfies the assumption in \eqref{eq:particle-system-vertex-count-assumption} with parameter $\frac{\gamma}{2}$. Thus, applying \Cref{thm:particle-system-mixing-via-coloring} with parameter $\frac{\gamma}{2}$, we conclude that for every $ N \geq N_0(\Delta, c_{\mathrm{col}})$, $\PssNr$ mixes within total variation distance $\zeta$ of $\muNr$ in 
\[
O_{\gamma,\Delta, c_{\mathrm{col}}}
\left(
Nn\log\frac{Nn}{\zeta}
\right)
\]
steps.

We are now ready to state the resulting algorithm. First, we give a
deterministic procedure for constructing an initial configuration in
$\Omega^{\mathrm{ss}}_{N,\vec r}$. This will be used as the starting point
for the particle-system chain in the final sampling algorithm. \Cref{alg:initial-particle-configuration} constructs an element of
$\Omega^{\mathrm{ss}}_{N,\vec r}$ whenever
\[
(\Delta + 1)\max_{v\in V}r_v \leq N.
\]
In our application, this condition holds for $r_v=\floor{Nm_v}$ by
\eqref{eq:alg:single-site-marginal-assumption}. Note that this is exactly the
constructive procedure used in the proof of
\Cref{lem:column-size-counting-comparison} to show that
$\Omega^{\mathrm{ss}}_{N,\vec r}\neq\emptyset$.

\begin{algorithm}[H]
\caption{Constructing an Initial Particle Configuration}
\label{alg:initial-particle-configuration}
\begin{algorithmic}[1]
\Require Graph $G=(V,E)$ with maximum degree $\Delta$, an integer vector
$\vec r\in\mathbb Z_{\ge 0}^{V}$  satisfying
$\max_{v\in V}r_v\le \frac{N}{\Delta + 1}$, and an integer $N$.
\Ensure A configuration $\vect\sigma_0\in\Omega^{\mathrm{ss}}_{N,\vec r}$.

\State Find a proper $(\Delta+1)$-coloring of $G$, and let
\[
V=V_1\cup\cdots\cup V_{\Delta+1}
\]
be the corresponding partition into independent sets.

\State Set
\[
M_k:=\max_{v\in V_k} r_v
\quad \text{with } M_k=0 \text{ if } V_k=\emptyset \qquad \forall k \in [\Delta + 1]
\]

\State Partition $[N]$ into pairwise disjoint sets
\[
I_1,\ldots,I_{\Delta+1},I_0
\]
such that $\size{I_k}=M_k$ for every $k\in[\Delta+1]$.

\State For every $k\in[\Delta+1]$ and every $v\in V_k$, choose
\[
J_v\subseteq I_k
\qquad\text{with}\qquad
\size{J_v}=r_v.
\]

\State Define
\[
\sigma^{(t)}:=\set{v\in V:t\in J_v}
\qquad
\text{for every }t\in[N].
\]

\State \Return $\vect\sigma_0=(\sigma^{(1)},\ldots,\sigma^{(N)})$.

\end{algorithmic}
\end{algorithm}

\begin{algorithm}[H]
\caption{Particle System Sampler}
\label{alg:prescribed-marginal-particle-system-sampler}
\begin{algorithmic}[1]
\Require Graph $G=(V,E)$ with maximum degree $\Delta$, target marginals
$\vect m=(m_v)_{v\in V}$ satisfying
\eqref{eq:alg:single-site-marginal-assumption}, and accuracy
$\eps\in(0,1)$.
\Ensure A random independent set whose law is within total variation distance
$\eps$ from $\mu_G^{\vect m}$.

\State Choose
\[
N
\ge
C\frac{n^2\log (4n)}{\eps^4},
\]
where $C=C(\gamma,\Delta,c_{\mathrm{col}})$ is a sufficiently large
constant.

\State Set
\[
r_v:=\floor{Nm_v}
\qquad
\text{for every }v\in V.
\]

\State Use \Cref{alg:initial-particle-configuration} to construct an initial
configuration
\[
\vect\sigma_0\in\OmegaNr.
\]

\State Run the particle system chain $\PssNr$ from $\vect\sigma_0$ for
\[
O_{\gamma,\Delta, c_{\mathrm{col}}}
\left(
Nn\log\frac{Nn}{\eps}
\right)
\]
steps.

\State \Return $\sigma^{(1)}$, where
\[
\vect\sigma=(\sigma^{(1)},\ldots,\sigma^{(N)})
\]
is the resulting particle configuration.
\end{algorithmic}
\end{algorithm}

The following result implies  \Cref{thm:intro-ss}; see also \Cref{rem:coloptbound}.
\begin{theorem}
\label{thm:particle-system-sampler}
Let $\gamma \in (0, 1)$, and let $c_{\mathrm{col}} > 1$ and $\Delta_0$ be the coloring-mixing constants in \Cref{assumption:coloring-mixing}. Let $G = (V, E)$ be a graph on $n$ vertices with maximum degree $\Delta \geq \Delta_0$, and let $\vect m$ be a marginal vector satisfying 
\[
\gamma \leq m_v \leq \frac{1}{c_{\mathrm{col}}(\Delta + 1)} \qquad \forall v \in V.
\]
Then, for every $\eps\in(0,1)$, 
\Cref{alg:prescribed-marginal-particle-system-sampler} outputs a random independent set whose law
is within total variation distance $\eps$ from $\mu_G^{\vect m}$. Moreover,
the running time is
\[
\widetilde O_{\gamma,\Delta, c_{\mathrm{col}}}
\left(
\frac{n^3}{\eps^4}
\right).
\]
\end{theorem}
\begin{proof}
By choosing the constant $C=C(\gamma,\Delta,c_{\mathrm{col}})$
sufficiently large, \Cref{cor:particle-marginal-approximation}, applied with
$\delta= 1 - \frac{1}{c_{\mathrm{col}}}$ and accuracy $\eps/2$, gives
\[
\left\|
P_1\muNr-\mu_G^{\vect m}
\right\|_{\TV}
\le
\frac{\eps}{2}.
\]
Also, by \Cref{thm:particle-system-mixing-via-coloring}, after running
$\PssNr$ for
\[
O_{\gamma,\Delta, c_{\mathrm{col}}}
\left(
Nn\log\frac{Nn}{\eps}
\right)
\]
steps, the law $\widehat\mu$ of the resulting particle configuration satisfies
\[
\left\|
\widehat\mu-\muNr
\right\|_{\TV}
\le
\frac{\eps}{2}.
\]
Let $\nu$ denote the law of the output $\sigma^{(1)}$. Since total variation
distance cannot increase under projection to the first particle,
\[
\left\|
\nu-P_1\muNr
\right\|_{\TV}
=
\left\|
P_1\widehat\mu-P_1\muNr
\right\|_{\TV}
\le
\frac{\eps}{2}.
\]
Therefore,
\[
\|\nu-\mu_G^{\vect m}\|_{\TV}
\le
\left\|
\nu-P_1\muNr
\right\|_{\TV}
+
\left\|
P_1\muNr-\mu_G^{\vect m}
\right\|_{\TV}
\le
\eps.
\]

Finally, by the choice of $N$,
\[
N=
O_{\gamma,\Delta, c_{\mathrm{col}}}
\left(
\frac{n^2\log n}{\eps^4}
\right).
\]
The initial construction takes $O_{\Delta}(Nn)$ time, and the running time of
the chain is
\[
O_{\gamma,\Delta, c_{\mathrm{col}}}
\left(
Nn\log\frac{Nn}{\eps}
\right).
\]
Thus the total running time is
\[
\widetilde O_{\gamma,\Delta, c_{\mathrm{col}}}
\left(
\frac{n^3}{\eps^4}
\right). \qedhere
\]
\end{proof}

\begin{remark}
    There is also a classical approach to the same algorithmic problem: first, estimate the fugacity vector $\vect \lambda_{m}$, and then sample from $\mu_{G, \vect \lambda_{m}}$ using existing samplers for the hard-core model. We present this approach in \Cref{app:gradient-descent-prescribed-marginals}.
\end{remark}

\section{Discussion and Open Questions}\label{sec:open}
In this paper, we have advanced the understanding of nonlinear
exchange dynamics for spin systems by deriving sharp convergence
rate results for the hard-core model using novel techniques.
We have also translated these results into efficient algorithms for sampling from the hard-core model with prescribed density or prescribed marginals that are simpler to implement and more direct than algorithms derived from traditional learning approaches.

Our results leave open several interesting questions for further research, some of which we enumerate below.
\begin{enumerate}
\item %We conjecture that, by analogy with linear Markov chains, exponential convergence to stationarity should hold for general irreducible nonlinear exchange dynamics, at least at some finite rate (which may be even exponentially small). Currently, this has to be established on an ad hoc basis for individual cases, such as the Ising model~\cite{caputo2024nonlineardynamicsisingmodel} and the hard-core model in our \Cref{thm:quant:mean-field-rapid-convergence,thm:quant:single-site-rapid-convergence}.
By analogy with linear Markov chains, where irreducibility implies exponential convergence to stationarity (albeit possibly with a very bad rate constant), it is natural to conjecture that exponential convergence actually holds throughout the entire irreducibility region covered by \Cref{thm:intro-convergence}.  
This would amount to an exponential ergodicity property for the nonlinear evolution. Currently, however, the arguments underlying the proof of \Cref{thm:intro-convergence} yield convergence to equilibrium but do not provide any quantitative rate, while the exponential decay  has to be established on an ad hoc basis for individual cases, such as in the high-temperature Ising model \cite{caputo2024nonlineardynamicsisingmodel} and, under low-density assumptions, in our \Cref{th:informal1,thm-informal-ct-ent}.

    \item Our convergence rate result for the single-site nonlinear dynamics in~\Cref{th:informal1} is valid only for densities up to the threshold $\frac{1}{3(\D+1)}$; however, we conjecture that the same rate of convergence should hold throughout the regime for which we prove convergence in \Cref{thm:intro-convergence}, 
    namely up to the threshold $\frac{1}{\D+1}$. 
    Similarly, we conjecture that our algorithmic result in \Cref{thm:particle-system-sampler} also extends to the threshold $\frac{1}{\D+1}$.
    As mentioned in \Cref{rem:coloptbound}, this threshold is achieved under the widely believed conjecture that the Glauber dynamics for $q$-colorings mixes rapidly for $q \geq \Delta + 2$. It remains an interesting open question to resolve our conjecture directly, without relying on the coloring conjecture. Specifically, the reduction to colorings yields a highly structured graph, and it would be sufficient to show that the Glauber dynamics mixes rapidly on such graphs.

    \item We derived our sampling algorithms by first defining and analyzing an appropriate nonlinear dynamics, which cannot be directly implemented, and then constructing an associated particle system that asymptotically approximates the nonlinear dynamics and that can be efficiently implemented. However, this approach is a little unsatisfactory because the analysis of the particle system proceeds independently of that of the nonlinear dynamics, and does not ultimately make use of our rapid convergence results for those dynamics. It would be interesting to establish a more direct efficient implementation of the nonlinear dynamics that inherits its convergence rate.
\end{enumerate}

\bibliographystyle{alpha}
\bibliography{bib}

\bigskip
\appendix
\section{Algorithms via optimization}
In this appendix, we answer the same algorithmic questions addressed in \Cref{sec:alg}, this time using a more traditional optimization
approach. 

We start with some preliminary observations about the marginal polytope.

\subsection{The marginal polytope}
\label[appsec]{app:marginal-polytope}
Let
\[
\mathcal{P}_{\mathrm{IS}}
:=
\operatorname{conv}\{\sigma:\sigma\in\OmegaG\}
\subseteq \mathbb R^V
\]
be the independent set  polytope. Equivalently,
$\mathcal{P}_{\mathrm{IS}}$ is the set of all feasible one vertex marginal vectors of
probability measures supported on $\OmegaG$. Also, we write $\mathcal{P}_{\mathrm{IS}}^\circ$ for its interior.

We first note that the marginal assumption used in \Cref{thm:intro-convergence} implies that
the prescribed vector belongs to the interior of this polytope.

\begin{lemma}\label{lem:app:small-marginals-in-polytope}
Let $\vect m=(m_v)_{v\in V}$ satisfy
\begin{equation}\label{eq:veri}
0\le m_v\le \frac1{\Delta+1}
\qquad \forall v\in V.
\end{equation}
Then $\vect m\in \mathcal{P}_{\mathrm{IS}}$. Moreover, if for some
$\gamma,\delta>0$,
\[
\gamma\le m_v
\le \frac{1-\delta}{\Delta+1} \qquad \forall v \in V,
\]
then $\vect m\in\mathcal{P}_{\mathrm{IS}}^\circ$.
\end{lemma}

\begin{proof}
The first statement follows from the initialization construction in
\Cref{alg:alg:prescribed-marginals-initialization}, which gives a
probability distribution on $\OmegaG$ with one-vertex marginals
$\vect m$ whenever $\max_v m_v\le 1/(\Delta+1)$.

For the second statement, choose
\[
\eta
:=
\frac12
\min\left\{
\gamma,\frac{\delta}{\Delta+1}
\right\}.
\]
If $\|\vect m-\vect m'\|_\infty\le\eta$, then
$\vect m'$ satisfies \eqref{eq:veri}.
% \[
% 0\le m'_v\le \frac1{\Delta+1}
% \qquad \forall v\in V.
% \]
Hence every such $\vect m'$ belongs to $\mathcal{P}_{\mathrm{IS}}$ by the first part.
Thus an $\ell^\infty$-neighborhood of $\vect m$ is contained in
$\mathcal{P}_{\mathrm{IS}}$, and so $\vect m\in\mathcal{P}_{\mathrm{IS}}^\circ$.
\end{proof}

By the standard canonical parameter correspondence for finite exponential families, for $\vect m\in\mathcal P_{\mathrm{IS}}^\circ$, there is a unique fugacity vector $\vect\lambda\in(0,\infty)^V$ whose hard-core measure has marginal vector $\vect m$; see, e.g., \cite{WainwrightJordan2008,BGS14}. Thus, \Cref{lem:app:small-marginals-in-polytope} verifies that the marginal assumptions used in the main text indeed correspond to a well-defined hard-core Gibbs measure $\mu_{G}^{\vect m}$.

\subsection{Gradient descent for prescribed marginals}
\label[appsec]{app:gradient-descent-prescribed-marginals}
Here we address the question of sampling from a hard-core model with given single-site marginals. 
The method we present is an adaptation of the approach of \cite{DiakonikolasKaneStewartSun2021}.

Note that the hard-core model can be viewed as a canonical
exponential family. In general, an exponential family has the form
\[
P_{\vect \theta}(x)
=
\exp\!\left(\langle \vect \theta,T(x)\rangle-A(\theta)\right),
\]
where $T(x)$ is the sufficient statistic and $A(\theta)$ is the log-partition
function. For the hard-core model, the sufficient statistic is simply the
independent-set vector $\sigma$. In the notation of the preliminaries, for
$\vect\theta\in\mathbb R^V$,
\[
\mu_{G,\vect\theta}(\sigma)
=
\exp\left(
\langle\vect\theta,\sigma\rangle-\AG(\vect\theta)
\right),
\qquad
\sigma\in\OmegaG.
\]
Equivalently, $\theta_v=\log\lambda_v$, where $\lambda_v$ is the fugacity of
$v$.

Let $
\mathcal{P}_{\mathrm{IS}}
$
be the independent set polytope introduced in \Cref{app:marginal-polytope}. %In \Cref{app:marginal-polytope} we 
%As noted that, 
As already noted, for every
$\vect m\in\mathcal{P}_{\mathrm{IS}}^{\circ}$, there is a unique finite canonical parameter
$\vect\theta_{\vect m}\in\mathbb R^V$ such that
\[
\E_{\mu_{G,\vect\theta_{\vect m}}}[\sigma]=\vect m.
\]
% This is the usual backward map from mean parameters to canonical parameters;
% see, for example,
% \cite{WainwrightJordan2008,BGS14}.

The problem of finding the fugacities corresponding to a prescribed marginal
vector $\vect m$ is therefore the problem of computing this backward map. The
standard dual objective for this problem is
\[
F_{\vect m}(\vect\theta)
:=
\AG(\vect\theta)-\langle \vect\theta,\vect m\rangle .
\]
Indeed, the Fenchel conjugate of $\AG$ is
\[
\AG^*(\vect m)
:=
\sup_{\vect\theta\in\mathbb R^V}
\left\{
\langle \vect m,\vect\theta\rangle-\AG(\vect\theta)
\right\}.
\]
For $\vect m\in\mathcal{P}_{\mathrm{IS}}^\circ$, the supremum is attained at the unique
parameter $\vect\theta_{\vect m}$ satisfying
\[
\nabla\AG(\vect\theta_{\vect m})=\vect m.
\]
Thus, for the target marginal vector $\vect m$, the desired log-fugacity vector
is characterized by
\[
\vect\theta_{\vect m}
\in
\argmax_{\vect\theta\in\mathbb R^V}
\left\{
\langle \vect m,\vect\theta\rangle-\AG(\vect\theta)
\right\}.
\]
Equivalently,
\[
\vect\theta_{\vect m}
\in
\argmin_{\vect\theta\in\mathbb R^V}
\left\{
\AG(\vect\theta)-\langle \vect\theta,\vect m\rangle
\right\}
=
\argmin_{\vect\theta\in\mathbb R^V}F_{\vect m}(\vect\theta).
\]
Thus minimizing $F_{\vect m}$ is exactly the convex-dual formulation of
recovering the canonical parameter whose Gibbs measure has marginals
$\vect m$. Projected gradient descent for this type of exponential-family
objective, including with approximate gradients, has been used in related
works; see, for example,
\cite{WainwrightJordan2008,BGS14,DiakonikolasKaneStewartSun2021}.

We now compute the derivative of $F_{\vect m}$. For every $v\in V$,
\[
\frac{\partial}{\partial \theta_v}\AG(\vect\theta)
=
\E_{\mu_{G,\vect\theta}}[\sigma_v].
\]
Hence
\[
\nabla F_{\vect m}(\vect\theta)
=
\E_{\mu_{G,\vect\theta}}[\sigma]-\vect m.
\]
Therefore, the minimizer $\vect\theta_*=\vect\theta_{\vect m}$ of
$F_{\vect m}$ satisfies
\[
\E_{\mu_{G,\vect\theta_*}}[\sigma]=\vect m.
\]
Thus gradient descent on $F_{\vect m}$ has a simple interpretation: if the
current marginal of a vertex $v$ is larger than the target value $m_v$, then
the update decreases the corresponding log-fugacity $\theta_v$; if it is
smaller than $m_v$, then the update increases $\theta_v$. In this way, the
algorithm iteratively adjusts the fugacities to reduce the discrepancy between
the current marginals and the target marginals.

Moreover,
\[
\nabla^2F_{\vect m}(\vect\theta)
=
\Cov_{\mu_{G,\vect\theta}}(\sigma).
\]
Consequently, a uniform lower bound on the minimum eigenvalue of
$\Cov_{\mu_{G,\vect\theta}}(\sigma)$ gives strong convexity, while a uniform
upper bound on the maximum eigenvalue gives smoothness. Once these two bounds
are available, we can apply projected gradient descent to $F_{\vect m}$. In
our setting, the gradient need not be computed exactly; instead, we estimate
$\E_{\mu_{G,\vect\theta}}[\sigma]$ by approximately sampling from
$\mu_{G,\vect\theta}$.

Throughout, we assume that
\begin{equation}
\label{eq:app:gd-marginal-assumption}
\gamma
\le
m_v
\le
(1-\delta)\alphac \qquad \forall v \in V,
\end{equation}
where $\gamma,\delta>0$ are fixed constants. The lower bound in
\eqref{eq:app:gd-marginal-assumption} will be used to obtain a uniform lower
bound on the minimum eigenvalue of the covariance matrix, and hence strong
convexity. Note that since $\alphac<1/(\Delta+1)$,
\eqref{eq:app:gd-marginal-assumption} implies that there exists
$\delta'>0$ such that
\[
\max_{v\in V}m_v
\le
\frac{1-\delta'}{\Delta+1}.
\]
Hence, by \Cref{lem:app:small-marginals-in-polytope}, the assumption implies
$\vect m\in\mathcal{P}_{\mathrm{IS}}^\circ$.

We now state the consequence of
\eqref{eq:app:gd-marginal-assumption} for the fugacities. Let
$\mu=\mu_{G,\vect\lambda_{\vect m}}$. For every $v\in V$, we have
\[
m_v
=
(\lambda_{\vect m})_v\,
\mu\!\left(
\sigma_u=0 \text{ for all }u\in\neighborCl{v}
\right).
\]
Therefore, since the probability above is at most one,
\[
(\lambda_{\vect m})_v\ge m_v\ge \gamma.
\]
On the other hand, by the union bound,
\[
\mu\!\left(
\sigma_u=0 \text{ for all }u\in\neighborCl{v}
\right)
\ge
1-\sum_{u\in\neighborCl{v}}m_u
\ge
1-(\Delta+1)\alphac.
\]
Hence
\[
(\lambda_{\vect m})_v
=
\frac{m_v}
{\mu(\sigma_u=0 \text{ for all }u\in\neighborCl{v})}
\le
\frac{(1-\delta)\alphac}
{1-(\Delta+1)\alphac}.
\]
Since
\[
\lambdac
=
\frac{\alphac}
{1-(\Delta+1)\alphac},
\]
we conclude that
\[
(\lambda_{\vect m})_v
\le
(1-\delta)\lambdac
\qquad
\text{for all }v\in V.
\]
Thus the desired fugacity vector lies uniformly below the uniqueness
threshold. In particular, if we define
\[
\lambda_{\max}:=(1-\delta)\lambdac,
\]
then
\[
\gamma
\le
(\lambda_{\vect m})_v
\le
\lambda_{\max}
\qquad
\text{for all }v\in V.
\]
Equivalently, the desired log-fugacity vector
$\vect\theta_*=\vect\theta_{\vect m}$ belongs to the box
\[
\mathcal D
:=
\left\{
\vect\theta\in\mathbb R^V:
\log\gamma
\le
\theta_v
\le
\log\lambda_{\max}
\text{ for all }v\in V
\right\}.
\]
We will run projected gradient descent on this set $\mathcal D$.

The next lemma gives the strong convexity of $F_{\vect m}$ on
$\mathcal D$. The proof uses the same conditioning idea as in
\cite{davies2023approximatelycountingindependentsets}: after choosing a large independent set, we
condition on its complement, so that the remaining available vertices behave
as independent Bernoulli random variables.

\begin{lemma}\label{lem:app:gd-strong-convexity-hard-core}
For every $\vect\theta\in\mathcal D$, one has
\[
\Cov_{\mu_{G,\vect\theta}}(\sigma)
\succeq
\ell I,
\]
where
\[
\ell
:=
\frac{\gamma}{(\Delta+1)(1+\lambda_{\max})^{\Delta+2}}.
\]
Consequently, $F_{\vect m}$ is $\ell$-strongly convex on $\mathcal D$.
\end{lemma}

\begin{proof}
Fix $\vect\theta\in\mathcal D$, and write $\lambda_v=e^{\theta_v}$. Then
\[
\gamma\le \lambda_v\le \lambda_{\max}
\qquad
\forall v\in V.
\]
It is enough to show that for every $a\in\mathbb R^V$,
\[
\Var_{\mu_{G,\vect\theta}}\!\left(\sum_{v\in V}a_v\sigma_v\right)
\ge
\ell\|a\|_2^2.
\]

Since $G$ has maximum degree at most $\Delta$, we can color it using
$\Delta+1$ colors. Let
$
V=V_1\cup\ldots\cup V_{\Delta+1}
$
be the corresponding partition of $V$ into independent sets. Choose
$i\in[\Delta+1]$ such that
\[
\sum_{v\in V_i}a_v^2
\ge
\frac{1}{\Delta+1}\|a\|_2^2.
\]
By the law of total variance,
\[
\Var_{\mu_{G,\vect\theta}}\!\left(\sum_{v\in V}a_v\sigma_v\right)
\ge
\E_{\mu_{G,\vect\theta}}\!\left[
\Var_{\mu_{G,\vect\theta}}\!\left(
\sum_{v\in V_i}a_v\sigma_v
\,\middle|\,
\sigma_{V\setminus V_i}
\right)
\right].
\]
Since $V_i$ is an independent set, once $\sigma_{V\setminus V_i}$ is fixed,
the variables $(\sigma_v)_{v\in V_i}$ are conditionally independent. Moreover,
for each $v\in V_i$, if some neighbor of $v$ is occupied, then $\sigma_v=0$
deterministically; otherwise $\sigma_v$ is Bernoulli with parameter
$\lambda_v/(1+\lambda_v)$. Therefore,
\[
\Var_{\mu_{G,\vect\theta}}\!\left(
\sigma_v
\,\middle|\,
\sigma_{V\setminus V_i}
\right)
=
\mathbf 1\{\sigma_u=0\text{ for all }u\in\neighbor{v}\}
\frac{\lambda_v}{(1+\lambda_v)^2}.
\]
Hence
\[
\begin{aligned}
&\E_{\mu_{G,\vect\theta}}\!\left[
\Var_{\mu_{G,\vect\theta}}\!\left(
\sum_{v\in V_i}a_v\sigma_v
\,\middle|\,
\sigma_{V\setminus V_i}
\right)
\right] \\
&\qquad =
\sum_{v\in V_i}
a_v^2\,
\mu_{G,\vect\theta}\!\left(\sigma_u=0\text{ for all }u\in\neighbor{v}\right)
\frac{\lambda_v}{(1+\lambda_v)^2}.
\end{aligned}
\]

We now lower bound the probability that all neighbors of $v$ are unoccupied.
Reveal the vertices of $\neighbor{v}$ one at a time. At each step, regardless
of the previously revealed values, the conditional probability that the next
vertex is unoccupied is at least $1/(1+\lambda_{\max})$. Indeed, if all its
neighbors are unoccupied, its conditional probability of being occupied is at
most $\lambda_{\max}/(1+\lambda_{\max})$, and otherwise it is forced to be
unoccupied. Therefore,
\[
\mu_{G,\vect\theta}\!\left(\sigma_u=0\text{ for all }u\in\neighbor{v}\right)
\ge
\frac{1}{(1+\lambda_{\max})^\Delta}.
\]
Also,
\[
\frac{\lambda_v}{(1+\lambda_v)^2}
\ge
\frac{\gamma}{(1+\lambda_{\max})^2}.
\]
Thus,
\[
\mu_{G,\vect\theta}\!\left(\sigma_u=0\text{ for all }u\in\neighbor{v}\right)
\frac{\lambda_v}{(1+\lambda_v)^2}
\ge
\frac{\gamma}{(1+\lambda_{\max})^{\Delta+2}}.
\]
Combining the previous estimates gives
\[
\begin{aligned}
\Var_{\mu_{G,\vect\theta}}\!\left(\sum_{v\in V}a_v\sigma_v\right)
&\ge
\frac{\gamma}{(1+\lambda_{\max})^{\Delta+2}}
\sum_{v\in V_i}a_v^2 \\
&\ge
\frac{\gamma}{(\Delta+1)(1+\lambda_{\max})^{\Delta+2}}
\|a\|_2^2.
\end{aligned}
\]
Since this holds for every $a\in\mathbb R^V$, we conclude that
\[
\Cov_{\mu_{G,\vect\theta}}(\sigma)
\succeq
\frac{\gamma}{(\Delta+1)(1+\lambda_{\max})^{\Delta+2}}I.
\]
Finally, since
\[
\nabla^2F_{\vect m}(\vect\theta)=\Cov_{\mu_{G,\vect\theta}}(\sigma),
\]
the same bound implies that $F_{\vect m}$ is $\ell$-strongly convex on
$\mathcal D$.
\end{proof}

We now prove the corresponding smoothness bound. We use the complete
spectral independence for the hard-core model below the uniqueness threshold
established in \cite{ChenFengYinZhang2021}. Since every
$\vect\theta\in\mathcal D$ satisfies
\[
e^{\theta_v}\le \lambda_{\max}=(1-\delta)\lambdac
\qquad\text{for all }v\in V,
\]
the measure $\mu_{G,\vect\theta}$ can be viewed as the hard-core model with
uniform fugacity $\lambda_{\max}$, followed by local fields
\[
\phi_v:=\frac{e^{\theta_v}}{\lambda_{\max}}\le 1.
\]
Thus complete spectral independence applies uniformly throughout
$\mathcal D$.

\begin{lemma}\label{lem:app:gd-smoothness-hard-core}
Assume $\Delta \geq 3$. Then there exists a constant
$
L=L(\Delta,\delta)
$
such that for every $\vect\theta\in\mathcal D$, one has
\[
\Cov_{\mu_{G,\vect\theta}}(\sigma)\preceq L I.
\]
Consequently, $F_{\vect m}$ is $L$-smooth on $\mathcal D$.
\end{lemma}

\begin{proof}
Fix $\vect\theta\in\mathcal D$, and for $u\neq v$, define the signed
influence
\[
J_{u,v}
:=
\mu_{G,\vect\theta}(\sigma_v=1\mid \sigma_u=1)
-
\mu_{G,\vect\theta}(\sigma_v=1\mid \sigma_u=0),
\]
and set $J_{v,v}=0$. By complete spectral independence of the hard-core
model below the uniqueness threshold
\cite[Lemma~8.4]{ChenFengYinZhang2021}, there exists
$
C_{\mathrm{SI}}=C_{\mathrm{SI}}(\Delta,\delta)
$
such that
$
\rho_{\mathrm{sp}}(|J|)\le C_{\mathrm{SI}},
$
where $|J|$ denotes the entrywise absolute value of $J$, and
$\rho_{\mathrm{sp}}$ denotes spectral radius.

Let
\[
\mathcal V_{\vect\theta}
:=
\operatorname{diag}\bigl(
\Var_{\mu_{G,\vect\theta}}(\sigma_v):v\in V
\bigr).
\]
For $u\neq v$, we have
\[
\Cov_{\mu_{G,\vect\theta}}(\sigma_u,\sigma_v)
=
\Var_{\mu_{G,\vect\theta}}(\sigma_u)
\left[
\mu_{G,\vect\theta}(\sigma_v=1\mid \sigma_u=1)
-
\mu_{G,\vect\theta}(\sigma_v=1\mid \sigma_u=0)
\right].
\]
This means that
\[
\Cov_{\mu_{G,\vect\theta}}(\sigma)
=
\mathcal V_{\vect\theta}(I+J),
\]
where $I$ is the identity matrix. Equivalently,
\[
\mathcal V_{\vect\theta}^{-1/2}
\Cov_{\mu_{G,\vect\theta}}(\sigma)
\mathcal V_{\vect\theta}^{-1/2}
=
\mathcal V_{\vect\theta}^{1/2}
(I+J)
\mathcal V_{\vect\theta}^{-1/2}.
\]
The matrix on the left is symmetric, and it is similar to $I+J$. Hence its
eigenvalues are the same as the eigenvalues of $I+J$. Also, note that
$
\rho_{\mathrm{sp}}(J)
\le
\rho_{\mathrm{sp}}(|J|)
\le
C_{\mathrm{SI}}.
$
Therefore, the max eigen value of 
$
\mathcal V_{\vect\theta}^{-1/2}
\Cov_{\mu_{G,\vect\theta}}(\sigma)
\mathcal V_{\vect\theta}^{-1/2}
$
is at most
$
1+C_{\mathrm{SI}}.
$
Thus
\[
\mathcal V_{\vect\theta}^{-1/2}
\Cov_{\mu_{G,\vect\theta}}(\sigma)
\mathcal V_{\vect\theta}^{-1/2}
\preceq
(1+C_{\mathrm{SI}})I.
\]
Multiplying on the left and right by
$\mathcal V_{\vect\theta}^{1/2}$ gives
\[
\Cov_{\mu_{G,\vect\theta}}(\sigma)
\preceq
(1+C_{\mathrm{SI}})\mathcal V_{\vect\theta}.
\]
Finally, since each $\sigma_v$ is $\{0,1\}$-valued,
$
\Var_{\mu_{G,\vect\theta}}(\sigma_v)\le \frac14,
$
and hence
\[
\Cov_{\mu_{G,\vect\theta}}(\sigma)
\preceq
\frac{1+C_{\mathrm{SI}}}{4}I.
\]
Taking
$
L:=\frac{1+C_{\mathrm{SI}}}{4}
$
gives
$
\Cov_{\mu_{G,\vect\theta}}(\sigma)\preceq L I.
$
Since
$
\nabla^2F_{\vect m}(\vect\theta)
=
\Cov_{\mu_{G,\vect\theta}}(\sigma),
$
we conclude that $F_{\vect m}$ is $L$-smooth on $\mathcal D$.
\end{proof}

We now apply projected gradient descent on $\mathcal D$. The projection onto
$\mathcal D$ is explicit. For $z\in\mathbb R^V$,
\[
\Pi_{\mathcal D}(z)_v
=
\min\left\{
\log\lambda_{\max},
\max\{\log\gamma,z_v\}
\right\}.
\]
Thus there is no projection error in our application.

We consider the projected gradient descent iterates
\[
\vect\theta^{(t+1)}
=
\Pi_{\mathcal D}
\left(
\vect\theta^{(t)}-\frac1L\widehat g^{(t)}
\right),
\]
where $\widehat g^{(t)}$ is an approximation to
$\nabla F_{\vect m}(\vect\theta^{(t)})$. We use the following guarantee, which
is the specialization of
\cite[Algorithm~3 and Claim~D.2]{DiakonikolasKaneStewartSun2021}
to our setting.

\begin{lemma}\label{lem:app:gd-pgd-hard-core}
Let $\vect\theta_{\vect m}$ be the unique minimizer of $F_{\vect m}$ over
$\mathcal D$. Suppose that
\[
\left\|
\widehat g^{(t)}-\nabla F_{\vect m}(\vect\theta^{(t)})
\right\|_2
\le
\eta_1
\qquad
\text{for all }t.
\]
Then after
$
T
=
O\left(
\frac{L}{\ell}
\log\frac{\operatorname{diam}(\mathcal D)}{\eta}
\right)
$
iterations, one has
\[
\left\|
\vect\theta^{(T)}-\vect\theta_{\vect m}
\right\|_2
\le
\eta
+
\frac{\eta_1/L}{1-\sqrt{1-\ell/L}}.
\]
In particular, since
$
1-\sqrt{1-\ell/L}\ge \frac{\ell}{2L},
$
we have
\[
\left\|
\vect\theta^{(T)}-\vect\theta_{\vect m}
\right\|_2
\le
\eta+\frac{2\eta_1}{\ell}.
\]
\end{lemma}

We now explain how accurate the output parameter needs to be. For any
$\widehat{\vect\theta}\in\mathcal D$, we have
\[
\begin{aligned}
D_{\mathrm{KL}}
\left(
\mu_{G,\vect\theta_{\vect m}}
\,\middle\|\,
\mu_{G,\widehat{\vect\theta}}
\right)
&=
\AG(\widehat{\vect\theta})-\AG(\vect\theta_{\vect m})
-
\left\langle
\nabla\AG(\vect\theta_{\vect m}),
\widehat{\vect\theta}-\vect\theta_{\vect m}
\right\rangle .
\end{aligned}
\]
Since $\AG$ is $L$-smooth on $\mathcal D$, it follows that
\[
D_{\mathrm{KL}}
\left(
\mu_{G,\vect\theta_{\vect m}}
\,\middle\|\,
\mu_{G,\widehat{\vect\theta}}
\right)
\le
\frac{L}{2}
\left\|
\widehat{\vect\theta}-\vect\theta_{\vect m}
\right\|_2^2.
\]
Therefore, by Pinsker's inequality,
\[
\left\|
\mu_{G,\widehat{\vect\theta}}
-
\mu_{G,\vect\theta_{\vect m}}
\right\|_{\TV}
\le
\frac{\sqrt L}{2}
\left\|
\widehat{\vect\theta}-\vect\theta_{\vect m}
\right\|_2.
\]
Note that our algorithm first runs projected gradient descent with approximate
gradients to find an estimate $\widehat{\vect\theta}$. At iteration $t$, it
estimates the gradient by sampling from $\mu_{G,\vect\theta^{(t)}}$. The
algorithm then uses Glauber dynamics to generate an approximate sample from
$\mu_{G,\widehat{\vect\theta}}$. Hence there are three sources of error in the
output: the event that some gradient estimate is inaccurate, the parameter
estimation error
$
\left\|\widehat{\vect\theta}-\vect\theta_{\vect m}\right\|_2,
$
and the final sampling error from approximately sampling
$\mu_{G,\widehat{\vect\theta}}$. The following lemma, assuming that all gradient estimates are accurate,
chooses the gradient accuracy and number of iterations so that the combined
contribution of the parameter error and the final sampling error is at most
$\eps/2$.

\begin{lemma}\label{lem:app:gd-deterministic-pgd-success}
Let $\eps>0$, and set
\[
\eta:=\frac{\eps}{4\sqrt L},
\qquad
\eta_1:=\frac{\ell\eps}{8\sqrt L}.
\]
Suppose that for every $t=0,\ldots,T-1$,
$
\left\|
\widehat g^{(t)}-\nabla F_{\vect m}(\vect\theta^{(t)})
\right\|_2
\le
\eta_1,
$
where
$
T
=
O\left(
\frac{L}{\ell}
\log\frac{\operatorname{diam}(\mathcal D)}{\eta}
\right).
$
Then the output $\widehat{\vect\theta}:=\vect\theta^{(T)}$ satisfies
\[
\left\|
\widehat{\vect\theta}-\vect\theta_{\vect m}
\right\|_2
\le
\frac{\eps}{2\sqrt L}.
\]
Consequently,
\[
\left\|
\mu_{G,\widehat{\vect\theta}}-\mu_G^{\vect m}
\right\|_{\TV}
\le
\frac{\eps}{4}.
\]
If, in addition, the final sample is generated from a distribution within
total variation distance $\eps/4$ from $\mu_{G,\widehat{\vect\theta}}$, then
the final output distribution is within total variation distance $\eps/2$
from $\mu_G^{\vect m}$.
\end{lemma}

\begin{proof}
By \Cref{lem:app:gd-pgd-hard-core}, after $T$ iterations we have
\[
\left\|
\vect\theta^{(T)}-\vect\theta_{\vect m}
\right\|_2
\le
\eta+\frac{2\eta_1}{\ell}.
\]
By our choice of $\eta$ and $\eta_1$,
\[
\eta+\frac{2\eta_1}{\ell}
=
\frac{\eps}{4\sqrt L}
+
\frac{\eps}{4\sqrt L}
=
\frac{\eps}{2\sqrt L}.
\]
Thus
\[
\left\|
\widehat{\vect\theta}-\vect\theta_{\vect m}
\right\|_2
\le
\frac{\eps}{2\sqrt L}.
\]
Using the total variation bound proved above, we get
\[
\left\|
\mu_{G,\widehat{\vect\theta}}-\mu_{G,\vect\theta_{\vect m}}
\right\|_{\TV}
\le
\frac{\sqrt L}{2}
\left\|
\widehat{\vect\theta}-\vect\theta_{\vect m}
\right\|_2
\le
\frac{\eps}{4}.
\]
Since $\mu_{G,\vect\theta_{\vect m}}=\mu_G^{\vect m}$, this gives
\[
\left\|
\mu_{G,\widehat{\vect\theta}}-\mu_G^{\vect m}
\right\|_{\TV}
\le
\frac{\eps}{4}.
\]
Finally, if the last Glauber-dynamics sample is drawn from a distribution
within total variation distance $\eps/4$ from
$\mu_{G,\widehat{\vect\theta}}$, then by the triangle inequality the final
output distribution is within total variation distance $\eps/2$ from
$\mu_G^{\vect m}$.
\end{proof}

Throughout the gradient descent we need
to estimate the gradient. Recall that
\[
\nabla F_{\vect m}(\vect\theta)
=
\E_{\mu_{G,\vect\theta}}[\sigma]-\vect m.
\]
The estimation of $\E_{\mu_{G,\vect\theta}}[\sigma]$ will be done by sampling
from $\mu_{G,\vect\theta}$ multiple times and then taking the average.
Therefore, we need to make sure that we can sample from
$\mu_{G,\vect\theta}$ for every $\vect\theta\in\mathcal D$.

Since every $\vect\theta\in\mathcal D$ satisfies
\[
e^{\theta_v}
\le
\lambda_{\max}
=
(1-\delta)\lambdac
\qquad
\text{for all }v\in V,
\]
the fugacity vector from which we need to sample is uniformly below the
uniqueness threshold, with slack $\delta$. 
%By the rapid-mixing theorem for Glauber dynamics for the hard-core model below the uniqueness threshold, 
Therefore, the Glauber dynamics for $\mu_{G,\vect\theta}$ mixes to total variation distance $\zeta$ in
\[
O_{\Delta,\delta}\left(n\log\frac{n}{\zeta}\right)
\]
steps \cite{CLV23, AJKPV22, CFYZ22, CE25}. We note that, although the results in these works are presented only for uniform fugacity below the critical threshold with slack \(\delta\), they can be easily extended to non-uniform fugacity vectors as well. This extension assumes that each entry is at most \((1-\delta)\lambdac\), consistent with our case.
Therefore, for every sampling accuracy $\zeta>0$, we can sample from a distribution within total variation distance $\zeta$ from
$\mu_{G,\vect\theta}$ in
\[
O_{\Delta,\delta}\left(n\log\frac{n}{\zeta}\right)
\]
steps of Glauber dynamics.

The following lemma shows that if the number of samples that are taken to
estimate the gradient at each iteration is sufficiently large, then the event
that all the gradient estimates are accurate occurs with probability at least $ 1- \frac{\eps}{2}$. Thus, together with \Cref{lem:app:gd-deterministic-pgd-success}, it shows that our algorithm outputs an independent set whose law is within $\eps$ variation distance of the target measure.

\begin{lemma}\label{lem:app:gd-gradient-estimation-success} 
Let $\eps\in(0,1)$ and $\eta_1>0$. At each step
$t=0,\ldots,T-1$, suppose that we generate independent samples
$
\sigma^{(t,1)},\ldots,\sigma^{(t,M)}
$
such that each sample has law within total variation distance
$
\frac{\eta_1}{2\sqrt n}
$
from $\mu_{G,\vect\theta^{(t)}}$. Define
\[
\widehat g^{(t)}
:=
\frac1M\sum_{i=1}^{M}\sigma^{(t,i)}-\vect m.
\]
If
$
M
\ge
\frac{2n}{\eta_1^2}\log\frac{4nT}{\eps},
$
then
\[
\Pr\left(
\left\|
\widehat g^{(t)}-\nabla F_{\vect m}(\vect\theta^{(t)})
\right\|_2
\le
\eta_1
\text{ for all }t=0,\ldots,T-1
\right)
\ge
1-\frac{\eps}{2}.
\]
\end{lemma}

\begin{proof}
For each $t$, let
\[
\mathcal E_t
:=
\left\{
\left\|
\widehat g^{(t)}-\nabla F_{\vect m}(\vect\theta^{(t)})
\right\|_2
\le
\eta_1
\right\}.
\]
We show that
$
\Pr(\mathcal E_t^c)\le \frac{\eps}{2T}
$
for every $t$, and then take a union bound.

Fix $t$, and condition on all randomness before step $t$. Then
$\vect\theta^{(t)}$ is fixed. Write
\[
q_v^{(t)}
:=
\E_{\mu_{G,\vect\theta^{(t)}}}[\sigma_v],
\qquad v\in V.
\]
For each $v\in V$, set
\[
\overline q_v^{(t)}
:=
\frac1M\sum_{i=1}^{M}\E[\sigma_v^{(t,i)}],
\]
where the expectation is with respect to the actual law of the samples used at
step $t$. Since each sample is within total variation distance
$\eta_1/(2\sqrt n)$ from $\mu_{G,\vect\theta^{(t)}}$, we have
\[
\left|
\overline q_v^{(t)}-q_v^{(t)}
\right|
\le
\frac{\eta_1}{2\sqrt n}.
\]
Moreover, the random variables
$\sigma_v^{(t,1)},\ldots,\sigma_v^{(t,M)}$ are independent and take values in
$\{0,1\}$. Hence, by Hoeffding's inequality,
\[
\Pr\left(
\left|
\frac1M\sum_{i=1}^{M}\sigma_v^{(t,i)}
-
\overline q_v^{(t)}
\right|
>
\frac{\eta_1}{2\sqrt n}
\right)
\le
2\exp\left(-\frac{M\eta_1^2}{2n}\right).
\]
By the union bound over $v\in V$, with probability at least
\[
1-2n\exp\left(-\frac{M\eta_1^2}{2n}\right),
\]
we have, for every $v\in V$,
\[
\left|
\frac1M\sum_{i=1}^{M}\sigma_v^{(t,i)}
-
q_v^{(t)}
\right|
\le
\frac{\eta_1}{\sqrt n}.
\]
On this event,
\[
\left\|
\frac1M\sum_{i=1}^{M}\sigma^{(t,i)}
-
\E_{\mu_{G,\vect\theta^{(t)}}}[\sigma]
\right\|_2
\le
\eta_1.
\]
Since
\[
\nabla F_{\vect m}(\vect\theta^{(t)})
=
\E_{\mu_{G,\vect\theta^{(t)}}}[\sigma]-\vect m,
\]
this implies the event $\mathcal E_t$. Therefore,
\[
\Pr(\mathcal E_t^c)
\le
2n\exp\left(-\frac{M\eta_1^2}{2n}\right).
\]
By the choice of $M$,
\[
2n\exp\left(-\frac{M\eta_1^2}{2n}\right)
\le
\frac{\eps}{2T}.
\]
Finally, taking a union bound over $t=0,\ldots,T-1$ gives
\[
\Pr\left(\bigcap_{t=0}^{T-1}\mathcal E_t\right)
\ge
1-\frac{\eps}{2}.
\]
This proves the lemma.
\end{proof}

We now state the resulting algorithm.

\begin{algorithm}[ht]
\caption{Sampling from the hard-core model with prescribed marginals}
\label{alg:app:gd-sampler}
\begin{algorithmic}[1]
\Require Graph $G=(V,E)$ with maximum degree $\Delta \geq 3$, target marginals
$\vect m=(m_v)_{v\in V}$ satisfying
\eqref{eq:app:gd-marginal-assumption} with parameters
$\gamma,\delta>0$, and accuracy $\eps\in(0,1)$.
\Ensure A random independent set whose law is within total variation distance
$\eps$ from $\mu_G^{\vect m}$.

\State Set
\[
\lambda_{\max}:=(1-\delta)\lambdac.
\]

\State Set
\[
\mathcal D
:=
\left\{
\vect\theta\in\mathbb R^V:
\log\gamma
\le
\theta_v
\le
\log\lambda_{\max}
\text{ for all }v\in V
\right\}.
\]

\State Set $\ell,L$ to be the constants introduced in
\Cref{lem:app:gd-strong-convexity-hard-core} and
\Cref{lem:app:gd-smoothness-hard-core}, respectively.

\State Define
\[
\eta:=\frac{\eps}{4\sqrt L},
\qquad
\eta_1:=\frac{\ell\eps}{8\sqrt L}.
\]

\State Set
\[
T
=
O\left(
\frac{L}{\ell}
\log\frac{\operatorname{diam}(\mathcal D)}{\eta}
\right).
\]

\State Choose an integer $M$ satisfying
\[
M
\ge
\frac{2n}{\eta_1^2}
\log\frac{4nT}{\eps}.
\]

\State Choose an arbitrary initial point $\vect\theta^{(0)}\in\mathcal D$.

\For{$t=0,\ldots,T-1$}
    \State Independently generate samples
    \[
    \sigma^{(t,1)},\ldots,\sigma^{(t,M)}
    \]
    such that each sample has law within total variation distance at most
    $\eta_1/(2\sqrt n)$ from $\mu_{G,\vect\theta^{(t)}}$.

    \State Define
    \[
    \widehat g^{(t)}
    :=
    \frac1M\sum_{i=1}^{M}\sigma^{(t,i)}-\vect m.
    \]

    \State Update
    \[
    \vect\theta^{(t+1)}
    =
    \Pi_{\mathcal D}
    \left(
    \vect\theta^{(t)}-\frac1L\widehat g^{(t)}
    \right).
    \]
\EndFor

\State Run Glauber dynamics for $\mu_{G,\vect\theta^{(T)}}$ to total
variation distance $\eps/4$.

\State \Return the resulting independent set.
\end{algorithmic}
\end{algorithm}

\begin{theorem}\label{thm:app:gd-sampler}
Let $\delta, \gamma \in (0, 1)$. Let $G = (V, E)$ be a graph on $n$ vertices with maximum degree $\Delta \geq 3$, and let $\vect m$ be a marginal vector  satisfying 
\[
\gamma \leq m_v \leq (1 - \delta) \alphac \qquad \forall v \in V.
\] 
Then, for every $\eps \in (0, 1)$, 
\Cref{alg:app:gd-sampler} outputs a random independent set whose law is
within total variation distance $\eps$ from $\mu_G^{\vect m}$. Moreover, its
running time is
\[
\widetilde O_{\Delta,\delta,\gamma}
\left(
\frac{n^2}{\eps^2}
\right).
\]
\end{theorem}

\begin{proof}
Let
\[
\mathcal E
:=
\left\{
\left\|
\widehat g^{(t)}-\nabla F_{\vect m}(\vect\theta^{(t)})
\right\|_2
\le
\eta_1
\text{ for all }t=0,\ldots,T-1
\right\}.
\]
By \Cref{lem:app:gd-gradient-estimation-success}, the choice of $M$ in
\Cref{alg:app:gd-sampler} implies that
\[
\Pr(\mathcal E)\ge 1-\frac{\eps}{2}.
\]
On the event $\mathcal E$, \Cref{lem:app:gd-deterministic-pgd-success}
implies that the law of the final output is within total variation distance
$\eps/2$ from $\mu_G^{\vect m}$. Therefore, if $\nu$ denotes the law of the
output, then
\[
\|\nu-\mu_G^{\vect m}\|_{\TV}
\le
\Pr(\mathcal E^c)
+
\|\nu(\,\cdot\,\mid \mathcal E)-\mu_G^{\vect m}\|_{\TV}
\le
\frac{\eps}{2}+\frac{\eps}{2}
=
\eps.
\]

The running time bound follows from the choices of $T$ and $M$ in
\Cref{alg:app:gd-sampler}, together with the
$O_{\Delta,\delta}(n\log(n/\zeta))$ mixing bound for Glauber dynamics below
the uniqueness threshold to within $\zeta$ total variation distance. Indeed,
$T$ is logarithmic in $n$ and $1/\eps$ up to constants depending on
$\Delta,\delta,\gamma$, and
\[
M
=
\widetilde O_{\Delta,\delta,\gamma}
\left(
\frac{n}{\eps^2}
\right).
\]
Each approximate sample costs
$\widetilde O_{\Delta,\delta,\gamma}(n)$ Glauber updates. Hence the total
running time is
\[
\widetilde O_{\Delta,\delta,\gamma}
\left(
\frac{n^2}{\eps^2}
\right). \qedhere
\]
\end{proof}

\subsection{Bisection for prescribed density}
\label[appsec]{app:bisection-prescribed-density}
Here we use the bisection method to   %answer the algorithmic question concerning 
provide an algorithm for the  hard-core model with uniform fugacity:
% addressed in \Cref{sec:alg-application-mean-field}:
given a target density $\alpha^*$, approximately sample from the uniform
fugacity hard-core Gibbs measure with density $\alpha^*$. Our bisection method is an adaptation of the binary search approach used in \cite{davies2023approximatelycountingindependentsets} for sampling from the uniform distribution on fixed-size independent sets. 

Following the canonical parametrization of exponential families, we use
$\theta:=\log\lambda$ and write
\[
Z_G(\theta)
:=
\sum_{\sigma\in\OmegaG} e^{\theta|\sigma|},
\qquad
\mu_{G,\theta}(\sigma)
=
\frac{e^{\theta|\sigma|}}{Z_G(\theta)}
\qquad
\forall \sigma\in\OmegaG.
\]
Note that
\[
q(\theta):=(\log Z_G(\theta))'
=
\E_{\sigma\sim\mu_{G,\theta}}[|\sigma|],
\qquad q'(\theta)=
(\log Z_G(\theta))''
=
\Var_{\sigma\sim\mu_{G,\theta}}(|\sigma|).
\]
Since $\Var_{\sigma\sim\mu_{G,\theta}}(|\sigma|)>0$, the map
$\theta\mapsto \alpha(\mu_{G,\theta})$ is strictly increasing. Moreover, it is
a bijection from $(-\infty,\infty)$ to $(0,i_{\max}/n)$, where $i_{\max}$ is
the size of a maximum independent set in $G$. Therefore, for every
$\alpha\in(0,i_{\max}/n)$, there is a unique $\theta(\alpha)$ such that the
density of $\mu_{G,\theta(\alpha)}$ is exactly $\alpha$.

% Consider the function
% \[
% q(\theta):=\E_{\sigma\sim\mu_{G,\theta}}[|\sigma|].
% \]
% Since
% \[
% q'(\theta)
% =
% \Var_{\sigma\sim\mu_{G,\theta}}(|\sigma|)>0,
% \]
% the map $\theta\mapsto q(\theta)$ is strictly increasing. 
Hence the unique
parameter $\theta_*$ satisfying
\[
q(\theta_*)=n\alpha^*
\]
can be approximated by the bisection method. Thus, we first approximate
$\theta_*$, and then approximately sample from the Gibbs measure
$\mu_{G,\widehat\theta}$, where $\widehat\theta$ is the output of the
bisection procedure.

Throughout, we assume that
\begin{equation}
\label{eq:app:bisection-density-assumption}
\gamma
\le
\alpha^*
\le
(1-\delta)\alphac,
\end{equation}
where $\gamma,\delta>0$ are fixed constants. We need to find the range of
fugacities corresponding to densities satisfying
\eqref{eq:app:bisection-density-assumption}. For any such density $\alpha$, let
$\lambda(\alpha)$ be the corresponding fugacity. Then
\[
n\alpha
=
\sum_{v\in V}
\mu_{G,\lambda(\alpha)}(\sigma_v=1)
=
\sum_{v\in V}
\lambda(\alpha)\,
\mu_{G,\lambda(\alpha)}
\left(
\sigma_u=0 \text{ for all } u\in\neighborCl{v}
\right)
\le
n\lambda(\alpha).
\]
Therefore,
\[
\lambda(\alpha)\ge \alpha\ge \gamma.
\]
On the other hand, by the union bound,
\[
\begin{aligned}
n\alpha
&=
\lambda(\alpha)
\sum_{v\in V}
\mu_{G,\lambda(\alpha)}
\left(
\sigma_u=0 \text{ for all } u\in\neighborCl{v}
\right) \\
% &\ge
% \lambda(\alpha)
% \sum_{v\in V}
% \left(
% 1-\sum_{u\in\neighborCl{v}}
% \mu_{G,\lambda(\alpha)}(\sigma_u=1)
% \right) \\
% %&=
% %\lambda(\alpha)
% %\left(
% %n-\sum_{v\in %V}\sum_{u\in\neighborCl{v}}
% %\mu_{G,\lambda(\alpha)}(\sigma_u=1)
% %\right) \\
% &\ge
% \lambda(\alpha)
% \left(
% n-(\Delta+1)\sum_{u\in V}
% \mu_{G,\lambda(\alpha)}(\sigma_u=1)
% \right) \\
% %&=
% %\lambda(\alpha)n\bigl(1-(\Delta+1)\alpha\bigr) \\
&\ge
\lambda(\alpha)n\bigl(1-(\Delta+1)\alphac\bigr).
\end{aligned}
\]
Hence, using $\alpha\le (1-\delta)\alphac$, we get
\[
\lambda(\alpha)
\le
\frac{(1-\delta)\alphac}{1-(\Delta+1)\alphac}
=
(1-\delta)\lambdac.
\]
Therefore, defining
\[
\lambda_{\max}:=(1-\delta)\lambdac,
\]
the corresponding range for the parameter $\theta$ is the interval
\[
\mathcal I
:=
[\log\gamma,\log\lambda_{\max}].
\]

We will apply the bisection method on $\mathcal I$. In
\Cref{lem:app:bisection-strong-convexity,lem:app:bisection-smoothness}, we
show lower and upper bounds on $q'(\theta)$. The lower bound ensures that
approximate signs of $q(\theta)-n\alpha^*$ are reliable away from $\theta_*$,
while the upper bound is used to convert parameter error into total variation
error.

\begin{lemma}\label{lem:app:bisection-strong-convexity}
For every $\theta\in\mathcal I$, 
\[
\Var_{\mu_{G,\theta}}(|\sigma|)
\ge
\ell n,
\]
where
\[
\ell
:=
\frac{\gamma}{(\Delta+1)(1+\lambda_{\max})^{\Delta+2}}.
\]
% Consequently, $q'(\theta)$ is lower bounded by $\ell n$ on $\mathcal I$.
\end{lemma}

\begin{proof}
Let $\vect\theta:=(\theta)_{v\in V}$. By the definition of $\mathcal I$, we
have $\vect\theta\in\mathcal D$, where $\mathcal D$ is the box used in
\Cref{app:gradient-descent-prescribed-marginals}. Also,
$\mu_{G,\theta}=\mu_{G,\vect\theta}$. By
\Cref{lem:app:gd-strong-convexity-hard-core}, 
\[
\Cov_{\mu_{G,\vect\theta}}(\sigma)\succeq \ell I.
\]
Taking the all-ones vector $\mathbf 1\in\mathbb R^V$, 
\[
\Var_{\mu_{G,\theta}}(|\sigma|)
=
\mathbf 1^{\transpose}
\Cov_{\mu_{G,\vect\theta}}(\sigma)
\mathbf 1
\ge
\ell \|\mathbf 1\|_2^2
=
\ell n.
\]
Since $q'(\theta)=\Var_{\mu_{G,\theta}}(|\sigma|)$, this proves the claim.
\end{proof}

\begin{lemma}\label{lem:app:bisection-smoothness}
There exists a constant
\(
L=L(\Delta,\delta)
\)
such that for every $\theta\in\mathcal I$, 
\[
\Var_{\mu_{G,\theta}}(|\sigma|)
\le
Ln.
\]
% Consequently, $q'(\theta)$ is upper bounded by $Ln$ on $\mathcal I$.
\end{lemma}

\begin{proof}
Let $\vect\theta:=(\theta)_{v\in V}$. By the definition of $\mathcal I$, we
have $\vect\theta\in\mathcal D$, where $\mathcal D$ is the box used in
\Cref{app:gradient-descent-prescribed-marginals}. Also,
$\mu_{G,\theta}=\mu_{G,\vect\theta}$. By
\Cref{lem:app:gd-smoothness-hard-core}, there exists a constant
$L=L(\Delta,\delta)$ such that
\[
\Cov_{\mu_{G,\vect\theta}}(\sigma)\preceq LI.
\]
Taking the all-ones vector $\mathbf 1\in\mathbb R^V$, this gives
\[
\Var_{\mu_{G,\theta}}(|\sigma|)
=
\mathbf 1^{\transpose}
\Cov_{\mu_{G,\vect\theta}}(\sigma)
\mathbf 1
\le
L\|\mathbf 1\|_2^2
=
Ln.
\]
Since $q'(\theta)=\Var_{\mu_{G,\theta}}(|\sigma|)$, this proves the claim.
\end{proof}

The following lemma gives the number of bisection steps needed to approximate
$\theta_*$, assuming access to approximate evaluations of $q$ on
$\mathcal I$. Recall that the bisection method proceeds as follows. Start with the interval $\mathcal I$. At each
step, we query the midpoint $\theta_t$ of the current interval and estimate
$q(\theta_t)$ by $\widehat q_t$. If
\[
|\widehat q_t-n\alpha^*|
\]
is small, then we stop and output $\theta_t$. Otherwise, since $q$ is
increasing, we keep the left half if $\widehat q_t>n\alpha^*$, and keep the
right half if $\widehat q_t<n\alpha^*$. The lower bound on $q'$ ensures that
stopping in the first case already gives a good approximation to $\theta_*$,
while in the second case the sign is reliable.

\begin{lemma}\label{lem:app:bisection-bisection}
Assume that at each bisection step $t$, with midpoint $\theta_t$, we have an
estimate $\widehat q_t$ satisfying
\[
|\widehat q_t-q(\theta_t)|\le \eta_1.
\]
Consider the following bisection procedure. If
\[
|\widehat q_t-n\alpha^*|\le \eta_1,
\]
then stop and output $\widehat\theta=\theta_t$. Otherwise, keep the left or
right half according to the sign of $\widehat q_t-n\alpha^*$. If the procedure
does not stop early, then after
\[
T
=
\left\lceil
\log_2\frac{\operatorname{diam}(\mathcal I)}{\eta}
\right\rceil
\]
bisection steps, output any point $\widehat\theta$ from the final interval.
Then the output satisfies
\[
|\widehat\theta-\theta_*|
\le
\max\left\{\eta,\frac{2\eta_1}{\ell n}\right\}.
\]
\end{lemma}

\begin{proof}
Suppose first that the algorithm stops at some step $t$. Then
\[
|\widehat q_t-n\alpha^*|\le \eta_1.
\]
Since $|\widehat q_t-q(\theta_t)|\le \eta_1$, we get
\[
|q(\theta_t)-n\alpha^*|
\le
|q(\theta_t)-\widehat q_t|
+
|\widehat q_t-n\alpha^*|
\le
2\eta_1.
\]
By \Cref{lem:app:bisection-strong-convexity}, we have
$q'(\theta)\ge \ell n$ on $\mathcal I$. Since
$q(\theta_*)=n\alpha^*$, this implies
\[
\ell n|\theta_t-\theta_*|
\le
|q(\theta_t)-q(\theta_*)|
=
|q(\theta_t)-n\alpha^*|
\le
2\eta_1.
\]
Therefore,
\[
|\theta_t-\theta_*|\le \frac{2\eta_1}{\ell n}.
\]
Since $\widehat\theta=\theta_t$, this gives the desired bound in the case
where the algorithm stops early.

It remains to consider the case where the algorithm does not stop during the
first $T$ steps. Then, at every step $t<T$,
\[
|\widehat q_t-n\alpha^*|>\eta_1.
\]
Since $|\widehat q_t-q(\theta_t)|\le \eta_1$, the quantities
$\widehat q_t-n\alpha^*$ and $q(\theta_t)-n\alpha^*$ have the same sign. Hence
every bisection decision is correct, and $\theta_*$ remains inside the
maintained interval throughout the procedure.

After
\[
T
=
\left\lceil
\log_2\frac{\operatorname{diam}(\mathcal I)}{\eta}
\right\rceil
\]
steps, the final interval has length at most $\eta$. Since $\theta_*$ belongs
to this interval, any output $\widehat\theta$ from the final interval satisfies
\[
|\widehat\theta-\theta_*|\le \eta.
\]

Combining the two cases, we obtain
\[
|\widehat\theta-\theta_*|
\le
\max\left\{\eta,\frac{2\eta_1}{\ell n}\right\}.
\]
This completes the proof.
\end{proof}

For $\theta\in\mathcal I$, let
\[
A(\theta):=\log Z_G(\theta).
\]
Then
\[
D_{\mathrm{KL}}
\left(
\mu_{G,\theta_*}
\,\middle\|\,
\mu_{G,\theta}
\right)
=
A(\theta)-A(\theta_*)-A'(\theta_*)(\theta-\theta_*).
\]
By \Cref{lem:app:bisection-smoothness}, the function $A$ is $Ln$-smooth on
$\mathcal I$, and thus
\[
D_{\mathrm{KL}}
\left(
\mu_{G,\theta_*}
\,\middle\|\,
\mu_{G,\theta}
\right)
\le
\frac{Ln}{2}(\theta-\theta_*)^2.
\]
Therefore, by Pinsker's inequality,
\[
\left\|
\mu_{G,\theta_*}-\mu_{G,\theta}
\right\|_{\TV}
\le
\frac{\sqrt{Ln}}{2}|\theta-\theta_*|.
\]
Consequently, to make the error coming from parameter estimation at most
$\eps/4$ in total variation, it is enough to ensure that
\[
|\theta-\theta_*|
\le
\frac{\eps}{2\sqrt{Ln}}.
\]
If we then sample from a distribution within total variation distance
$\eps/4$ from $\mu_{G,\theta}$, the final output distribution is within
$\eps/2$ of $\mu_{G,\theta_*}$.

The following, assuming
that the estimates of the expected size $q(\theta_t)$ used in the bisection
procedure are accurate, chooses the estimation accuracy and number of
bisection steps so that the combined contribution of the parameter error and
the final sampling error is at most $\eps/2$. The remaining $\eps/2$ will
be used to account for the probability that one of the expected-size
estimates fails.

\begin{lemma}\label{lem:app:bisection-deterministic-success}
Let $\eps>0$, and set
\[
\eta:=\frac{\eps}{2\sqrt{Ln}},
\qquad
\eta_1:=\frac{\ell n\eps}{4\sqrt{Ln}}.
\]
Suppose that at every bisection step $t$, with midpoint $\theta_t$, we have
an estimate $\widehat q_t$ satisfying
\[
\left|
\widehat q_t-q(\theta_t)
\right|
\le
\eta_1.
\]
Run the bisection procedure of \Cref{lem:app:bisection-bisection} for at most
\[
T
=
\left\lceil
\log_2\frac{\operatorname{diam}(\mathcal I)}{\eta}
\right\rceil
\]
steps, and let $\widehat\theta$ be its output. Then
\[
\left|
\widehat\theta-\theta_*
\right|
\le
\frac{\eps}{2\sqrt{Ln}}.
\]
Consequently,
\[
\left\|
\mu_{G,\widehat\theta}-\mu_{G,\theta_*}
\right\|_{\TV}
\le
\frac{\eps}{4}.
\]
If, in addition, the final sample is generated from a distribution within
total variation distance $\eps/4$ from $\mu_{G,\widehat\theta}$, then the
final output distribution is within total variation distance $\eps/2$ from
$\mu_{G,\theta_*}$.
\end{lemma}

\begin{proof}
By \Cref{lem:app:bisection-bisection}, the output $\widehat\theta$ satisfies
\[
|\widehat\theta-\theta_*|
\le
\max\left\{\eta,\frac{2\eta_1}{\ell n}\right\}.
\]
With our choices of $\eta$ and $\eta_1$, we have
\[
\eta
=
\frac{\eps}{2\sqrt{Ln}},
\qquad
\frac{2\eta_1}{\ell n}
=
\frac{\eps}{2\sqrt{Ln}}.
\]
Therefore,
\[
|\widehat\theta-\theta_*|
\le
\frac{\eps}{2\sqrt{Ln}}.
\]

By the smoothness discussion above,
\[
\left\|
\mu_{G,\widehat\theta}-\mu_{G,\theta_*}
\right\|_{\TV}
\le
\frac{\sqrt{Ln}}{2}\cdot
\frac{\eps}{2\sqrt{Ln}}
=
\frac{\eps}{4}.
\]

Finally, if the distribution used to generate the final sample is within
total variation distance $\eps/4$ from $\mu_{G,\widehat\theta}$, then by the
triangle inequality its distance from $\mu_{G,\theta_*}$ is at most
$
\frac{\eps}{4}
+
\frac{\eps}{4}
=
\frac{\eps}{2}.
$
This completes the proof.
\end{proof}

We first compute an approximation $\widehat\theta\in\mathcal I$ to
$\theta_*$ using the bisection method, and then sample from the hard-core
model with fugacity $e^{\widehat\theta}$. Note that, since
\(
q(\theta)=\E_{\mu_{G,\theta}}[|\sigma|],
\)
throughout the bisection process we also need to sample from
$\mu_{G,\theta}$ in order to estimate $q(\theta)$. %$\E_{\mu_{G,\theta}}[|\sigma|]$.

Since $\theta\in\mathcal I$, %we have
\[
e^\theta\le \lambda_{\max}=(1-\delta)\lambdac.
\]
Thus the fugacity is uniformly below the uniqueness threshold, with slack
$\delta$. 
Hence, by the rapid-mixing theorem for Glauber dynamics for the hard-core model below the uniqueness threshold \cite{CLV23, AJKPV22, CFYZ22, CE25}, for every $\zeta>0$ we can
sample from a distribution within total variation distance $\zeta$ from
$\mu_{G,\theta}$ in
\[
O_{\Delta,\delta}\left(n\log\frac{n}{\zeta}\right)
\]
steps of Glauber dynamics.

The following lemma shows that by taking sufficiently many approximate samples from $\mu_{G, \theta_t}$ at each bisection step, the estimates of the expected size $q(\theta_t)$ are accurate throughout the bisection procedure with probability at least $1 - \frac{\eps}{2}$. Together with \Cref{lem:app:bisection-deterministic-success}, this proves that the algorithm outputs an independent set whose law is within total variation distance $\eps$ of the target measure.

\begin{lemma}\label{lem:app:bisection-density-estimation-success}
Let $\eps\in(0,1)$ and $\eta_1>0$. Set
\[
\zeta:=\min\left\{\frac{\eta_1}{2n},\frac{L}{n}\right\}.
\]
At each bisection step $t=0,\ldots,T-1$, suppose that we generate independent
samples
\[
\sigma^{(t,1)},\ldots,\sigma^{(t,M)}
\]
such that each sample has law within total variation distance $\zeta$ from
$\mu_{G,\theta_t}$. Write $M=BK$, and partition the samples into $K$ blocks
of size $B$. For each $j=1,\ldots,K$, define
\[
Y_j^{(t)}
:=
\frac1B\sum_{i=(j-1)B+1}^{jB}|\sigma^{(t,i)}|.
\]
Let
\[
\widehat q_t
:=
\operatorname{median}\{Y_1^{(t)},\ldots,Y_K^{(t)}\}.
\]
If
\[
B\ge \frac{32Ln}{\eta_1^2},
\qquad
K\ge 8\log\frac{2T}{\eps},
\]
then
\[
\Pr\left(
\left|
\widehat q_t-q(\theta_t)
\right|
\le
\eta_1
\text{ for all }t=0,\ldots,T-1
\right)
\ge
1-\frac{\eps}{2}.
\]
\end{lemma}

\begin{proof}
For each $t$, let
\[
\mathcal E_t
:=
\left\{
\left|
\widehat q_t-q(\theta_t)
\right|
\le
\eta_1
\right\}.
\]
We show that
\begin{equation}
\label{eq:topro}
\Pr(\mathcal E_t^c)\le \frac{\eps}{2T}    
\end{equation}
for every $t$, and then take a union bound.

Fix $t$, and condition on all randomness before step $t$. Then $\theta_t$ is
fixed. Write
\(
q_t:=
q(\theta_t).
% =
% \E_{\mu_{G,\theta_t}}[|\sigma|].
\)
% By \Cref{lem:app:bisection-smoothness},
% \(
% \Var_{\mu_{G,\theta_t}}(|\sigma|)\le Ln.
% \)

Fix a block $j$, and set
\[
\overline q_j^{(t)}
:=
\E[Y_j^{(t)}],
\]
where the expectation is with respect to the actual laws of the samples used
at step $t$. Since each sample is within total variation distance $\zeta$
from $\mu_{G,\theta_t}$, and since $0\le |\sigma|\le n$, we have
\[
\left|
\overline q_j^{(t)}-q_t
\right|
\le
n\zeta
\le
\frac{\eta_1}{2}.
\]

Moreover, if $X$ denotes one of the random variables
$|\sigma^{(t,i)}|$ in the block, then
\[
\Var(X)
\le
\E\left[(X-q_t)^2\right]
\le
\Var_{\mu_{G,\theta_t}}(|\sigma|)
+
n^2\zeta
\le
2Ln,
\]
where we used $\zeta\le L/n$ and \Cref{lem:app:bisection-smoothness}. Since the samples in the block are independent,
\[
\Var(Y_j^{(t)})\le \frac{2Ln}{B}.
\]
Therefore, by Chebyshev's inequality,
\[
\Pr\left(
\left|
Y_j^{(t)}-\overline q_j^{(t)}
\right|
>
\frac{\eta_1}{2}
\right)
\le
\frac{8Ln}{B\eta_1^2}
\le
\frac14.
\]
It follows that
\[
\Pr\left(
\left|
Y_j^{(t)}-q_t
\right|
>
\eta_1
\right)
\le
\frac14.
\]

Thus each block is good, meaning
\(
\left|
Y_j^{(t)}-q_t
\right|
\le
\eta_1,
\)
with probability at least $3/4$. Since the blocks are independent, Hoeffding's
inequality gives
\[
\Pr\left(
\left|
\widehat q_t-q_t
\right|
>
\eta_1
\right)
\le
\exp\left(-\frac{K}{8}\right).
\]
By the choice of $K$, the latter is bounded above by $\frac{\eps}{2T}$. This proves  \eqref{eq:topro}.
% \[
% \exp\left(-\frac{K}{8}\right)
% \le
% \frac{\eps}{2T}.
% \]
% Therefore,
% \[
% \Pr(\mathcal E_t^c)
% \le
% \frac{\eps}{2T}.
% \]
Finally, taking a union bound over $t=0,\ldots,T-1$ gives
\[
\Pr\left(\bigcap_{t=0}^{T-1}\mathcal E_t\right)
\ge
1-\frac{\eps}{2}.
\]
This proves the lemma.
\end{proof}

We now state the resulting algorithm.

\begin{algorithm}[ht]
\small
\caption{Bisection sampler for the hard-core model with prescribed density}
\label{alg:app:bisection-sampler}
\begin{algorithmic}[1]
\Require Graph $G=(V,E)$ with maximum degree $\Delta \geq 3$, target density
$\alpha^*$ satisfying \eqref{eq:app:bisection-density-assumption} with
parameters $\gamma,\delta>0$, and accuracy $\eps\in(0,1)$.
\Ensure A random independent set whose law is within total variation distance
$\eps$ from $\mu_{G,\theta_*}$, where $q(\theta_*)=n\alpha^*$.

\State Set $\lambda_{\max}:=(1-\delta)\lambdac$ and
$\mathcal I:=[\log\gamma,\log\lambda_{\max}]$.

\State Let $\ell,L$ be the constants from
\Cref{lem:app:bisection-strong-convexity,lem:app:bisection-smoothness}.

\State Set
\[
\eta:=\frac{\eps}{2\sqrt{Ln}},
\qquad
\eta_1:=\frac{\ell n\eps}{4\sqrt{Ln}},
\qquad
\zeta:=\min\left\{\frac{\eta_1}{2n},\frac{L}{n}\right\}.
\]

\State Set
\[
T
:=
\left\lceil
\log_2\frac{\operatorname{diam}(\mathcal I)}{\eta}
\right\rceil.
\]

\State Choose integers $B,K$ with
\[
B\ge \frac{32Ln}{\eta_1^2},
\qquad
K\ge 8\log\frac{2T}{\eps},
\]
and set $M:=BK$.

\State Initialize $a_0:=\log\gamma$ and $b_0:=\log\lambda_{\max}$.

\For{$t=0,\ldots,T-1$}
    \State Set $\theta_t:=(a_t+b_t)/2$.

    \State Generate independent samples
    $\sigma^{(t,1)},\ldots,\sigma^{(t,M)}$, each within total variation
    distance $\zeta$ from $\mu_{G,\theta_t}$.

    \State Partition the samples into $K$ blocks of size $B$, and define
    \[
    Y_j^{(t)}
    :=
    \frac1B\sum_{i=(j-1)B+1}^{jB}|\sigma^{(t,i)}|,
    \qquad j=1,\ldots,K.
    \]

    \State Set
    \[
    \widehat q_t
    :=
    \operatorname{median}\{Y_1^{(t)},\ldots,Y_K^{(t)}\}.
    \]

    \If{$|\widehat q_t-n\alpha^*|\le \eta_1$}
        \State Set $\widehat\theta:=\theta_t$ and terminate the loop.
    \ElsIf{$\widehat q_t>n\alpha^*$}
        \State Set $a_{t+1}:=a_t$ and $b_{t+1}:=\theta_t$.
    \Else
        \State Set $a_{t+1}:=\theta_t$ and $b_{t+1}:=b_t$.
    \EndIf
\EndFor

\If{the loop was not terminated early}
    \State Set $\widehat\theta:=(a_T+b_T)/2$.
\EndIf

\State Run Glauber dynamics for $\mu_{G,\widehat\theta}$ to total variation
distance $\eps/4$.

\State \Return the resulting independent set.
\end{algorithmic}
\end{algorithm}

\begin{theorem}\label{thm:app:bisection-sampler}
Let $\gamma, \delta \in (0, 1)$. Let $G = (V, E)$ be a graph on $n$ vertices with maximum degree $\Delta \geq 3$, and let
$\alpha^*$ be a density satisfying
\[
\gamma \leq \alpha^* \leq (1 - \delta)\alphac.
\] 
Then, for every $\eps \in (0, 1)$,  \Cref{alg:app:bisection-sampler} outputs a random independent set whose
law is within total variation distance $\eps$ from $\mu_{G,\theta_*}$, where $q(\theta_*)=n\alpha^*$.
Moreover, its running time is
\[
\widetilde O_{\Delta,\delta,\gamma}
\left(
\frac{n}{\eps^2}
\right).
\]
\end{theorem}

\begin{proof}
Let $\mathcal E$ be the event that every estimate actually used by the
bisection procedure satisfies
\[
|\widehat q_t-q(\theta_t)|\le \eta_1.
\]
By \Cref{lem:app:bisection-density-estimation-success}, the choice of $M$ in
\Cref{alg:app:bisection-sampler} implies that
\[
\Pr(\mathcal E)\ge 1-\frac{\eps}{2}.
\]
On the event $\mathcal E$, the deterministic bisection guarantee in
\Cref{lem:app:bisection-deterministic-success} implies that the law of the
final output is within total variation distance $\eps/2$ from
$\mu_{G,\theta_*}$. Therefore, if $\nu$ denotes the law of the output, then
\[
\|\nu-\mu_{G,\theta_*}\|_{\TV}
\le
\Pr(\mathcal E^c)
+
\|\nu(\,\cdot\,\mid \mathcal E)-\mu_{G,\theta_*}\|_{\TV}
\le
\frac{\eps}{2}+\frac{\eps}{2}
=
\eps.
\]

The running time bound follows from the choices of $T,M,B,K$ in
\Cref{alg:app:bisection-sampler}, together with the
$O_{\Delta,\delta}(n\log(n/\zeta))$ mixing bound for Glauber dynamics below
the uniqueness threshold to within $\zeta$ total variation distance. Indeed,
$T$ and $K$ are logarithmic in $n$ and $1/\eps$, up to constants depending on
$\Delta,\delta,\gamma$, and
\(
B
=
O_{\Delta,\delta,\gamma}\left(
\eps^{-2}
%\frac{1}{\eps^2}
\right).
\)
Thus
\(
M=BK
=
\widetilde O_{\Delta,\delta,\gamma}
\left(\eps^{-2}
%\frac{1}{\eps^2}
\right).
\)
Each approximate sample costs
$\widetilde O_{\Delta,\delta,\gamma}(n)$ Glauber updates. Hence the total
running time is
\(
\widetilde O_{\Delta,\delta,\gamma}
\left(
\frac{n}{\eps^2}
\right).
\)
\end{proof}

\subsection{Proof of \Cref{lem:hardness-density-to-tv-approx}}
\label[appsec]{app:hardness-missing-proofs}
% \begin{proof}[Proof of \Cref{lem:hardness-density-to-tv-approx}]
   Similar to \Cref{app:bisection-prescribed-density}, we use the log-fugacity coordinate $\theta = \log \lambda$. We write $\theta_\alpha$ for the log-fugacity parameter corresponding to $\alpha$.
   Recall from \Cref{app:bisection-prescribed-density} that
   \[
   A(\theta_\alpha) := \log Z_{G}(\theta_{\alpha}), \qquad A'(\theta_{\alpha}) = n\alpha, \qquad A''(\theta_{\alpha}) = \Var_{\sigma \sim \mu_{G}^{\alpha}}(|\sigma|).
   \]
   Also, the interval $[\alpha, \hat \alpha]$ corresponds to an interval $[\theta_{\alpha}, \theta_{\hat \alpha}]$ for the log-fugacity parameters. 
   Note that the upper bound  assumption on the variance for densities in $[\alpha, \hat \alpha]$ shows that $A$ is $Cn$-smooth on $[\theta_{\alpha}, \hat \theta_{\alpha}]$. Thus, as we proved in \Cref{app:bisection-prescribed-density}, we have
   \[
   \|\mu_{G}^{\alpha} - \mu_{G}^{\hat \alpha}\|_{\TV} \leq \frac{\sqrt{Cn}}{2}|\theta_{\alpha} - \theta_{\hat \alpha}|.
   \]
   Moreover, note that by the lower bound assumption on the variance and the mean value theorem, there exists a density $\beta \in [\alpha, \hat \alpha]$ such that
   \[
   |n\hat \alpha - n\alpha| = \Var_{\sigma \sim \mu_{G}^{\beta}}(|\sigma|) |\theta_{\alpha} - \theta_{\hat \alpha}|\geq cn|\theta_{\alpha} - \theta_{\hat \alpha}|.
   \]
   Thus, we conclude that
   \[
   \|\mu_{G}^{\alpha} - \mu_{G}^{\hat \alpha}\|_{\TV} \leq \frac{\sqrt{Cn}}{2c}|\alpha - \hat \alpha|.
   \]
%\end{proof}
\end{document}